\documentclass[11pt,oneside]{article}

\usepackage[a4paper,margin=27mm,headheight=15pt]{geometry}
\usepackage[T1]{fontenc}
\usepackage[utf8]{inputenc}
\usepackage{mathpazo}
\usepackage[scaled=0.91]{helvet}
\usepackage{microtype}
\usepackage{amsmath,amssymb,amsthm,mathtools,mathrsfs}
\usepackage{bm}
\usepackage{booktabs,array,tabularx,multirow}
\usepackage{enumitem}
\usepackage{xcolor}
\usepackage{tikz}
\usetikzlibrary{arrows.meta,calc,decorations.pathmorphing,fit,patterns,positioning,shapes.geometric,backgrounds}
\usepackage{quantikz}
\usepackage[most]{tcolorbox}
\usepackage{caption}
\usepackage{subcaption}
\usepackage{titlesec}
\usepackage{fancyhdr}
\usepackage{xurl}
\usepackage[colorlinks=true,linkcolor=DeepBlue,citecolor=DeepBlue,urlcolor=Coral]{hyperref}
\usepackage[capitalize,noabbrev]{cleveref}

\definecolor{DeepBlue}{HTML}{17324D}
\definecolor{Blue}{HTML}{285B78}
\definecolor{PaleBlue}{HTML}{EAF2F6}
\definecolor{Coral}{HTML}{D85A4F}
\definecolor{PaleCoral}{HTML}{F8E8E5}
\definecolor{SpiderGreen}{HTML}{3B8C70}
\definecolor{PaleGreen}{HTML}{E5F2ED}
\definecolor{SpiderRed}{HTML}{CF4F68}
\definecolor{Ink}{HTML}{24313B}
\definecolor{WarmGray}{HTML}{F5F2ED}
\definecolor{MidGray}{HTML}{66727B}

\hypersetup{
  pdftitle={Rethinking Quantum Circuits},
  pdfauthor={Steven Rayan},
  pdfsubject={Lecture notes from the Niels Bohr Quantum Summer School},
  pdfkeywords={quantum circuits, ZX-calculus, quantum error correction, hyperbolic geometry, superconducting circuits}
}

\color{Ink}

\setlist[itemize]{topsep=3pt,itemsep=2pt,leftmargin=1.5em}
\setlist[enumerate]{topsep=3pt,itemsep=2pt,leftmargin=1.7em}

\titleformat{\section}
  {\sffamily\bfseries\color{DeepBlue}\Large}
  {\thesection}{0.75em}{}
\titleformat{\subsection}
  {\sffamily\bfseries\color{Blue}\large}
  {\thesubsection}{0.65em}{}
\titleformat{\subsubsection}
  {\sffamily\bfseries\color{Ink}\normalsize}
  {\thesubsubsection}{0.55em}{}
\titlespacing*{\section}{0pt}{2.2em}{0.8em}
\titlespacing*{\subsection}{0pt}{1.6em}{0.5em}

\newtheoremstyle{rayanplain}
  {7pt}{7pt}{\itshape}{}{
  \sffamily\bfseries\color{DeepBlue}}{.}{0.5em}{}
\theoremstyle{rayanplain}
\newtheorem{theorem}{Theorem}[section]
\newtheorem{proposition}[theorem]{Proposition}

\theoremstyle{definition}

\newtheorem{example}[theorem]{Example}

\newtcolorbox{keyidea}[1][]{
  enhanced,breakable,colback=PaleBlue,colframe=Blue,
  boxrule=0.7pt,arc=1.2mm,left=2.2mm,right=2.2mm,top=1.5mm,bottom=1.5mm,
  fonttitle=\sffamily\bfseries,title=#1}
\newtcolorbox{warningbox}[1][]{
  enhanced,breakable,colback=PaleCoral,colframe=Coral,
  boxrule=0.7pt,arc=1.2mm,left=2.2mm,right=2.2mm,top=1.5mm,bottom=1.5mm,
  fonttitle=\sffamily\bfseries,title=#1}
\newtcolorbox{bridgebox}[1][]{
  enhanced,breakable,colback=PaleGreen,colframe=SpiderGreen,
  boxrule=0.7pt,arc=1.2mm,left=2.2mm,right=2.2mm,top=1.5mm,bottom=1.5mm,
  fonttitle=\sffamily\bfseries,title=#1}

\newcommand{\C}{\mathbb{C}}
\newcommand{\Ftwo}{\mathbb{F}_{2}}
\providecommand{\ket}[1]{\lvert #1\rangle}
\providecommand{\bra}[1]{\langle #1\rvert}
\providecommand{\braket}[2]{\langle #1\mid #2\rangle}
\providecommand{\proj}[1]{\lvert #1\rangle\!\langle #1\rvert}
\newcommand{\id}{\mathrm{id}}
\newcommand{\Foot}{\operatorname{Foot}}
\newcommand{\im}{\operatorname{im}}
\newcommand{\rank}{\operatorname{rank}}
\newcommand{\Tr}{\operatorname{Tr}}
\newcommand{\cH}{\mathcal{H}}
\newcommand{\cC}{\mathcal{C}}
\newcommand{\cE}{\mathcal{E}}
\newcommand{\cR}{\mathcal{R}}
\newcommand{\cN}{\mathcal{N}}
\newcommand{\Zsp}[3]{Z^{#1}_{#2,#3}}
\newcommand{\Xsp}[3]{X^{#1}_{#2,#3}}
\providecommand{\llbracket}{\mathopen{[\![}}
\providecommand{\rrbracket}{\mathclose{]\!]}}
\newcommand{\versiondate}{19 August 2026}

\tikzset{
  wire/.style={line width=0.85pt,draw=Ink},
  flow/.style={-{Latex[length=2.2mm]},line width=0.9pt,draw=DeepBlue},
  gate/.style={draw=DeepBlue,fill=white,rounded corners=1.2mm,minimum height=7mm,minimum width=11mm,font=\sffamily\small},
  boxnode/.style={draw=DeepBlue,fill=PaleBlue,rounded corners=1.4mm,align=center,inner sep=4pt,font=\sffamily\small},
  coralnode/.style={draw=Coral,fill=PaleCoral,rounded corners=1.4mm,align=center,inner sep=4pt,font=\sffamily\small},
  greennode/.style={draw=SpiderGreen,fill=PaleGreen,rounded corners=1.4mm,align=center,inner sep=4pt,font=\sffamily\small},
  zspider/.style={circle,draw=SpiderGreen,fill=SpiderGreen,minimum size=6mm,inner sep=0pt},
  xspider/.style={circle,draw=SpiderRed,fill=SpiderRed,minimum size=6mm,inner sep=0pt},
  port/.style={circle,draw=DeepBlue,fill=white,minimum size=4.2mm,inner sep=0pt},
  resonator/.style={circle,draw=DeepBlue,fill=PaleBlue,minimum size=3.5mm,inner sep=0pt}
}

\begin{document}

\hypersetup{pageanchor=false}
\begin{titlepage}
\thispagestyle{empty}
\begin{tikzpicture}[remember picture,overlay]
  \fill[WarmGray] (current page.south west) rectangle (current page.north east);
  \fill[DeepBlue] (current page.north west) rectangle ([yshift=-29mm]current page.north east);
  \fill[Coral] ([yshift=-29mm]current page.north west) rectangle ([yshift=-33mm]current page.north east);
\end{tikzpicture}

\vspace*{36mm}
{\sffamily\bfseries\fontsize{30}{35}\selectfont\color{DeepBlue}
Rethinking Quantum Circuits\par}
\vspace{5mm}
{\sffamily\fontsize{17}{21}\selectfont\color{Coral}
Lecture notes from the 2026 Niels Bohr Quantum Summer School\par}
\vspace{11mm}
{\Large Steven Rayan\par}
\vspace{3mm}
{\sffamily\small
Centre for Quantum Topology and Its Applications (quanTA)\par
Department of Mathematics and Statistics\par
University of Saskatchewan\par
\href{mailto:rayan@math.usask.ca}{rayan@math.usask.ca}\par}
\vspace{7mm}
\begin{center}
\begin{tikzpicture}
  \begin{scope}
    \clip[rounded corners=2.5mm] (0,0) rectangle (94mm,94mm);
    \node[anchor=south west,inner sep=0] at (0,0)
      {\includegraphics[width=94mm]{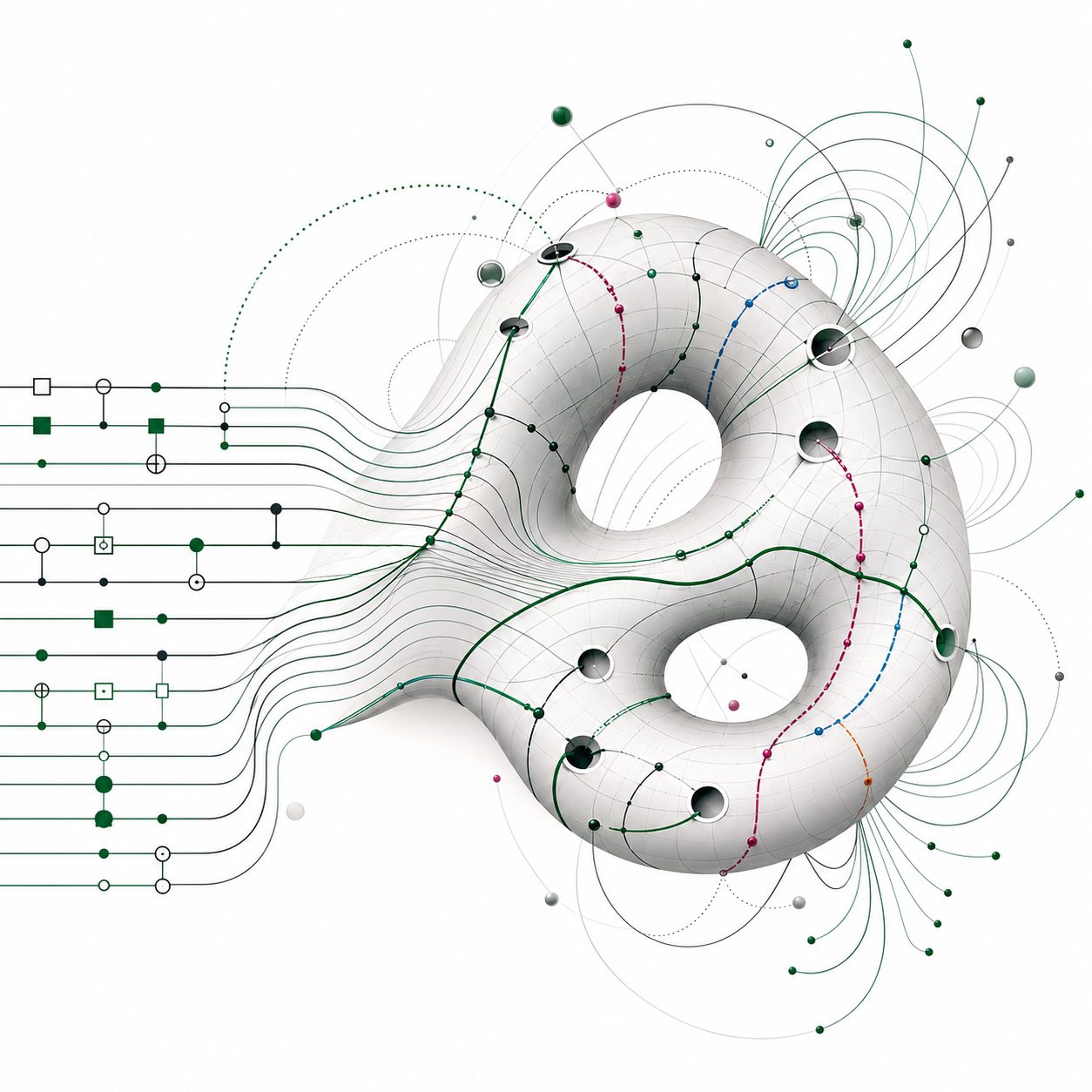}};
  \end{scope}
  \draw[DeepBlue,line width=1.2pt,rounded corners=2.5mm]
    (0,0) rectangle (94mm,94mm);
\end{tikzpicture}
\end{center}
\vfill
{\sffamily\small Syddansk Universitet, Odense, 10--13 August 2026\par}
{\sffamily\small Version Date: \versiondate\par}
\end{titlepage}

\hypersetup{pageanchor=true}
\pagenumbering{roman}
\begin{abstract}
These notes develop four interconnected ways of reading a quantum circuit.  A circuit for us begins as an operational composition of gates; then, it becomes a diagram whose local equalities may be used as calculations; next, it becomes a protected process once errors, syndromes, and logical degrees of freedom are separated; and finally, it becomes geometric when its connectivity, topology, and boundary data are treated as physical design parameters.  The development begins at the level of bits and qubits before appealing to Deutsch's and Grover's algorithms as basic examples of quantum circuits.  With the basics in hand, we interpret quantum circuits diagramatically, leading us to compact closed string diagrams and the ZX-calculus.  After that, we consider how to correct quantum circuits by introducing the Knill--Laflamme condition, homological surface codes, and related concepts with a view towards thinking of these as operations on diagrams. The lectures eventually arrive at the properties of hyperbolic quantum codes and the prospect of physical superconducting circuits emulating the negatively-curved lattices needed to support those codes. These mathematical ideas and physical experiments, taken together, represent one way to impart a geometric layer onto quantum circuits. By the very end, we bring the ideas nearly full circle by assessing the extent to which these device physics experiments operationalize the basic ZX diagrams encountered much earlier in the story.  While the later material reports on original research, and while the discussion becomes increasingly mathematical as the sections progress, no prior knowledge of quantum information, quantum computing, or quantum error correction is actually assumed.
\end{abstract}

\tableofcontents
\clearpage
\pagenumbering{arabic}

\setcounter{section}{-1}

\section{Talking Quantum Circuits}
\label{sec:talking}

\subsection{Origins of these notes}

These notes grew out of four lectures that I delivered at the Niels Bohr Quantum Summer School at the Centre for Quantum Mathematics, University of Southern Denmark, in Odense in August 2026.  The school brought together PhD students from mathematics, computer science, physics, chemistry, and neighbouring subjects for a two-week programme devoted in 2026 to quantum algorithms and quantum software.  The breadth of the audience was a welcome pedagogical challenge for me: a notation obvious to a quantum information scientist could be unfamiliar to a geometer, while a topological construction natural to a mathematician could look remote and unmotivated to an experimentalist or device physicist.

The series was called \emph{Rethinking Quantum Circuits}.  Its four lectures were \emph{Thinking Quantum Circuits}, \emph{Drawing Quantum Circuits}, \emph{Correcting Quantum Circuits}, and \emph{Geometrizing Quantum Circuits}.  Their reception, and especially the sustained questions after each lecture, suggested that the lectures should survive in an expanded and easily referenced form.  The questions and interactions after the lectures have markedly affected or even shifted the focus in some places.

The title of this introductory section is deliberate.  Talking about quantum circuits across disciplines requires more than fixing symbols.  It requires deciding what a circuit is allowed to mean.  Is it a schedule of unitary gates?  A tensor contraction?  A proof object?  A noisy spacetime history?  A graph that should be compiled into hardware?  The answer throughout these notes is that all of these readings are useful, but they must be related without conflating them.

\subsection{The central viewpoint}

The basic object is a \emph{composable process}.  A process has an input type, an output type, and a rule that transforms the former into the latter.  In elementary quantum computing the types are tensor powers
\[
  \cH_n=(\C^2)^{\otimes n},
\]
and a closed-system gate on $n$ qubits is a unitary $U:\cH_n\to\cH_n$.  But the grammar of circuits is wider than unitary dynamics.  State preparations have no quantum inputs, selected measurement branches have no quantum outputs, encoders can change the number of wires, and channels include an environment that may later be ignored.  We will therefore reserve the word \emph{unitary} for the reversible square maps and use \emph{gate} more broadly for a typed linear process when the context is diagrammatic.

The four sections apply four verbs to this common object:

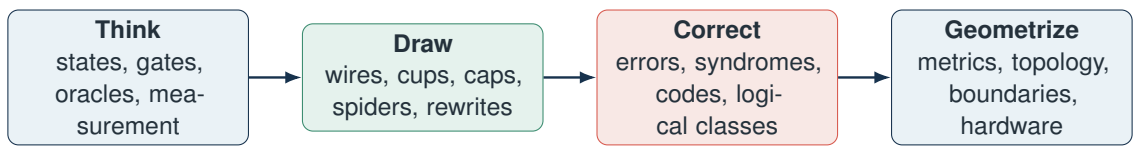
\begin{figure}[ht]
\centering
\begin{tikzpicture}[node distance=6mm and 7mm]
  \node[boxnode,text width=29mm] (think) {\textbf{Think}\\states, gates,\\oracles, measurement};
  \node[greennode,text width=29mm,right=of think] (draw) {\textbf{Draw}\\wires, cups, caps,\\spiders, rewrites};
  \node[coralnode,text width=29mm,right=of draw] (correct) {\textbf{Correct}\\errors, syndromes,\\codes, logical classes};
  \node[boxnode,text width=29mm,right=of correct] (geo) {\textbf{Geometrize}\\metrics, topology,\\boundaries, hardware};
  \draw[flow] (think) -- (draw);
  \draw[flow] (draw) -- (correct);
  \draw[flow] (correct) -- (geo);
\end{tikzpicture}
\caption{The same circuit is studied operationally, diagrammatically, fault-tolerantly, and geometrically.}
\label{fig:series-map}
\end{figure}

\begin{itemize}
\item \textbf{Thinking} asks what information the wires carry, what rules the gates implement, and how interference changes computational query complexity.
\item \textbf{Drawing} asks how composition and tensor product can be made visible, and when a local redraw is a proof of equality.
\item \textbf{Correcting} asks how logical information can survive imperfect physical processes, and how errors may be classified without reading the logical state.
\item \textbf{Geometrizing} asks what changes when adjacency, topology, metric data, and boundary access become part of the design of the process itself.
\end{itemize}

The progression is not a hierarchy in which one description replaces the last.  Matrices remain indispensable after diagrams are introduced; topology does not replace the Knill--Laflamme theorem; and a fabricated graph is not automatically a quantum code.  The aim is instead to move between descriptions at the moment each exposes the relevant structure.

\begin{keyidea}[A recurring principle]
Local components acquire meaning through composition, but the property we care about may be global.  Interference detects a global property of an oracle; a logical operator is a global class invisible to local checks; and a hyperbolic chip uses local couplings to encode global graph geometry.
\end{keyidea}

\subsection{Four changes of language}

One way to understand the series is as a sequence of translations.  Each translation preserves an underlying process while changing which questions become easy to ask.

The first translation is from a truth table to a linear operator.  A classical reversible gate permutes basis states; a quantum gate extends the allowed transformations to every unitary.  This does more than replace bits by vectors.  It introduces relative phase, tensor-product state spaces, and interference.  The operational question becomes: what property of the input can be made visible at the output after the alternatives have recombined?

The second translation is from an array of matrix entries to a typed diagram.  Composition becomes the joining of wires, tensor product becomes juxtaposition, and states and effects appear as legitimate zero-input or zero-output processes.  A local graphical relation can then explain a large matrix identity.  The price is discipline: a picture calculates only after its generators, types, scalar convention, and rewrite rules have been declared.

The third translation is from one intended circuit to a family of imperfect histories.  An encoder correlates a logical degree of freedom across a larger physical system.  Measurements extract a syndrome or, more generally, a footprint.  The central question is no longer whether a single circuit equals the identity, but whether every pair of allowed histories acts indistinguishably on the logical information.  The Knill--Laflamme condition makes that statement exact.

The fourth translation is from a code graph to a physical geometry.  Incidence relations become couplings, nontrivial cycles become logical operators, and boundary access becomes a choice of input and output ports.  This translation has several intermediate stages: a graph can be realized spectrally before it is operated in a nonlinear quantum regime, and a nonlinear quantum simulator can exist before it supports repeated syndrome extraction.  Keeping those stages separate is part of the scientific content.

\begin{center}
\begin{tabularx}{0.96\textwidth}{@{}>{\sffamily\bfseries}lXXX@{}}
\toprule
viewpoint & primitive object & equality or evidence & characteristic obstruction \\
\midrule
operational & state and gate & output statistics & destructive measurement \\
diagrammatic & typed generator & sound local rewrite & hidden scalar or type mismatch \\
fault-tolerant & encoded history & logical-state-independent comparison & an undetectable logical class \\
geometric & weighted cellulation or graph & spectrum, correlations, syndromes & confusing adjacency with quantum order \\
\bottomrule
\end{tabularx}
\end{center}

The table is also a warning against category mistakes.  A spectral match is excellent evidence that a weighted graph was compiled correctly, but it cannot certify a logical qubit.  Conversely, an abstract code may have excellent parameters while ignoring fabrication, measurement, and decoding constraints.  The full path from oracle to device requires several kinds of evidence, none of which can simply substitute for the others.

\subsection{A common mathematical spine}

Despite their different vocabularies, the four lectures repeatedly use the same small collection of mathematical operations.

\paragraph{Composition.}  If $A:U\to V$ and $B:V\to W$, their sequential composite is $BA:U\to W$.  A decoder follows a noisy encoder in exactly this sense; a microwave scattering experiment also composes drive, device response, and readout.  Type compatibility is the first correctness check in every section.

\paragraph{Tensor product.}  Independent interfaces combine as $U\otimes V$.  This operation creates the exponentially large state spaces of quantum computation, the multiple legs of a spider, the physical register of a code, and the many local modes of a lattice.  Tensor product permits correlation, but correlation is an additional property of a state or process, not a synonym for tensor product itself.

\paragraph{Adjoint and folding.}  The adjoint reverses inputs and outputs.  In diagrams it reflects a process; in error correction it allows two possible error histories to be compared by the folded operator $V^\dagger E_a^\dagger E_bV$; in spectroscopy it appears in probabilities and response functions.  Folding is therefore a general way of turning two histories into one observable comparison.

\paragraph{Quotient.}  A quotient declares some differences irrelevant.  Global phase is ignored for a pure state ray.  Stabilizers are quotiented out of the space of closed error strings.  In a compact hyperbolic surface, polygon sides are identified to form a quotient geometry.  Each use requires the equivalence relation to be stated: an informal claim that two objects are ``the same'' is incomplete without it.

\paragraph{Boundary.}  Circuit wires end at input and output boundaries; an effect is a process with only an incoming boundary.  A syndrome can be boundary data of a spacetime region.  Relative homology permits error strings to end on selected physical boundaries.  Microwave ports are literal experimental boundaries through which the device is driven and observed.  This recurrence is one reason a field-theoretic language---using boundary-sensitive composition in the spirit of Atiyah's axioms---is natural for circuit error correction \cite{atiyah,rayan-diagrammatic}.

\subsection{Pedagogical commitments}

The intended reader may know linear algebra without quantum information, or quantum algorithms without topology, or geometry without superconducting circuits.  The notes therefore follow four commitments.

First, each new formalism is motivated by a failure of the preceding one.  Reversible classical gates do not provide interference; matrices obscure local diagrammatic structure; a noiseless diagram does not explain physical reliability; and a topological code graph does not specify a device.  The next language is introduced to answer a concrete question left open by the last.

Second, elementary examples are carried far enough to expose the exact mechanism.  Deutsch's algorithm is calculated state by state.  Grover search is treated both by amplitudes and as a rotation.  The three-qubit repetition code is followed from encoding through syndrome extraction to the pairwise Knill--Laflamme test.  These examples are small, but the statements drawn from them are labelled so that a special feature is not mistaken for a general theorem.

Third, pictures and algebra are used together.  A string diagram should make composition visible; an equation should certify its scalar and basis convention.  A homological drawing should reveal a cycle; the chain complex should decide whether it is a boundary.  A device schematic should show the compiler logic; the stated vertex, edge, face, port, and genus counts should anchor it to the actual experiment.

Fourth, the notes distinguish a demonstration from a programme.  This matters most in \cref{sec:geometrizing}.  The reported superconducting devices demonstrate controlled linear spectra of finite hyperbolic graphs.  Josephson-enabled interactions, entanglement measurements, holographic correlators, and repeated quantum error correction are subsequent experimental goals.  The argument is stronger, not weaker, when the evidentiary boundary is explicit.

\subsection{How to read the notes}

Section 1 offers a compact recap of fundamental ideas of quantum information as well as very basic examples of quantum circuits, including Deutsch's and Grover's.  Readers already familiar with quantum algorithms may begin at \cref{sec:drawing}; readers interested primarily in codes may begin at \cref{sec:correcting}, using the spider formula in \cref{eq:z-spider} when it first appears; and readers approaching from experiment may begin with the evidence ladder in \cref{fig:experimental-ladder} before returning to the code geometry that motivates it.

Each lecture section contains worked calculations and a final set of problems.  The problems are not a separate sixth section: they close the relevant lecture and test its central distinctions.  Several are designed to be solved twice, once diagrammatically and once algebraically.  This is deliberate.  Agreement between representations is one of the best ways to discover a missing transpose, normalization, or boundary condition.

The notation becomes denser as the notes progress, but the conceptual questions remain stable:

\begin{enumerate}
\item What are the input and output types?
\item Which information is observable, and which must remain hidden?
\item What equivalence relation defines the desired notion of sameness?
\item Which claim is exact, which is model-dependent, and which is prospective?
\end{enumerate}

These questions are useful far beyond the examples here.  They apply to a compiled algorithm, a tensor-network proof, a decoder, a many-body simulator, or a hardware benchmark.

\subsection{What is proved, what is illustrated, and what is prospective}

Three levels of claim appear in the notes.

First, there are exact mathematical statements: the unitary oracle construction, the Deutsch calculation, the snake identities, spider fusion, the Knill--Laflamme criterion, and the homological classification of strings.  These are stated precisely and, where the proof is short enough to illuminate the idea, proved.

Second, there are controlled models.  The three-qubit repetition code corrects a specified set of bit flips but not arbitrary one-qubit noise.  A tight-binding Hamiltonian on a resonator graph tests the intended weighted adjacency but is not by itself a many-body code Hamiltonian.  The value of such a model lies in stating its promise exactly.

Third, there are research directions.  Adding nonlinear elements to hyperbolic superconducting lattices, measuring boundary correlations and entanglement, or executing repeated hyperbolic syndrome extraction goes beyond what the present linear devices have demonstrated.  These proposals will always be labelled as prospective.

\begin{warningbox}[Three distinctions to keep in view]
\begin{enumerate}
\item A \emph{normal-mode spectral gap} is not a \emph{quantum code gap}.
\item A ZX diagram can represent a map without every useful equality being easy to derive from a chosen rewrite system.
\item A locally undetectable error is logically harmless only when its closed difference is a boundary, not merely when it has no endpoint.
\end{enumerate}
\end{warningbox}

\subsection{Conventions}

Qubit 1 is the top wire in a circuit.  Circuits are read from left to right.  Composition is written algebraically in the opposite order: if $U$ is followed by $V$, the composite is $VU$.  The computational basis is $\{\ket0,\ket1\}$ and
\[
 \ket\pm=\frac{\ket0\pm\ket1}{\sqrt2}.
\]
Unless stated otherwise, homology in the coding sections uses coefficients in $\Ftwo$, so chains record only whether an edge is present or absent and orientation signs disappear.  In the ZX section we keep scalar factors visible.  This is important whenever a diagram is interpreted as a normalized state, a probability amplitude, or a physical channel.

The symbol $\doteq$ denotes equality up to a global phase.  Equalities without this symbol are exact unless stated otherwise.

\subsection{A note on the experimental figures}

Lastly, we note that the figures we have included here to describe the superconducting experiments are original schematics prepared directly for these notes.  They illustrate graph compilation, device architecture, and spectral organization without reproducing micrographs or figures that appear, or will appear, in other published articles and which therefore may be subject to copyright.  The reader is encouraged to consult nearby references to the accompanying literature in order to access the original figures arising from these experiments.

\subsection*{Acknowledgements}

I thank the summer school's programme chairs and organizers, J\o rgen Ellegaard Andersen and Jaco van de Pol, for the invitation.  I thank the Centre for Quantum Mathematics (QM) at Syddansk Universitet (SDU), the Danish Institute for Advanced Study (DIAS), and the Niels Bohr Institute for their hospitality and for providing stimulating environments for the school and its surrounding discussions.  I gratefully acknowledge the Danish e-infrastructure Consortium (DeiC), through which funding from the Danish Agency for Higher Education and Science was made available to support the summer school, itself an initiative arising from Part 2 of the Danish National Quantum Strategy.  At the same time, I am very thankful for support from the ERC Synergy Grant ReNewQuantum (PI: J\o rgen Andersen).  My participation was directly facilitated through these funding sources.  I am grateful to Louise Juel Broch and Jane Jamshidi for facilitation and support that made the event run seamlessly for all.

I am especially grateful to the PhD participants whose enthusiasm and excellent questions inspired me to expand the lectures and undertake this written account.  I also thank the collaborators (alphabetically by surname) whose contributions to our joint projects have entered the later sections: Kamal M. Ali, Sven Bachmann, Noah Gorgichuk, Ahmed A. Mahmoud, Matteo Mariantoni, Ronny Thomale, Gabrielle Tournaire, and Xicheng Xu.  The broader research programme described here has benefited from the intellectual community of the Centre for Quantum Topology and Its Applications (quanTA) at the University of Saskatchewan.

Some of the original work here was funded in part by Natural Sciences and Engineering Research Council of Canada (NSERC) Discovery and Alliance International Catalyst Quantum (G7 Stream) grants that I hold.  I am grateful to these programs for their support.

\clearpage
\section{Thinking Quantum Circuits}
\label{sec:thinking}

\subsection{Classical circuits are compositions of functions}

A classical bit takes one of the values $0$ and $1$.  A register of $n$ bits takes one value in $\{0,1\}^n$.  A classical gate is therefore a function
\[
 f:\{0,1\}^n\longrightarrow\{0,1\}^m.
\]
The familiar gates NOT, AND, OR, XOR, and NAND are distinguished by their truth tables.  Wiring one gate into another is ordinary function composition.  Placing gates on disjoint registers forms their Cartesian product and permits parallel evaluation.

Most familiar Boolean gates lose information.  AND maps four inputs to two possible outputs, so no inverse can reconstruct the input.  A reversible $n$-bit gate, by contrast, is a permutation of $\{0,1\}^n$.  For one bit there are only two: the identity and NOT.

An irreversible function can nevertheless be embedded into a reversible one.  If $f:\{0,1\}^n\to\{0,1\}^m$, define
\begin{equation}
  U_f:(x,y)\longmapsto (x,y\oplus f(x)),
  \qquad x\in\{0,1\}^n,\;y\in\{0,1\}^m,
  \label{eq:reversible-embedding}
\end{equation}
where $\oplus$ is bitwise XOR.  The input $x$ is retained and the target register $y$ stores the function value reversibly.  Applying $U_f$ twice gives the identity.  On computational basis states, the same permutation defines a unitary matrix.

\begin{figure}[ht]
\centering
\begin{tikzpicture}[node distance=12mm and 17mm]
  \node[boxnode,text width=28mm] (bool) {Boolean function\\$f:x\mapsto f(x)$};
  \node[coralnode,text width=33mm,right=of bool] (rev) {reversible embedding\\$(x,y)\mapsto(x,y\oplus f(x))$};
  \node[greennode,text width=34mm,right=of rev] (uni) {quantum oracle\\$\ket{x,y}\mapsto\ket{x,y\oplus f(x)}$};
  \draw[flow] (bool) -- node[above,font=\sffamily\scriptsize]{retain input} (rev);
  \draw[flow] (rev) -- node[above,font=\sffamily\scriptsize]{extend linearly} (uni);
\end{tikzpicture}
\caption{An irreversible truth table can be placed inside a reversible, and hence unitary, oracle.}
\label{fig:oracle-embedding}
\end{figure}
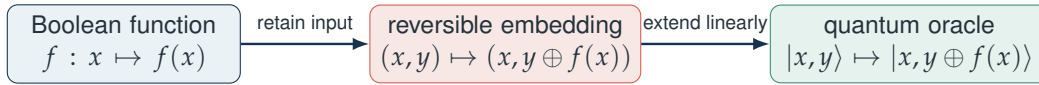

\subsubsection{Truth tables, information loss, and workspace}

For two input bits the common Boolean gates are summarized by
\begin{center}
\begin{tabular}{@{}cc|cccc@{}}
\toprule
$x$ & $y$ & $x\wedge y$ & $x\vee y$ & $x\oplus y$ & $\operatorname{NAND}(x,y)$ \\
\midrule
0&0&0&0&0&1\\
0&1&0&1&1&1\\
1&0&0&1&1&1\\
1&1&1&1&0&0\\
\bottomrule
\end{tabular}
\end{center}
The table displays two different notions that are sometimes conflated.  A gate may be deterministic without being reversible.  AND is completely deterministic, but the output $0$ has three preimages.  A reversible gate must be bijective on the complete register on which it acts.

Retaining the input in \cref{eq:reversible-embedding} supplies enough workspace to make the action injective.  For the one-bit NOT function $f(x)=1\oplus x$, the two-register oracle flips the target when $x=0$: it is a NOT on the target controlled on $x=0$, often called a negative-control or anti-controlled NOT:
\[
 U_{\mathrm{NOT}}\ket{x,y}=\ket{x,y\oplus(1\oplus x)}.
\]
For the identity function $f(x)=x$, it is the usual CNOT.  For a two-bit AND function, the reversible oracle is the three-bit Toffoli action
\[
 \ket{x_1,x_2,y}\longmapsto
 \ket{x_1,x_2,y\oplus x_1x_2}.
\]
The target bit is workspace, but it is not disposable garbage if it remains correlated with data needed later.  Reversible circuit design therefore includes an \emph{uncomputation} discipline: copy out a classical answer when permitted, run the temporary computation backwards, and return workspace to a known state.

On the ordered computational basis $\ket{00},\ket{01},\ket{10},\ket{11}$, CNOT has matrix
\begin{equation}
 \operatorname{CNOT}=
 \begin{pmatrix}
 1&0&0&0\\
 0&1&0&0\\
 0&0&0&1\\
 0&0&1&0
 \end{pmatrix}.
 \label{eq:cnot-matrix}
\end{equation}
Every column contains one $1$ because each basis input has one output; every row contains one $1$ because the action is bijective.  Thus reversible Boolean gates give permutation matrices, which are orthogonal and hence unitary.  Quantum computation keeps all of these gates and adds unitary matrices whose columns are coherent superpositions.

\begin{example}[A reversible evaluation followed by uncomputation]
Suppose a reversible circuit $U_f$ computes $\ket{x,0}\mapsto\ket{x,f(x)}$ and the value $f(x)$ is known to be a computational-basis label.  Introduce a clean target and apply CNOTs to copy that label:
\[
 \ket{x,0,0}\xrightarrow{U_f}
 \ket{x,f(x),0}\longrightarrow
 \ket{x,f(x),f(x)}\xrightarrow{U_f^\dagger}
 \ket{x,0,f(x)}.
\]
The temporary register has been reset.  If the middle register were in an arbitrary unknown superposition unrelated to a basis computation, the copying step would instead create entanglement; it would not violate the no-cloning theorem.
\end{example}

\subsection{Qubits and the enlarged gate space}

A pure qubit state is a unit vector in $\C^2$,
\[
  \ket\psi=\alpha\ket0+\beta\ket1,
  \qquad |\alpha|^2+|\beta|^2=1.
\]
The coefficients are amplitudes, not probabilities.  They may be negative or complex and can interfere before their squared magnitudes are observed.  A measurement in the computational basis returns $0$ with probability $|\alpha|^2$ and $1$ with probability $|\beta|^2$.  One run returns one classical outcome; it does not reveal $\alpha$ and $\beta$.

An $n$-qubit pure state lies in
\[
  \cH_n=(\C^2)^{\otimes n},\qquad \dim \cH_n=2^n.
\]
A closed-system quantum gate is a unitary $U:\cH_n\to\cH_n$.  Reversible classical gates occupy the finite subgroup of permutation matrices inside the continuous group $U(2^n)$.  Quantum logic therefore enlarges the gate space from permutations to all orthonormal changes of basis, including continuous rotations and relative phases.

The Pauli and Hadamard matrices are
\[
 X=\begin{pmatrix}0&1\\1&0\end{pmatrix},\qquad
 Z=\begin{pmatrix}1&0\\0&-1\end{pmatrix},\qquad
 H=\frac{1}{\sqrt2}\begin{pmatrix}1&1\\1&-1\end{pmatrix}.
\]
The Hadamard gate maps $\ket0\mapsto\ket+$ and $\ket1\mapsto\ket-$, and $H^2=I$.  It creates coherent alternatives and later recombines them.  The conjugation identity
\[
  HZH=X
\]
is the elementary form of the statement that a phase difference can be turned into a bit difference.

Circuit grammar has two products.  Sequential composition gives $VU$ when $U$ is followed by $V$.  Parallel composition gives $U\otimes V$.  These are not merely two matrix operations; they are the two ways that boxes can be joined in a typed circuit.

\subsubsection{Global phase, relative phase, and the Bloch sphere}

Multiplying a state vector by a common phase does not change any measurement probability:
\[
 \ket\psi\sim e^{i\gamma}\ket\psi.
\]
Pure physical states are therefore rays rather than individual unit vectors.  After removing a global phase, every one-qubit pure state can be written
\begin{equation}
 \ket{\psi(\vartheta,\varphi)}
 =\cos\frac{\vartheta}{2}\ket0
 +e^{i\varphi}\sin\frac{\vartheta}{2}\ket1,
 \qquad 0\leq\vartheta\leq\pi.
 \label{eq:bloch-state}
\end{equation}
The corresponding Bloch vector is
\[
 \bm r=(\sin\vartheta\cos\varphi,
        \sin\vartheta\sin\varphi,
        \cos\vartheta),
\]
and the density matrix is
\begin{equation}
 \rho=\proj\psi=\frac12(I+r_xX+r_yY+r_zZ),
 \qquad
 Y=\begin{pmatrix}0&-i\\ i&0\end{pmatrix}.
 \label{eq:bloch-density}
\end{equation}
The phase $\varphi$ is relative: it changes the $X$- and $Y$-basis statistics even though the computational-basis probabilities depend only on $\vartheta$.  This is why a phase gate can be invisible to one measurement and decisive after a basis change.

\begin{figure}[ht]
\centering
\begin{tikzpicture}[scale=1.0]
  \draw[DeepBlue,line width=1pt] (0,0) circle (19mm);
  \draw[MidGray,dashed] (-19mm,0) arc[start angle=180,end angle=360,x radius=19mm,y radius=6mm];
  \draw[MidGray] (19mm,0) arc[start angle=0,end angle=180,x radius=19mm,y radius=6mm];
  \draw[flow] (0,-24mm)--(0,24mm) node[above,font=\sffamily\small] {$z$};
  \draw[flow] (-24mm,0)--(24mm,0) node[right,font=\sffamily\small] {$x$};
  \draw[flow] (-14mm,-8mm)--(14mm,8mm) node[above right,font=\sffamily\small] {$y$};
  \draw[Coral,-{Latex[length=2.4mm]},line width=1.5pt] (0,0)--(10mm,9mm) coordinate (r);
  \draw[Coral,dashed] (r)--(10mm,0);
  \node[Coral,font=\sffamily\small,above right] at (r) {$\bm r$};
  \node[font=\sffamily\small] at (0,22mm) {$\ket0$};
  \node[font=\sffamily\small] at (0,-22mm) {$\ket1$};
  \node[font=\sffamily\small] at (22mm,3mm) {$\ket+$};
  \begin{scope}[xshift=72mm]
    \node[boxnode,text width=58mm] (m) at (0,0) {measurement along unit vector $\bm n$\\[1mm]
      $p(\pm)=\frac12(1\pm\bm r\!\cdot\!\bm n)$};
    \node[below=7mm of m,font=\sffamily\small,align=center] {a basis choice selects an axis;\\one shot returns one endpoint};
  \end{scope}
\end{tikzpicture}
\caption{A pure qubit is a point on the Bloch sphere.  Relative phase is an observable direction, not an extra classical probability.}
\label{fig:bloch-sphere}
\end{figure}

The most common one-qubit rotations are
\[
 R_x(\theta)=e^{-i\theta X/2},\qquad
 R_y(\theta)=e^{-i\theta Y/2},\qquad
 R_z(\theta)=e^{-i\theta Z/2}.
\]
Up to global phase, every one-qubit unitary has an Euler decomposition
\begin{equation}
 U\doteq R_z(\alpha)R_x(\beta)R_z(\gamma).
 \label{eq:euler-one-qubit}
\end{equation}
This continuous three-parameter family has no classical one-bit analogue: a reversible classical bit has only the identity and NOT.

\subsubsection{A working gate dictionary}

The following gates recur throughout the notes.  Equalities are exact unless the symbol $\doteq$ indicates equality up to global phase.
\begin{center}
\begin{tabularx}{0.96\textwidth}{@{}>{\sffamily\bfseries}l>{\centering\arraybackslash}p{33mm}X@{}}
\toprule
gate & matrix or action & operational role \\
\midrule
$X$ & $\ket0\leftrightarrow\ket1$ & reversible classical NOT; bit flip \\
$Z$ & $\operatorname{diag}(1,-1)$ & relative phase; phase flip \\
$S$ & $\operatorname{diag}(1,i)$ & quarter turn about the Bloch $z$ axis \\
$T$ & $\operatorname{diag}(1,e^{i\pi/4})$ & non-Clifford eighth turn used in fault-tolerant synthesis \\
$H$ & $(X+Z)/\sqrt2$ & exchanges $X$ and $Z$ bases; creates and recombines alternatives \\
CNOT & $\ket{a,b}\mapsto\ket{a,a\oplus b}$ & coherent parity addition; entangling gate \\
CZ & $\ket{a,b}\mapsto(-1)^{ab}\ket{a,b}$ & symmetric controlled phase; graph-state edge \\
SWAP & $\ket{a,b}\mapsto\ket{b,a}$ & exchanges wire types or physical locations \\
\bottomrule
\end{tabularx}
\end{center}

The controlled gates are related by basis change:
\begin{equation}
 \operatorname{CNOT}=(I\otimes H)\operatorname{CZ}(I\otimes H).
 \label{eq:cnot-cz-basis}
\end{equation}
This identity is both an algebraic compilation rule and a preview of colour change in the ZX-calculus.  It says that adding the control bit into the target in the computational basis is the same as writing a conditional phase and reading the target in the complementary basis.

The set $\{H,S,\operatorname{CNOT}\}$ generates the Clifford group, which maps Pauli operators to Pauli operators under conjugation.  Adding $T$ gives a discrete approximately universal set.  Arbitrary rotations are then approximated to a requested precision; the exact real-angle gates used in an abstract algorithm and their discrete fault-tolerant implementations are different compilation levels.

\subsubsection{Measurements are typed processes}

A projective measurement in an orthonormal basis $\{\ket{b_j}\}$ has outcome probabilities
\[
 p(j)=\bra{b_j}\rho\ket{b_j}.
\]
If the outcome is retained classically, the process is not a unitary map from one qubit to one qubit.  It is a channel from a quantum system to a classical register, possibly together with a conditional post-measurement state.  If a particular outcome is selected, its unnormalized branch is the effect $\bra{b_j}$ acting on the state.  This broader typing will become visible in \cref{sec:drawing} when effects are drawn with no outgoing quantum wire.

Repeated preparation and measurement estimates a distribution.  It does not expose the amplitudes of a single unknown system.  To reconstruct a generic qubit state, one must measure complementary observables on an ensemble of identically prepared systems; for example, the expectation values of $X,Y,Z$ determine the Bloch vector in \cref{eq:bloch-density}.

\subsubsection{Mixed states and purification}

A density operator $\rho$ is positive semidefinite and has unit trace.  It represents either classical uncertainty over preparations or the local state of a subsystem entangled with something else.  Those two origins cannot in general be distinguished from measurements on the subsystem alone.  The purity satisfies
\[
 \frac1d\leq\Tr(\rho^2)\leq1,
\]
with equality $1$ precisely for a pure state.

Every mixed state can be purified.  If
\[
 \rho=\sum_j\lambda_j\proj j,
\]
then
\begin{equation}
 \ket\Psi_{SE}=\sum_j\sqrt{\lambda_j}\ket j_S\ket j_E
 \label{eq:purification}
\end{equation}
satisfies $\Tr_E\proj\Psi=\rho$.  Any two minimal purifications are related by a unitary on the purifying system.  This is the state-level analogue of Stinespring dilation: apparent non-unitarity of a subsystem can be represented as unitary evolution on a larger space followed by forgetting an environment.

For the Bell state, \cref{eq:purification} purifies $I/2$.  The subsystem is maximally mixed not because a definite but unknown pure state has been selected, but because the missing quantum information is stored in correlations.  This distinction becomes operational when access to the partner system restores interference.

A general measurement with outcomes $j$ is described by operators $M_j$ with $\sum_jM_j^\dagger M_j=I$.  The probability and conditional state are
\[
 p(j)=\Tr(M_j\rho M_j^\dagger),
 \qquad
 \rho_j=\frac{M_j\rho M_j^\dagger}{p(j)}.
\]
Ignoring the outcome gives the channel $\rho\mapsto\sum_jM_j\rho M_j^\dagger$.  The distinction between selecting and ignoring a branch will reappear graphically as effect versus discard.

\subsection{Entanglement and the failure of naive copying}

The controlled-NOT gate acts on basis states by
\[
 \operatorname{CNOT}\ket{a,b}=\ket{a,a\oplus b}.
\]
It does not measure a superposed control.  By linearity,
\[
 \ket{00}\xrightarrow{H\otimes I}
 \frac{\ket{00}+\ket{10}}{\sqrt2}
 \xrightarrow{\operatorname{CNOT}}
 \frac{\ket{00}+\ket{11}}{\sqrt2}=\ket{\Phi^+}.
\]
The Bell state $\ket{\Phi^+}$ cannot be written as a product of two one-qubit states.  Its two subsystems are not independent carriers of copied amplitudes; they are parts of one correlated state.

This distinction matters later for error correction.  There is no unitary that clones every unknown state,
\[
  \ket\psi\ket0\longmapsto\ket\psi\ket\psi
  \quad\text{for all }\ket\psi,
\]
because preservation of inner products would require
$\braket{\phi}{\psi}=\braket{\phi}{\psi}^2$ for arbitrary pairs.  Encoding instead maps a chosen orthogonal basis to correlated codewords and extends linearly.  It can copy \emph{classical labels in one basis} without copying an arbitrary quantum state \cite{wootters-zurek}.

\subsubsection{How to recognize two-qubit entanglement}

A general two-qubit pure state has the form
\begin{equation}
 \ket\Psi=a\ket{00}+b\ket{01}+c\ket{10}+d\ket{11}.
 \label{eq:two-qubit-state}
\end{equation}
Arrange the amplitudes into the coefficient matrix
\[
 M_\Psi=\begin{pmatrix}a&b\\c&d\end{pmatrix}.
\]
The state is a product precisely when $M_\Psi$ has rank one, equivalently
\begin{equation}
 ad-bc=0.
 \label{eq:two-qubit-separable}
\end{equation}
For $\ket{\Phi^+}$ the coefficient matrix is $I/\sqrt2$, so its determinant is $1/2$ and the state is entangled.

The Schmidt decomposition gives a basis-independent version.  There exist orthonormal one-qubit bases such that
\[
 \ket\Psi=\sqrt{\lambda_0}\ket{u_0}\ket{v_0}
          +\sqrt{\lambda_1}\ket{u_1}\ket{v_1},
 \qquad \lambda_0+\lambda_1=1.
\]
The state is a product if and only if one Schmidt coefficient vanishes.  It is maximally entangled when $\lambda_0=\lambda_1=1/2$.

For the Bell state, the reduced state of either qubit is maximally mixed:
\[
 \Tr_2\proj{\Phi^+}=\Tr_1\proj{\Phi^+}=\frac I2.
\]
There is no pure local state hidden behind the perfect correlations.  Measuring both qubits in the $Z$ basis produces equal random bits; measuring both in the $X$ basis also produces equal random bits.  The correlations persist across complementary bases even though each subsystem alone is completely featureless.

\begin{example}[A phase changes the correlation basis]
The state $\ket{\Phi^-}=(\ket{00}-\ket{11})/\sqrt2$ has the same computational-basis probabilities as $\ket{\Phi^+}$.  Applying $H\otimes H$ gives
\[
 (H\otimes H)\ket{\Phi^+}=\ket{\Phi^+},\qquad
 (H\otimes H)\ket{\Phi^-}=\frac{\ket{01}+\ket{10}}{\sqrt2}.
\]
Thus the relative sign, invisible in the first basis, becomes agreement versus disagreement in the second.  Error correction will exploit the same conversion when Hadamards turn phase flips into bit flips.
\end{example}

\subsubsection{A second proof of no-cloning}

Linearity makes the contradiction concrete.  Suppose a unitary copier correctly acts on the computational basis:
\[
 U\ket0\ket0=\ket0\ket0,
 \qquad
 U\ket1\ket0=\ket1\ket1.
\]
Then on $\ket+= (\ket0+\ket1)/\sqrt2$ it must give
\[
 U\ket+\ket0=\frac{\ket{00}+\ket{11}}{\sqrt2},
\]
whereas perfect cloning would require
\[
 \ket+\ket+=\frac{\ket{00}+\ket{01}+\ket{10}+\ket{11}}2.
\]
The first output is a Bell state and the second is a product state.  A basis-copying map is therefore possible and useful, but its linear extension creates correlations rather than universal copies.

\subsection{Interference as a computational resource}

Consider the sequence $HZH$ on $\ket0$:
\[
 \ket0\xrightarrow{H}\frac{\ket0+\ket1}{\sqrt2}
 \xrightarrow{Z}\frac{\ket0-\ket1}{\sqrt2}
 \xrightarrow{H}\ket1.
\]
Immediately after $Z$, a computational-basis measurement would still be uniform.  The second Hadamard converts a relative sign into a deterministic classical bit.  This three-step pattern---prepare coherent alternatives, write information into phase, recombine by interference---drives the two algorithms below.

Interference is sometimes described as if all amplitudes were simply added together.  The more precise statement is basis dependent.  If $U$ maps input basis states to amplitudes $U_{yx}$, then the amplitude of output $y$ is
\[
 \braket{y}{U\psi}=\sum_x U_{yx}\braket{x}{\psi}.
\]
Different computational paths contribute complex numbers to the \emph{same} output basis state.  Their phases can make the sum larger or smaller.  Alternatives that remain recorded in distinguishable environment or workspace states cannot interfere completely because they do not arrive at the same final state.

This observation explains the need for uncomputation.  Suppose two alternatives acquire the desired phases but leave different garbage states $\ket{g_0}$ and $\ket{g_1}$:
\[
 \frac{\ket0\ket{g_0}+e^{i\phi}\ket1\ket{g_1}}{\sqrt2}.
\]
If $\braket{g_0}{g_1}=0$, measuring only the first qubit shows no phase interference.  Resetting both garbage states to a common reference erases the which-path record and restores coherence.

\begin{keyidea}[The operational pattern]
A quantum algorithm does not gain useful output merely by producing a large superposition.  It must (i) prepare alternatives, (ii) correlate a desired global property with relative phase or amplitude, (iii) remove irrelevant which-path information, and (iv) recombine in a basis where the property becomes a classical outcome.
\end{keyidea}

\subsection{Oracles and what a query counts}

An oracle is a typed interface, not necessarily a literal memory device.  For a Boolean function the reversible form
\[
 U_f\ket{x,y}=\ket{x,y\oplus f(x)}
\]
specifies how an algorithm may access $f$ coherently.  A query-complexity result counts uses of this interface while allowing other gates at no cost.  It isolates information access from the engineering cost of realizing the oracle.

If the target is prepared in $\ket-$, phase kickback turns the bit oracle into a phase oracle on the control:
\begin{equation}
 O_f\ket x=(-1)^{f(x)}\ket x.
 \label{eq:boolean-phase-oracle}
\end{equation}
Conversely, with standard additional controls one can often convert between bit and phase access, but the exact cost and assumptions should be stated.  In more general algorithms an oracle may mark a subspace, perform modular arithmetic, load data, or implement a block encoding of a matrix.  Two papers using the word ``query'' may therefore be counting access to different primitive operations.

Global phase remains unobservable.  Replacing $f$ by $1\oplus f$ multiplies $O_f$ by $-1$, so a phase oracle alone cannot distinguish a Boolean function from its complement.  Deutsch's promise asks only whether the two values agree, which is invariant under complementation.  The oracle type and the requested output are matched.

Controlled access is another assumption.  A formal unitary $U$ does not automatically grant a controlled-$U$ at the same query cost if $U$ is an unknown black box with an unspecified global phase.  Standard query models declare controlled access when it is needed.  Careful algorithm statements separate what follows from unitarity from what is part of the oracle contract.

\subsection{Deutsch's problem}

Let $f:\{0,1\}\to\{0,1\}$ be promised either constant or balanced.  Classically, one query cannot determine the promise with certainty: the observed value at one input is compatible with one constant and one balanced truth table.  Two queries suffice.

The quantum oracle is the reversible map from \cref{eq:reversible-embedding},
\[
 U_f\ket{x,y}=\ket{x,y\oplus f(x)}.
\]
Prepare $\ket0\ket1$, apply a Hadamard to each qubit, call $U_f$ once, apply $H$ to the first qubit, and measure it.

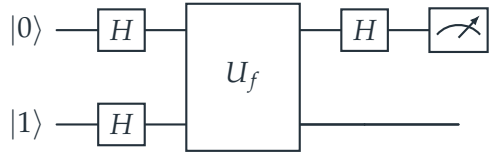
\begin{figure}[ht]
\centering
\begin{quantikz}[row sep=0.55cm,column sep=0.55cm]
\lstick{$\ket0$} & \gate{H} & \gate[wires=2][1.5cm]{U_f} & \gate{H} & \meter{} \\
\lstick{$\ket1$} & \gate{H} &                       & \qw      & \qw
\end{quantikz}
\caption{Deutsch's circuit makes the parity $f(0)\oplus f(1)$ observable with one oracle call.}
\label{fig:deutsch-circuit}
\end{figure}

The second register becomes $\ket-$.  Since $X\ket-=-\ket-$,
\begin{equation}
 U_f\ket{x}\ket-=(-1)^{f(x)}\ket{x}\ket-.
 \label{eq:phase-kickback}
\end{equation}
The oracle therefore leaves the ancilla factored and transforms the first qubit to
\[
 \frac{(-1)^{f(0)}\ket0+(-1)^{f(1)}\ket1}{\sqrt2}.
\]
If $f$ is constant, this is $\ket+$ up to a global sign; if $f$ is balanced, it is $\ket-$ up to a global sign.  The final Hadamard returns $\ket0$ or $\ket1$ respectively.  Thus the measurement is exactly
\[
 f(0)\oplus f(1).
\]

\begin{proposition}[Deutsch's query separation]
Under the constant-or-balanced promise, one quantum query determines the class of $f$ with certainty, whereas every exact classical algorithm requires two queries.
\end{proposition}

\subsubsection{The four functions and the full calculation}

There are exactly four one-bit Boolean functions:
\begin{center}
\begin{tabular}{@{}cclc@{}}
\toprule
$f(0)$ & $f(1)$ & function & promise class \\
\midrule
0&0&$f(x)=0$&constant\\
1&1&$f(x)=1$&constant\\
0&1&$f(x)=x$&balanced\\
1&0&$f(x)=1\oplus x$&balanced\\
\bottomrule
\end{tabular}
\end{center}
The classical lower bound follows immediately: after querying only $x=0$, either observed value is compatible with one row from each promise class.

For the quantum circuit, write each stage explicitly.  The prepared state is
\[
 \ket{\psi_0}=\ket0\ket1.
\]
After the first pair of Hadamards,
\begin{align*}
 \ket{\psi_1}
 &=\ket+\ket-\\
 &=\frac12(\ket{00}-\ket{01}+\ket{10}-\ket{11}).
\end{align*}
For a fixed value of $x$, the target superposition obeys
\[
 \frac{\ket{f(x)}-\ket{1\oplus f(x)}}{\sqrt2}
 =(-1)^{f(x)}\ket-.
\]
Hence one oracle call produces
\[
 \ket{\psi_2}
 =\frac{(-1)^{f(0)}\ket0+(-1)^{f(1)}\ket1}{\sqrt2}\otimes\ket-.
\]
The target is separable and can now be ignored.  Applying the final Hadamard gives the first-register amplitudes
\begin{align*}
 \braket0{\psi_{\mathrm{out}}}
   &=\frac{(-1)^{f(0)}+(-1)^{f(1)}}2,\\
 \braket1{\psi_{\mathrm{out}}}
   &=\frac{(-1)^{f(0)}-(-1)^{f(1)}}2.
\end{align*}
Exactly one of these amplitudes is nonzero.  Equal truth-table values give outcome $0$; unequal values give outcome $1$.

The target preparation $\ket-$ is essential.  If it were prepared in $\ket+$, then $X\ket+=\ket+$ and the oracle would write no phase at all.  If it were prepared in a computational-basis state and then discarded, the first register would generally become entangled with the target and the deterministic interference would be lost.

\subsubsection{What the algorithm does and does not compute}

The output is the parity $f(0)\oplus f(1)$, not an ordered pair containing both values.  This is an early example of a common quantum-algorithm design principle: a promise problem may ask for less information than a full input reconstruction.  Quantum interference can be arranged to retain exactly that quotient information.

The speedup is also a query statement.  It compares the number of uses of a black-box oracle when all other gates are counted as free.  For such a small problem the circuit is not a claim of practical runtime advantage.  Its value is conceptual: a relative phase created by one coherent oracle call can encode a global property that one exact classical query cannot determine.

The advantage is not well described as ``evaluating both inputs at once.''  The circuit does not return both values.  It engineers interference so that one global relation between them is the only surviving output \cite{deutsch,nielsen-chuang}.

\subsection{Grover search and amplitude amplification}

Let one item $w$ be marked among $N$ unstructured possibilities.  The phase oracle acts by
\[
 O_w\ket{x}=(-1)^{[x=w]}\ket{x}.
\]
Begin in the uniform state
\[
 \ket{s}=\frac1{\sqrt N}\sum_{x=0}^{N-1}\ket{x}.
\]
The diffusion operator
\[
 D=2\proj{s}-I
\]
reflects amplitudes about their mean.  One Grover iteration is $G=DO_w$.  The evolution stays in the two-dimensional plane spanned by the marked state $\ket w$ and the uniform superposition $\ket r$ of unmarked states.  Writing
\[
 \ket s=\sin\theta\ket w+\cos\theta\ket r,
 \qquad \sin\theta=\frac1{\sqrt N},
\]
each iteration rotates the state by $2\theta$ toward $\ket w$.  Approximately $\frac\pi4\sqrt N$ oracle calls suffice, and this quadratic scaling is optimal for black-box unstructured search \cite{grover,nielsen-chuang}.

For $N=4$, the mechanism is exact after one iteration.  Before the oracle all four amplitudes are $1/2$.  The oracle changes the marked amplitude to $-1/2$; their mean is then $1/4$.  Reflection about the mean sends the marked amplitude to $1$ and every unmarked amplitude to $0$.

\begin{figure}[ht]
\centering
\begin{tikzpicture}[x=9mm,y=22mm]
  \foreach \x/\lab in {0/00,1/01,2/10,3/11}{
    \draw[MidGray] (\x,-0.02) -- (\x,0.02);
    \node[below,font=\sffamily\scriptsize] at (\x,-0.04) {$\ket{\lab}$};
  }
  \foreach \x in {0,1,2,3}{\draw[DeepBlue,line width=2pt] (\x,0) -- (\x,0.5);}
  \node[above,font=\sffamily\small,color=DeepBlue] at (1.5,0.57) {uniform amplitudes};
  \begin{scope}[xshift=47mm]
    \foreach \x/\lab in {0/00,1/01,2/10,3/11}{
      \draw[MidGray] (\x,-0.52) -- (\x,0.52);
      \node[below,font=\sffamily\scriptsize] at (\x,-0.55) {$\ket{\lab}$};
    }
    \draw[DeepBlue,line width=2pt] (0,0)--(0,0.5);
    \draw[Coral,line width=2pt] (1,0)--(1,-0.5);
    \draw[DeepBlue,line width=2pt] (2,0)--(2,0.5);
    \draw[DeepBlue,line width=2pt] (3,0)--(3,0.5);
    \node[above,font=\sffamily\small,color=Coral] at (1.5,0.57) {phase oracle};
  \end{scope}
  \begin{scope}[xshift=94mm]
    \foreach \x/\lab in {0/00,1/01,2/10,3/11}{
      \draw[MidGray] (\x,-0.02) -- (\x,1.02);
      \node[below,font=\sffamily\scriptsize] at (\x,-0.04) {$\ket{\lab}$};
    }
    \draw[Coral,line width=3pt] (1,0)--(1,1);
    \node[above,font=\sffamily\small,color=Coral] at (1.5,1.04) {after diffusion};
  \end{scope}
  \draw[flow] (34mm,7mm) -- (43mm,7mm);
  \draw[flow] (81mm,7mm) -- (90mm,7mm);
\end{tikzpicture}
\caption{For four items, phase marking followed by inversion about the mean concentrates all amplitude on the marked state (shown here as $\ket{01}$).}
\label{fig:grover-four}
\end{figure}
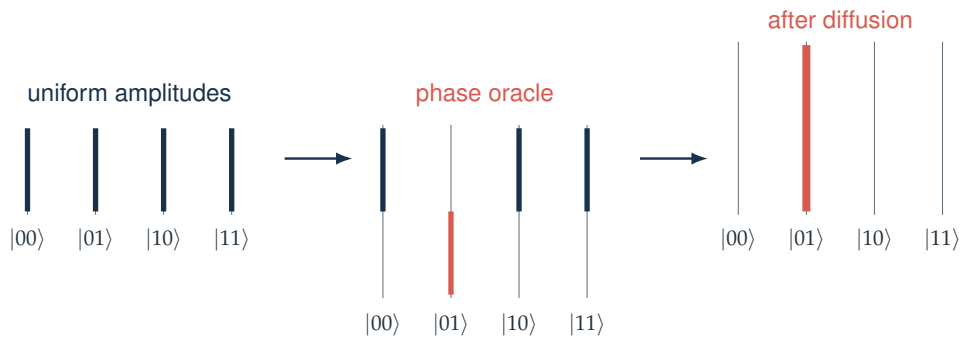

\subsubsection{Two reflections make a rotation}

The oracle $O_w=I-2\proj w$ is a reflection: it changes the sign of the component parallel to $\ket w$ and fixes the orthogonal hyperplane.  The diffusion operator $D=2\proj s-I$ is the reflection that fixes $\ket s$ and negates the orthogonal complement.  The product of two reflections is a rotation in the plane generated by their normals.

Let
\[
 \ket r=\frac1{\sqrt{N-1}}\sum_{x\neq w}\ket x,
 \qquad
 \ket s=\sin\theta\ket w+\cos\theta\ket r,
 \qquad
 \sin\theta=\frac1{\sqrt N}.
\]
In the ordered basis $(\ket w,\ket r)$,
\[
 O_w=\begin{pmatrix}-1&0\\0&1\end{pmatrix},
 \qquad
 D=\begin{pmatrix}
 -\cos2\theta&\sin2\theta\\
 \sin2\theta&\cos2\theta
 \end{pmatrix},
\]
and therefore
\begin{equation}
 G=DO_w=
 \begin{pmatrix}
 \cos2\theta&\sin2\theta\\
 -\sin2\theta&\cos2\theta
 \end{pmatrix}
 \label{eq:grover-rotation}
\end{equation}
up to the orientation chosen for the plane.  Iterating gives
\begin{equation}
 G^k\ket s
 =\sin((2k+1)\theta)\ket w
  +\cos((2k+1)\theta)\ket r.
 \label{eq:grover-iterate}
\end{equation}
The success probability is $\sin^2((2k+1)\theta)$.  The best integer $k$ is the one for which $(2k+1)\theta$ lies nearest $\pi/2$ without substantially overshooting.  For large $N$, $\theta\approx N^{-1/2}$ and
\[
 k\approx\frac{\pi}{4}\sqrt N-\frac12.
\]
Continuing to iterate after the optimum rotates the state away from the marked vector.  Grover amplification is coherent oscillation, not a monotone relaxation process.

\begin{figure}[ht]
\centering
\begin{tikzpicture}[scale=1.05]
  \draw[flow] (-0.4,0)--(5.4,0) node[right,font=\sffamily\small] {$\ket r$};
  \draw[flow] (0,-0.4)--(0,4.3) node[above,font=\sffamily\small] {$\ket w$};
  \draw[DeepBlue,line width=1.3pt,-{Latex[length=2.2mm]}] (0,0)--(4.3,1.1) node[right,font=\sffamily\small] {$\ket s$};
  \draw[Coral,line width=1.3pt,-{Latex[length=2.2mm]}] (0,0)--(3.55,2.55) node[right,font=\sffamily\small] {$G\ket s$};
  \draw[SpiderGreen,line width=1.3pt,-{Latex[length=2.2mm]}] (0,0)--(2.1,3.55) node[above right,font=\sffamily\small] {$G^2\ket s$};
  \draw[MidGray,dashed] (4.3,1.1) arc[start angle=14.4,end angle=50.5,radius=44mm];
  \node[font=\sffamily\small,DeepBlue] at (3.7,0.55) {$\theta$};
  \node[font=\sffamily\small,Coral] at (4.05,1.75) {$2\theta$};
  \begin{scope}[xshift=78mm]
    \node[boxnode,text width=60mm] at (0,2.3) {oracle reflection\\across the unmarked axis};
    \node[greennode,text width=60mm] at (0,0.8) {diffusion reflection\\across the uniform state};
    \node[coralnode,text width=60mm] at (0,-0.7) {combined effect\\rotation by $2\theta$};
  \end{scope}
\end{tikzpicture}
\caption{Grover's iterate is confined to a two-dimensional invariant plane even though the Hilbert space has dimension $N$.}
\label{fig:grover-rotation}
\end{figure}
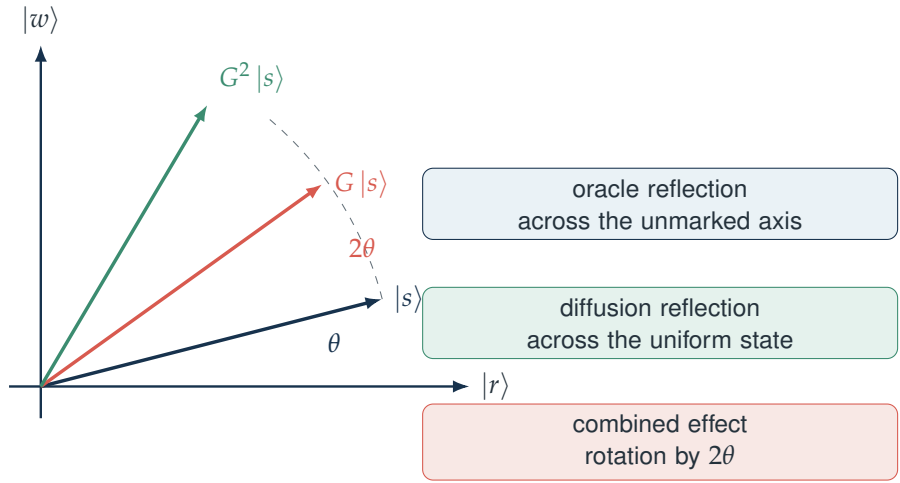

\subsubsection{Inversion about the mean}

For an amplitude vector $\sum_x a_x\ket x$, define the mean
\[
 \bar a=\frac1N\sum_x a_x.
\]
Since $\braket{s}{\psi}=\sqrt N\,\bar a$, the diffusion operator acts componentwise as
\begin{equation}
 a_x\longmapsto 2\bar a-a_x.
 \label{eq:inversion-mean}
\end{equation}
This formula is simply the coordinate expression of reflection about $\ket s$.  For $N=4$ with one marked state, the post-oracle amplitudes are $(-1/2,1/2,1/2,1/2)$ in a suitable ordering.  Their mean is $1/4$, so \cref{eq:inversion-mean} sends the marked amplitude to $1$ and every other amplitude to $0$.

For more than one marked item, replace $\ket w$ by the uniform superposition of the $M$ marked basis states.  Then $\sin\theta=\sqrt{M/N}$, and the query count becomes order $\sqrt{N/M}$ when $M$ is known.  If the number of marked items is unknown, one must modify the schedule; blindly using the single-solution stopping time can overshoot.

\subsubsection{Amplitude amplification beyond search}

Grover's geometry does not depend on the basis states being database entries.  Suppose a unitary $A$ prepares
\[
 A\ket0=\sqrt p\ket{\mathrm{good}}+\sqrt{1-p}\ket{\mathrm{bad}},
\]
where the good and bad subspaces can be distinguished coherently.  Reflections about the prepared state and the good subspace rotate amplitude toward the good component, reducing the number of repetitions from order $1/p$ to order $1/\sqrt p$.  Unstructured search is the special case $p=1/N$.

This generalization makes the conceptual continuity with Deutsch's algorithm clearer.  Both algorithms create a low-dimensional interference problem inside a much larger Hilbert space.  The oracle marks a promised property coherently; a basis change or reflection makes that property observable; the rest of the Hilbert space is organized so that it does not retain which-path information.

\subsection{How many ancillas does a non-unitary process need?}

The question has no single answer until ``non-unitary process'' is typed correctly.

\subsubsection{A contraction as a postselected block}

Let $A:\cH\to\cH$ satisfy $\lVert A\rVert\leq1$.  Define the defect operators
\[
 D_A=(I-A^\dagger A)^{1/2},\qquad D_{A^\dagger}=(I-AA^\dagger)^{1/2}.
\]
The Julia--Halmos block operator
\begin{equation}
 U_A=
 \begin{pmatrix}
 A & D_{A^\dagger}\\
 D_A & -A^\dagger
 \end{pmatrix}
 :\cH\oplus\cH\longrightarrow\cH\oplus\cH
 \label{eq:halmos-dilation}
\end{equation}
is unitary.  Identifying the direct sum with a one-qubit ancilla, $A$ is the ancilla-$0$ block.  Preparing and postselecting that ancilla therefore realizes $A$ probabilistically.  In this mathematical sense, zero extra qubits suffice if $A$ is already unitary; otherwise one ancilla qubit suffices for the block dilation.  Circuit synthesis constraints may demand more workspace.

\subsubsection{A quantum channel as reduced unitary evolution}

A channel $\cE$ has a Kraus representation
\[
 \cE(\rho)=\sum_{j=1}^{r}K_j\rho K_j^\dagger,
 \qquad \sum_jK_j^\dagger K_j=I.
\]
The Stinespring isometry
\[
 V\ket\psi=\sum_{j=1}^{r}K_j\ket\psi\otimes\ket j
\]
satisfies $\cE(\rho)=\Tr_E(V\rho V^\dagger)$.  The minimum environment dimension is the Kraus rank, equivalently the rank of the Choi matrix $J(\cE)$ \cite{choi,stinespring,watrous}.  Hence the minimum number of environment qubits is
\begin{equation}
 q_{\min}=\left\lceil\log_2 \rank J(\cE)\right\rceil.
 \label{eq:channel-ancilla-count}
\end{equation}
A unitary channel has rank one and needs no environment; qubit dephasing and amplitude damping have rank two and need one environment qubit; the generic qubit depolarizing channel has rank four and needs two.

\subsubsection{Three channel examples}

For dephasing with strength $p$,
\[
 \cE_{\mathrm{deph}}(\rho)=(1-p)\rho+pZ\rho Z,
\]
one may choose Kraus operators $K_0=\sqrt{1-p}\,I$ and $K_1=\sqrt p\,Z$.  A one-qubit environment records two alternatives.  Tracing it out suppresses the off-diagonal entries of $\rho$ by the factor $1-2p$ while leaving computational-basis populations unchanged.

Amplitude damping, which models relaxation $\ket1\to\ket0$, has
\[
 K_0=\begin{pmatrix}1&0\\0&\sqrt{1-\gamma}\end{pmatrix},
 \qquad
 K_1=\begin{pmatrix}0&\sqrt\gamma\\0&0\end{pmatrix}.
\]
An isometry on system and one environment qubit is specified by
\begin{align*}
 \ket0\ket0_E&\longmapsto\ket0\ket0_E,\\
 \ket1\ket0_E&\longmapsto
 \sqrt{1-\gamma}\ket1\ket0_E+\sqrt\gamma\ket0\ket1_E.
\end{align*}
The environment state $\ket1_E$ is a record that an excitation was lost.  A unitary extension on the unused input subspace always exists because the displayed images are orthonormal.

The depolarizing channel
\[
 \cE_{\mathrm{dep}}(\rho)
 =(1-p)\rho+\frac p3(X\rho X+Y\rho Y+Z\rho Z)
\]
has four linearly independent Kraus operators for generic $p$.  Its Choi matrix has rank four, so a minimal Stinespring environment has dimension four, namely two qubits.  This count concerns an exact unrestricted dilation.  A particular hardware architecture might use additional ancillas for state preparation, control decomposition, reset, or fault-tolerant implementation.

\subsubsection{Postselection and trace preservation are different tasks}

The block dilation in \cref{eq:halmos-dilation} and the Stinespring dilation solve different problems.  A contraction $A$ appears as one selected block of a larger unitary.  If the ancilla is measured and only the successful outcome is retained, then a normalized input $\ket\psi$ succeeds with probability
\[
 p_{\mathrm{succ}}=\lVert A\ket\psi\rVert^2.
\]
The conditional output is $A\ket\psi/\sqrt{p_{\mathrm{succ}}}$.  This is a non-deterministic branch.

A quantum channel is deterministic and trace preserving after the environment is ignored.  No postselection is required, but coherence between different environment records is lost from the system description.  Confusing these two tasks is the main source of apparently contradictory ancilla counts.

\begin{center}
\begin{tabularx}{0.96\textwidth}{@{}>{\sffamily\bfseries}lXXX@{}}
\toprule
target & enlarged operation & final environment step & minimal size \\
\midrule
unitary $U$ & $U$ itself & none & zero ancillas \\
contraction $A$ & Julia--Halmos unitary & prepare and postselect block & one qubit suffices \\
channel $\cE$ & Stinespring isometry/unitary & trace out environment & $\lceil\log_2\rank J(\cE)\rceil$ qubits \\
\bottomrule
\end{tabularx}
\end{center}

\subsection{Problems for Thinking Quantum Circuits}

\begin{enumerate}[label=\textbf{1.\arabic*.},leftmargin=3.3em]
\item \textbf{Reversible embedding.}  Write the $8\times8$ permutation matrix for the Toffoli gate in the basis $\ket{000},\ldots,\ket{111}$.  Verify directly that it is its own inverse.  Which basis vectors move?

\item \textbf{Relative phase.}  For $\ket{\psi_\phi}=(\ket0+e^{i\phi}\ket1)/\sqrt2$, compute the probabilities of $\pm$ outcomes in an $X$ measurement and of the two eigenstates of $Y$.  Identify $\phi$ modulo $2\pi$ from the two expectation values.

\item \textbf{Separability.}  Apply the determinant test \cref{eq:two-qubit-separable} to each Bell state and to $(\ket{00}+\ket{01})/\sqrt2$.  Find an explicit product factorization in the separable case.

\item \textbf{Deutsch with the wrong ancilla.}  Run the Deutsch circuit symbolically when the second qubit begins in $\ket0$, the first Hadamard is applied, and the target Hadamard is omitted.  Compute the reduced density matrix of the first qubit after the oracle for all four functions.  Which information has become entangled with the target?

\item \textbf{Grover stopping time.}  For $N=8$ and one marked item, calculate $\theta$, the success probabilities after $k=0,1,2,3$ Grover iterations, and the best integer $k$.  Repeat for two marked items.

\item \textbf{A channel ancilla count.}  Compute the Choi matrices of the complete dephasing channel $\rho\mapsto(\rho+Z\rho Z)/2$ and the completely depolarizing qubit channel $\rho\mapsto I/2$.  Determine their ranks and minimal environment dimensions.
\end{enumerate}

\paragraph{Checks.}  In Problem 1 only $\ket{110}$ and $\ket{111}$ are exchanged.  In Problem 2, $\langle X\rangle=\cos\phi$ and $\langle Y\rangle=\sin\phi$.  In Problem 4, the constant functions leave the first qubit in $\ket{+}\!\bra{+}$, while the balanced functions give $I/2$; in the balanced cases the target records the input branch.  For Problem 5, the success probability is given by \cref{eq:grover-iterate}; plotting it against $k$ makes the overshoot visible.  These checks indicate the intended convention but not every intermediate step.

\begin{bridgebox}[From thinking to drawing]
The preceding constructions already use diagrams implicitly: boxes are typed maps, wires carry vector spaces, sequential wiring is composition, and parallel wiring is tensor product.  The next section makes that grammar explicit and asks when pictures can calculate, not merely illustrate.
\end{bridgebox}

\subsection*{End-of-lecture synthesis: thinking operationally}

The circuit model becomes reliable when every calculation begins with a question about \emph{type}.  Is the object a reversible Boolean map, a unitary, a contraction used in a successful branch, or a trace-preserving channel?  The distinction determines what may be copied, what must be uncomputed, whether a measurement record remains, and how many environmental degrees of freedom an implementation needs.  Much confusion about quantum circuits is a type error disguised as physical intuition.

Three further habits organize the lecture.  First, phase is observable only through a reference: a relative phase must be converted into population by interference in a chosen basis.  Second, entanglement is a correlation structure, not a stock of hidden copies.  Its reduced states can be maximally uncertain even when the joint state is pure, which is exactly why it supports protocols unavailable to classical shared randomness.  Third, an algorithm should be read as a controlled flow of information.  A useful pattern is
\[
\begin{aligned}
 &\text{prepare a coherent workspace}
 \;\longrightarrow\;
 \text{mark alternatives by phase}\\[-1mm]
 &\hspace{28mm}\longrightarrow\;
 \text{uncompute records}
 \;\longrightarrow\;
 \text{interfere and measure}.
\end{aligned}
\]
Deutsch's algorithm realizes the smallest instance of phase kickback; Grover's algorithm turns the same principle into a repeated rotation in a two-dimensional invariant subspace.  Their speedups are query statements, so the cost of constructing and decomposing the oracle must be reported separately from the number of oracle calls.

Before accepting a circuit argument, ask five questions: What are the input and output spaces?  Which information is retained by an ancilla or environment?  Which phases are global and which are relative?  What claim is exact, asymptotic, or architecture-dependent?  Finally, where is the decisive interference visible?  These questions are also a specification for the diagrammatic language of the next lecture.  A good drawing should make the types, retained records, relative phases, and points of interference harder---not easier---to overlook.

\clearpage
\section{Drawing Quantum Circuits}
\label{sec:drawing}

\subsection{A diagram is a typed expression}

Fix a finite-dimensional Hilbert space $V$, usually $V=\C^2$.  A box with $m$ input wires and $n$ output wires denotes a linear map
\[
  A:V^{\otimes m}\longrightarrow V^{\otimes n}.
\]
The diagram is not an artist's rendering of a matrix.  It is a typed expression in which the boundary tells us the domain and codomain.  Connecting an output to a compatible input means composition.  Placing two diagrams side by side means tensor product.  These rules form the string-diagram language of monoidal categories \cite{selinger}.

\begin{figure}[ht]
\centering
\begin{tikzpicture}[scale=0.95,transform shape]
  \draw[wire] (0,1) -- (1.3,1); \node[gate] (u) at (2,1) {$U$};
  \draw[wire] (u.east) -- ++(1.2,0); \node[gate] (v) at (4.6,1) {$V$};
  \draw[wire] (v.east) -- ++(1.3,0);
  \node[font=\sffamily\small] at (3.1,0.25) {sequential: $VU$};
  \begin{scope}[xshift=75mm]
    \draw[wire] (0,1.35)--(1.2,1.35); \node[gate] (a) at (1.9,1.35) {$A$}; \draw[wire] (a.east)--++(1.2,0);
    \draw[wire] (0,0.65)--(1.2,0.65); \node[gate] (b) at (1.9,0.65) {$B$}; \draw[wire] (b.east)--++(1.2,0);
    \node[font=\sffamily\small] at (1.9,0.05) {parallel: $A\otimes B$};
  \end{scope}
\end{tikzpicture}
\caption{The geometry distinguishes composition from tensor product before coordinates are chosen.}
\label{fig:two-compositions}
\end{figure}
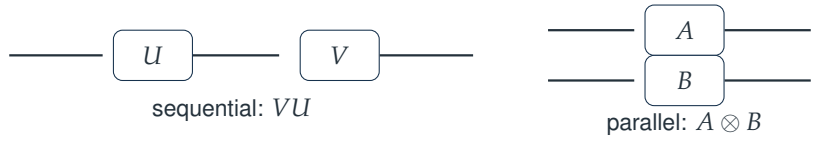

Unitary gates are the reversible $n$-to-$n$ cases, but the graphical language also includes:
\[
 \ket\psi:\C\to V^{\otimes n},
 \qquad
 \bra\phi:V^{\otimes m}\to\C.
\]
The first is a state preparation with no quantum input wire; the second is an effect with no quantum output wire.  Writing them as boundary pieces of the same language avoids pretending that all physically meaningful boxes are square unitaries.

\subsubsection{Indices explain the drawing rules}

Choose bases and write a map $A:U\to V$ as components $A^i{}_j$.  If $B:V\to W$, connecting the output of $A$ to the input of $B$ means summing the shared index:
\[
 (BA)^k{}_j=\sum_i B^k{}_iA^i{}_j.
\]
For maps $A:U\to V$ and $C:X\to Y$, juxtaposition gives
\[
 (A\otimes C)^{i\alpha}{}_{j\beta}=A^i{}_jC^\alpha{}_\beta.
\]
The internal labels are dummy summation indices.  A continuous deformation that preserves which ports are connected does not change the contraction, because it does not change the index expression.  This is the algebra behind planar deformation.

A crossing requires a decision.  In ordinary symmetric quantum circuits it denotes the swap map
\[
 \sigma_{U,V}:U\otimes V\longrightarrow V\otimes U,
 \qquad u\otimes v\longmapsto v\otimes u.
\]
It is not generally safe to erase a crossing as though one wire had passed through another without an operation.  In braided or fermionic settings, the crossing may carry additional phase or braiding data.  The graphical language inherits the structural rules of the category being represented.

\subsubsection{Types catch mistakes before matrices do}

Consider $A:V\to V\otimes V$ and $B:V\otimes V\to V$.  The composites $BA:V\to V$ and $AB:V\otimes V\to V\otimes V$ are both defined but need not be equal and do not even have the same type.  Their diagrams have different external boundaries.  This elementary observation becomes powerful in a large calculation: many proposed circuit identities fail before any entry is multiplied because the numbers, directions, or labels of boundary wires disagree.

Classical outputs also have a different type from coherent quantum wires.  A measurement outcome can control a later unitary, but the control line then carries a classical label.  Treating it as an ordinary qubit wire would incorrectly preserve coherence between different measurement outcomes.  Later, when discarding is introduced, it too will be an explicit typed operation rather than a wire that simply stops.

\begin{example}[Trace as a closed wire]
Let $A:V\to V$.  Bending its output around to its input with the cup and cap of the next subsection gives
\[
 \sum_i\bra iA\ket i=\Tr(A).
\]
The diagram is closed and therefore denotes a scalar.  Its lack of boundary wires is a type statement: no quantum system remains after all indices have been contracted.
\end{example}

\subsection{The categorical structure in plain language}

The formal setting for these pictures is a symmetric monoidal category with a dagger and chosen duals.  The terminology is compact; the operational content is familiar.

\begin{itemize}
\item \textbf{Objects} are system types such as $V$, $V\otimes W$, and the trivial system $\C$.
\item \textbf{Morphisms} are processes $A:V\to W$.
\item \textbf{Composition} joins compatible processes in sequence.
\item The \textbf{monoidal product} places systems and processes in parallel.
\item The \textbf{symmetry} swaps two parallel systems.
\item The \textbf{dagger} reverses a process and takes its adjoint.
\item A chosen \textbf{dual} supplies cups and caps satisfying the snake equations.
\end{itemize}

Associativity and unit laws are normally suppressed in a circuit drawing.  The threefold tensor products $(U\otimes V)\otimes W$ and $U\otimes(V\otimes W)$ are canonically identified, and a trivial wire $\C$ is not drawn.  A coherence theorem guarantees that diagrams built only from these structural identifications agree whenever their connectivity agrees.  This is why one can omit a forest of parentheses without making the semantics ambiguous.

The dagger flips inputs and outputs:
\[
 (BA)^\dagger=A^\dagger B^\dagger,
 \qquad
 (A\otimes B)^\dagger=A^\dagger\otimes B^\dagger.
\]
A unitary is a morphism with $U^\dagger U=UU^\dagger=I$.  An isometry has only $V^\dagger V=I$ and may have more output dimensions than input dimensions.  An effect is the dagger of a state.  Thus several definitions that look unrelated in matrix notation become orientation statements in the same language.

Compact closure adds the controlled bending of wires.  It does not imply that every physical process is time reversible.  Bending an input into an output uses the algebraic dual and a Bell tensor; it is not a laboratory operation that turns a dissipative channel backwards in time.  This distinction matters when a diagram is used as a proof rather than a spacetime cartoon.

Symmetric monoidal diagrams are also topological only in a limited sense.  Stretching and sliding boxes without changing order or connectivity is sound.  Cutting a wire, moving one box through another, deleting a crossing, or changing a knot can require an equation not supplied by the structural axioms.  The diagram is therefore rigid enough to catch types and flexible enough to ignore irrelevant coordinates.

\subsection{Cups, caps, and the controlled bending of wires}

Choose an orthonormal basis $\{\ket i\}_{i=1}^{d}$ of $V$ and define the unnormalized Bell tensor
\begin{equation}
 \ket\Omega=\sum_{i=1}^{d}\ket i\otimes\ket i,
 \qquad
 \bra\Omega=\sum_{i=1}^{d}\bra i\otimes\bra i.
 \label{eq:cup-cap}
\end{equation}
The state $\ket\Omega:\C\to V\otimes V$ is drawn as a cup and the effect $\bra\Omega:V\otimes V\to\C$ as a cap.  For a qubit, $\ket\Omega=\ket{00}+\ket{11}$.  The normalized Bell state is $\ket\Omega/\sqrt d$, but the unnormalized convention makes the wire-straightening identity exact.

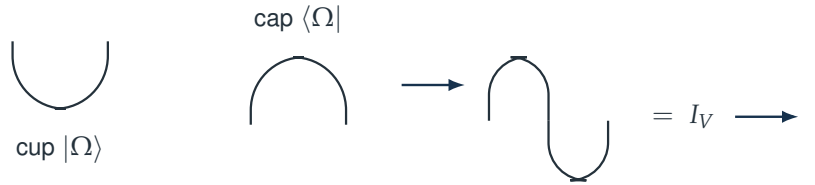
\begin{figure}[ht]
\centering
\begin{tikzpicture}[scale=1.05]
  \draw[wire,rounded corners=7mm] (0,1.2) -- (0,0.35) -- (1.2,0.35) -- (1.2,1.2);
  \node[below,font=\sffamily\small] at (0.6,0.15) {cup $\ket\Omega$};
  \draw[wire,rounded corners=7mm] (3.0,0.15) -- (3.0,1.0) -- (4.2,1.0) -- (4.2,0.15);
  \node[above,font=\sffamily\small] at (3.6,1.18) {cap $\bra\Omega$};
  \draw[flow] (4.9,0.65)--(5.7,0.65);
  \draw[wire,rounded corners=5mm] (6.0,0.2)--(6.0,1.0)--(6.75,1.0)--(6.75,0.2);
  \draw[wire,rounded corners=5mm] (6.75,0.2)--(6.75,-0.55)--(7.5,-0.55)--(7.5,0.2);
  \node[font=\sffamily\small] at (8.45,0.25) {$=\ I_V$};
  \draw[flow] (9.1,0.25)--(9.9,0.25);
  \draw[wire] (10.2,-0.55)--(10.2,1.0);
\end{tikzpicture}
\caption{The cup and cap are dual choices.  Partial contraction gives the snake, or yanking, identity.}
\label{fig:cup-cap-snake}
\end{figure}

Algebraically,
\begin{equation}
 (I_V\otimes\bra\Omega)(\ket\Omega\otimes I_V)=I_V,
 \qquad
 (\bra\Omega\otimes I_V)(I_V\otimes\ket\Omega)=I_V.
 \label{eq:snake}
\end{equation}
One verifies the first equation on a basis vector:
\[
 \sum_i \ket i\,\braket{i}{j}=\ket j.
\]
By linearity it holds on all of $V$.

Why may the picture be rotated?  The precise answer is tensor--Hom adjunction, not freehand topology:
\[
 \operatorname{Hom}(\C,V\otimes V)
 \cong \operatorname{Hom}(V^*,V),
 \qquad
 \operatorname{Hom}(V\otimes V,\C)
 \cong \operatorname{Hom}(V,V^*).
\]
The Bell pairing identifies $V\cong V^*$ after a basis-compatible compact structure has been chosen.  Bending a leg is this currying and dualization operation.  The vertical view is useful because it makes transposition, traces, and the identity \cref{eq:snake} visible.  Contracting horizontally gives the same linear algebra.

The dual choice is essential.  If a cup $\ket{\Omega_C}$ and cap $\bra{\Omega_D}$ are defined by different nondegenerate tensors, their partial contraction is a basis-dependent transfer operator $T(C,D)$, generally not $I$.  Their full closure is
\[
 \braket{\Omega_D}{\Omega_C}=\Tr(D^\dagger C).
\]
If that overlap vanishes, the closed diagram is zero.  If $T(C,D)$ is invertible but not the identity, bending a wire inserts that map.  The common cup--cap convention is exactly what makes yanking valid.

\subsubsection{Sliding a box around a bend produces a transpose}

The maximally entangled tensor satisfies the vectorization identity
\begin{equation}
 (A\otimes I)\ket\Omega=(I\otimes A^{\mathsf T})\ket\Omega.
 \label{eq:cup-transpose}
\end{equation}
Indeed,
\[
 (A\otimes I)\sum_j\ket j\ket j
 =\sum_{i,j}A_{ij}\ket i\ket j
 =(I\otimes A^{\mathsf T})\sum_i\ket i\ket i.
\]
Graphically, moving a box from one leg of a cup to the other turns it over.  The transpose is basis dependent because the compact structure in \cref{eq:cup-cap} was basis dependent.  Reflecting a box and taking its complex conjugate gives the adjoint.  It is therefore important not to say merely that a gate can be ``dragged around'' a bend: the orientation determines whether $A$, $A^{\mathsf T}$, $\bar A$, or $A^\dagger$ appears.

The same identity is the algebraic core of the Choi--Jamio\l kowski correspondence.  A linear map can be encoded as the bipartite vector
\[
 \ket A=(A\otimes I)\ket\Omega,
\]
or, for a channel $\cE$, as the positive operator
\[
 J(\cE)=(\cE\otimes\id)(\proj\Omega).
\]
Bending a wire is what converts an input port into half of a bipartite state.  The ancilla count in \cref{eq:channel-ancilla-count} is consequently also a statement about the rank of this bent representation.

\subsubsection{Teleportation as conditional wire straightening}

Prepare a normalized Bell pair on systems $B,C$ and let an unknown state occupy $A$.  The Bell basis is
\[
 \ket{\beta_{mn}}=(I\otimes X^mZ^n)\ket{\Phi^+},
 \qquad m,n\in\{0,1\}.
\]
The teleportation identity is \cite{bennett-teleportation}
\begin{equation}
 \ket\psi_A\ket{\Phi^+}_{BC}
 =\frac12\sum_{m,n=0}^1
 \ket{\beta_{mn}}_{AB}\otimes X^mZ^n\ket\psi_C.
 \label{eq:teleportation}
\end{equation}
A Bell measurement on $A,B$ selects one branch $(m,n)$.  Two classical bits specify the Pauli correction needed on $C$.  In the $(0,0)$ branch, the cup and cap straighten directly to an identity wire; the other branches have an $X^mZ^n$ decoration that the classical feed-forward removes.

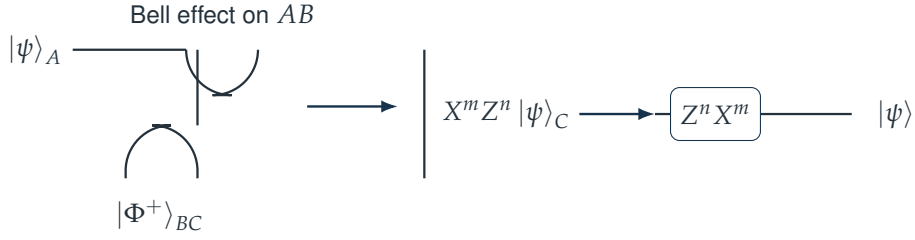
\begin{figure}[ht]
\centering
\begin{tikzpicture}[scale=1.0]
  \draw[wire] (-0.5,1.5)--(1.0,1.5);
  \node[left,font=\sffamily\small] at (-0.5,1.5) {$\ket\psi_A$};
  \draw[wire,rounded corners=6mm] (0.2,-0.2)--(0.2,0.5)--(1.15,0.5)--(1.15,-0.2);
  \node[below,font=\sffamily\small] at (0.65,-0.35) {$\ket{\Phi^+}_{BC}$};
  \draw[wire,rounded corners=6mm] (1.0,1.5)--(1.0,0.9)--(1.95,0.9)--(1.95,1.5);
  \node[above,font=\sffamily\small] at (1.48,1.72) {Bell effect on $AB$};
  \draw[wire] (1.15,0.5)--(1.15,1.5);
  \draw[flow] (2.6,0.75)--(3.8,0.75);
  \draw[wire] (4.15,-0.2)--(4.15,1.5);
  \node[right,font=\sffamily\small] at (4.25,0.65) {$X^mZ^n\ket\psi_C$};
  \draw[flow] (6.2,0.65)--(7.2,0.65);
  \node[gate] (p) at (8.0,0.65) {$Z^nX^m$};
  \draw[wire] (7.2,0.65)--(p); \draw[wire] (p.east)--++(1.2,0);
  \node[right,font=\sffamily\small] at (10.0,0.65) {$\ket\psi$};
\end{tikzpicture}
\caption{A selected teleportation branch is a bent identity wire with a known Pauli decoration.  Classical feed-forward removes the decoration.}
\label{fig:teleportation-straightening}
\end{figure}

\Cref{eq:teleportation} also explains why teleportation does not copy or signal faster than light.  Before the two classical bits arrive, the receiver averages over four Pauli-decorated branches and holds the maximally mixed state.  The unknown input is consumed by the Bell measurement.

\subsubsection{Map--state duality as a calculation tool}

Define the vectorization of $A:V\to W$ by
\[
 \operatorname{vec}(A)=(A\otimes I_V)\ket{\Omega_V}\in W\otimes V.
\]
Then
\begin{align}
 \braket{\operatorname{vec}(A)}{\operatorname{vec}(B)}
   &=\Tr(A^\dagger B),\label{eq:vec-inner}\\
 (C\otimes D)\operatorname{vec}(A)
   &=\operatorname{vec}(CAD^{\mathsf T}).\label{eq:vec-sandwich}
\end{align}
The first identity turns Hilbert--Schmidt operator geometry into ordinary bipartite-state geometry.  The second turns left and right multiplication into local actions on the two halves of a bent wire.

For a unitary $U$ on a $d$-dimensional system, $\operatorname{vec}(U)/\sqrt d$ is maximally entangled.  Conversely, a bipartite pure state is maximally entangled precisely when its unvectorized operator is proportional to a unitary.  This relates gate fidelity to overlap of Choi states.  For example,
\[
 \frac1{d^2}|\Tr(U^\dagger V)|^2
\]
is the squared overlap of the normalized vectorizations of two unitaries.

For a channel $\cE$, the Choi operator
\[
 J(\cE)=\sum_{i,j}\cE(\ket i\bra j)\otimes\ket i\bra j
\]
is positive exactly when $\cE$ is completely positive.  Trace preservation becomes
\[
 \Tr_{\mathrm{out}}J(\cE)=I_{\mathrm{in}}.
\]
These conditions are diagrams with a bent input and a discarded or capped output.  The channel is recoverable from $J(\cE)$ by unbending the input with the same basis convention.  Thus the Choi matrix is not an unrelated trick: it is what a process looks like when one input boundary is converted into a state boundary.

Partial transpose also has a graphical meaning: transpose the operator acting on only one bent leg.  Because partial transpose depends on the chosen tensor factorization and basis, a casual reflection of only part of a diagram can alter positivity.  This is one reason orientation markers and explicit compact conventions matter in mixed-state graphical calculi.

\subsection{Spiders package basis-dependent correlations}

For $m,n\geq0$ and $\alpha\in\mathbb{R}/2\pi\mathbb{Z}$, define the green $Z$-spider
\begin{equation}
 \Zsp{\alpha}{m}{n}
 =\ket0^{\otimes n}\bra0^{\otimes m}
 +e^{i\alpha}\ket1^{\otimes n}\bra1^{\otimes m}.
 \label{eq:z-spider}
\end{equation}
The arity is read directly from the legs.  Several familiar maps occur as special cases:
\[
 \Zsp{0}{0}{2}=\ket{00}+\ket{11},\quad
 \Zsp{0}{2}{0}=\bra{00}+\bra{11},\quad
 \Zsp{0}{1}{1}=I,
\]
and
\[
 \Zsp{0}{1}{2}\ket0=\ket{00},\qquad
 \Zsp{0}{1}{2}\ket1=\ket{11}.
\]
Thus the spider copies the two labels of the computational basis.  It does not clone an arbitrary state:
\[
 \Zsp{0}{1}{2}(\alpha\ket0+\beta\ket1)
 =\alpha\ket{00}+\beta\ket{11},
\]
which is generally entangled and is not $(\alpha\ket0+\beta\ket1)^{\otimes2}$.

The red spider is the same structure in the $X$ basis:
\begin{equation}
 \Xsp{\alpha}{m}{n}
 =H^{\otimes n}\Zsp{\alpha}{m}{n}H^{\otimes m}.
 \label{eq:x-spider}
\end{equation}
It copies $\ket+$ and $\ket-$.  Hadamards on every leg therefore change the colour of a spider.

\subsubsection{One family contains states, effects, products, and coproducts}

Suppressing the phase, write
\[
 \delta=\Zsp{0}{1}{2},\qquad
 \mu=\Zsp{0}{2}{1}=\delta^\dagger,
 \qquad
 \eta=\Zsp{0}{0}{1},\qquad
 \epsilon=\Zsp{0}{1}{0}.
\]
On basis states,
\[
 \delta\ket j=\ket{jj},\qquad
 \mu\ket{jk}=\delta_{jk}\ket j,
\]
while $\eta(1)=\ket0+\ket1$ and $\epsilon=\bra0+\bra1$.  These maps obey associativity, coassociativity, commutativity, cocommutativity, and the Frobenius relation.  They are also \emph{special}:
\[
 \mu\delta=I.
\]
The spider theorem packages all connected diagrams built from this compatible multiplication and comultiplication into a single node with the same external arity.

The zero-input, three-output case is the unnormalized GHZ state
\[
 \Zsp{0}{0}{3}=\ket{000}+\ket{111}.
\]
The same tensor will reappear as the repetition-code encoder once one leg is regarded as an input rather than an output.  The relationship is not a visual coincidence: compact closure converts between these types by bending legs, with the transposes and normalizations fixed by the chosen cup.

\subsubsection{Classical points depend on the observable}

A state $\ket\chi$ is copied by $\delta$ when
\[
 \delta\ket\chi=\ket\chi\otimes\ket\chi
\]
up to the normalization convention appropriate to the diagram.  For the green structure the copied pure states are precisely the computational-basis vectors $\ket0,\ket1$.  The superposition $\ket+$ is not copied:
\[
 \delta\ket+=\frac{\ket{00}+\ket{11}}{\sqrt2}
 \neq\ket+\ket+.
\]
For the red structure the classical points are $\ket+$ and $\ket-$.  Thus the word ``classical'' is relative to a chosen observable structure.  A label that can be broadcast in one basis is coherent information in a complementary basis.

This basis dependence is the graphical form of the no-cloning theorem.  A spider can copy an orthogonal family because linearity then determines an entangling action on all superpositions.  It cannot copy two nonorthogonal states without violating inner-product preservation.

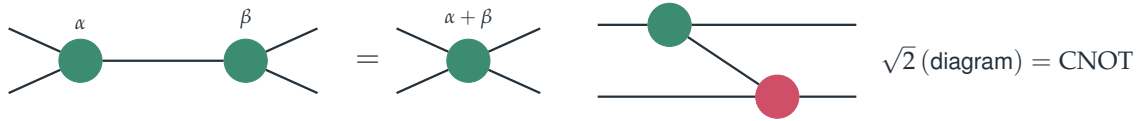
\begin{figure}[ht]
\centering
\begin{tikzpicture}[scale=0.95,transform shape]
  \node[zspider] (z1) at (0,0) {}; \node[zspider] (z2) at (2.3,0) {};
  \draw[wire] (-1.0,0.45)--(z1); \draw[wire] (-1.0,-0.45)--(z1);
  \draw[wire] (z1)--(z2);
  \draw[wire] (z2)--(3.3,0.45); \draw[wire] (z2)--(3.3,-0.45);
  \node[above,font=\sffamily\scriptsize] at (z1.north) {$\alpha$};
  \node[above,font=\sffamily\scriptsize] at (z2.north) {$\beta$};
  \node[font=\sffamily\large] at (4.0,0) {$=$};
  \node[zspider] (z3) at (5.4,0) {};
  \draw[wire] (4.4,0.45)--(z3); \draw[wire] (4.4,-0.45)--(z3);
  \draw[wire] (z3)--(6.4,0.45); \draw[wire] (z3)--(6.4,-0.45);
  \node[above,font=\sffamily\scriptsize] at (z3.north) {$\alpha+\beta$};
  \begin{scope}[xshift=82mm]
    \node[zspider] (g) at (0,0.5) {};
    \node[xspider] (r) at (1.5,-0.5) {};
    \draw[wire] (-1,0.5)--(g); \draw[wire] (g)--(2.6,0.5);
    \draw[wire] (-1,-0.5)--(r); \draw[wire] (r)--(2.6,-0.5);
    \draw[wire] (g)--(r);
    \node[right,font=\sffamily\small] at (2.8,0) {$\sqrt2\,(\text{diagram})=\operatorname{CNOT}$};
  \end{scope}
\end{tikzpicture}
\caption{Left: same-colour spider fusion adds phases.  Right: with normalized $X$-basis states, the standard connected green--red pattern equals $\operatorname{CNOT}/\sqrt2$; scalar-free ZX conventions suppress this factor.}
\label{fig:spider-fusion-cnot}
\end{figure}

\begin{proposition}[Spider fusion]
Two connected spiders of the same colour may be replaced by a single spider whose external legs are the remaining legs and whose phase is the sum of the two phases.
\end{proposition}

\begin{proof}
Expand both spiders using \cref{eq:z-spider}.  Contraction of a connected leg produces inner products $\braket0{1}=\braket1{0}=0$, so the mixed labels vanish.  The all-zero labels survive with coefficient $1$, while the all-one labels survive with coefficient $e^{i\alpha}e^{i\beta}=e^{i(\alpha+\beta)}$.  The red rule follows by conjugating every leg by $H$.
\end{proof}

The proof reveals why the rule is local: orthogonality on the connected wire eliminates the cross terms regardless of how many external legs remain.

\subsection{The ZX-calculus}

The ZX-calculus combines the green and red spider families with Hadamard gates and a collection of sound local equations \cite{coecke-duncan,van-de-wetering}.  Among the basic rules are same-colour fusion, identity spiders, colour change, and complementary-basis interaction rules such as bialgebra and Hopf identities.  The calculus can express any pure qubit linear map when arbitrary real phases and scalar factors are admitted.

It is useful to separate three questions.

\begin{description}[style=nextline,leftmargin=2.6em]
\item[Expressivity.] Can some ZX diagram denote the target map?  With both colours and arbitrary phases, the answer is yes for finite pure-qubit linear maps, hence for every qubit unitary.
\item[Soundness.] Does each permitted redraw preserve the denoted linear map?  Soundness is mandatory: a graphical derivation must imply an algebraic equality.
\item[Completeness.] If two diagrams denote the same map, can their equality be derived from a specified axiom set?  Completeness depends on the fragment and the rules.  It is known for stabilizer quantum mechanics, for suitable Clifford+$T$ axiom systems, and for the full pure-state qubit language with arbitrary phases \cite{backens,jeandel-perdrix-vilmart,hadzihasanovic-ng-wang}.
\end{description}

Fusion alone is not a synthesis method for every unitary.  A one-qubit unitary has an Euler decomposition, up to global phase,
\[
 U\doteq Z(\alpha)X(\beta)Z(\gamma),
\]
where the one-input/one-output phase spiders satisfy
\[
 Z(\alpha)=\operatorname{diag}(1,e^{i\alpha})=e^{i\alpha/2}R_z(\alpha),
 \qquad
 X(\beta)=HZ(\beta)H=e^{i\beta/2}R_x(\beta).
\]
Thus this spider decomposition agrees projectively with the rotation convention in \cref{eq:euler-one-qubit}.  Multi-qubit universality follows from one-qubit rotations together with an entangling gate such as CNOT.  Green and red spiders with phases can represent these ingredients.  Fusion then simplifies adjacent same-colour structure inside the larger diagram.

The available phases determine the exactly reachable fragment.  Phases in $(\pi/2)\mathbb{Z}$ give the Clifford/stabilizer fragment.  Phases in $(\pi/4)\mathbb{Z}$ include Clifford+$T$, which is approximately universal but does not express every real rotation exactly.  Arbitrary real phases restore exact expressivity for finite pure-qubit maps.

\subsubsection{A core rewrite dictionary}

The diagrams vary by convention, but the semantic content of the most-used rules is stable.

\begin{center}
\begin{tabularx}{0.97\textwidth}{@{}>{\sffamily\bfseries}lXX@{}}
\toprule
rule & graphical action & algebraic reason \\
\midrule
spider fusion & merge connected same-colour nodes and add phases & orthogonality leaves one common basis label \\
identity & remove a zero-phase one-input/one-output spider & $\ket0\bra0+\ket1\bra1=I$ \\
colour change & apply $H$ on every leg and swap green/red & $H$ exchanges the $Z$ and $X$ eigenbases \\
phase addition & fuse serial rotations of the same colour & diagonal phases multiply \\
bialgebra & redistribute copying of one colour through multiplication of the other & complementary bases define interacting Frobenius algebras \\
Hopf & remove a doubled opposite-colour connection & parity addition performed twice cancels \\
Hadamard cancellation & remove two consecutive $H$ boxes & $H^2=I$ \\
scalar loop & replace a closed unnormalized dimension-$d$ loop by $d$ & trace of the identity \\
\bottomrule
\end{tabularx}
\end{center}

The table is a study guide, not a complete axiom system.  A complete presentation specifies orientations, phase signs, Hadamard-edge conventions, zero and nonzero scalars, and any supplemental rules required by the chosen fragment.  When importing a rewrite from software or another paper, one should translate conventions before applying it.

Each rule is local, but locality of a rule does not mean locality of the circuit cost after extraction.  Fusing two distant-looking phase gadgets may require routing changes; a local complementation can change many graph edges; a scalar component disconnected from all qubit wires can still determine a branch probability.  ``Local proof'' and ``local hardware change'' are different claims.

\subsubsection{Complementary observables: bialgebra and Hopf structure}

Green and red spiders do more than coexist.  Their interaction expresses the complementarity of the $Z$ and $X$ bases.  One representative bialgebra equation says that copying in one colour distributes through multiplication in the other.  Algebraically it can be checked on basis states; graphically it replaces a small complete bipartite patch of red--green connections by a differently grouped patch with the same boundary.

The Hopf rule is a related cancellation: a doubled connection between opposite-colour spiders can disconnect after the appropriate identities and scalar convention are applied.  These local rules are what make CNOT reasoning effective.  A CNOT is represented by a green copying node on the control joined to a red addition node on the target.  Composing two identical CNOT diagrams places two such interactions in sequence.  Spider fusion collects like colours, the Hopf interaction removes the doubled red--green link, and the diagram reduces to two straight wires, matching
\[
 \operatorname{CNOT}^2=I.
\]
The matrix proof is immediate because $a\oplus a\oplus b=b$; the diagrammatic proof reveals which local algebraic laws make the cancellation possible.

\begin{figure}[ht]
\centering
\begin{tikzpicture}[scale=0.95]
  \node[zspider] (g1) at (0,0.7) {};
  \node[xspider] (r1) at (1.6,-0.7) {};
  \node[zspider] (g2) at (3.5,0.7) {};
  \node[xspider] (r2) at (5.1,-0.7) {};
  \draw[wire] (-1.1,0.7)--(g1)--(g2)--(6.2,0.7);
  \draw[wire] (-1.1,-0.7)--(r1)--(r2)--(6.2,-0.7);
  \draw[wire] (g1)--(r1); \draw[wire] (g2)--(r2);
  \node[font=\sffamily\large] at (7.0,0) {$\rightsquigarrow$};
  \node[zspider] (g) at (8.3,0.7) {};
  \node[xspider] (r) at (9.9,-0.7) {};
  \draw[wire] (7.5,0.7)--(g)--(11.0,0.7);
  \draw[wire] (7.5,-0.7)--(r)--(11.0,-0.7);
  \draw[wire,bend left=18] (g) to (r);
  \draw[wire,bend right=18] (g) to (r);
  \node[font=\sffamily\large] at (11.7,0) {$\rightsquigarrow$};
  \draw[wire] (12.4,0.7)--(14.8,0.7);
  \draw[wire] (12.4,-0.7)--(14.8,-0.7);
  \node[below,font=\sffamily\small] at (2.55,-1.25) {two CNOT patterns};
  \node[below,font=\sffamily\small] at (9.15,-1.25) {fusion};
  \node[below,font=\sffamily\small] at (13.6,-1.25) {Hopf cancellation};
\end{tikzpicture}
\caption{Schematic local reduction of two consecutive CNOTs.  Exact placement of scalar nodes depends on the normalized or scalar-free ZX convention.}
\label{fig:two-cnot-zx}
\end{figure}
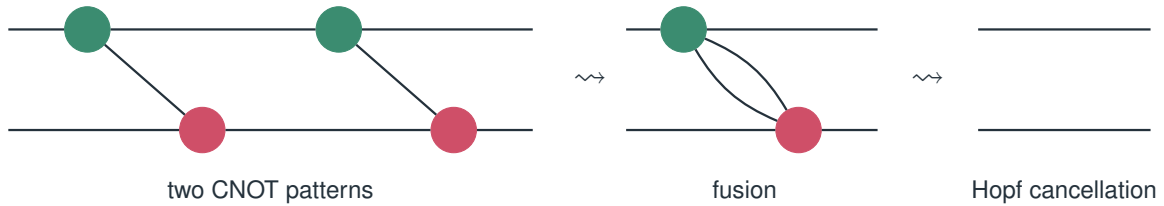

\subsubsection{Phase gadgets expose parity}

A $Z$-phase gadget with phase $\alpha$ and $k$ legs denotes, up to a stated global phase convention, the parity-controlled rotation
\begin{equation}
 U_{\alpha,S}=\exp\!\left(-\frac{i\alpha}{2}
   \bigotimes_{j\in S} Z_j\right).
 \label{eq:phase-gadget}
\end{equation}
The diagram is a phase-labelled spider connected through a complementary-colour parity structure to the selected qubits.  A computational basis vector is an eigenvector of the Pauli product $\bigotimes_{j\in S}Z_j$, so the phase depends only on the XOR parity of those bits.

Phase gadgets are useful in circuit optimization because identical supports combine by adding phases, commuting supports can be reordered, and linear relations among parities become visible in the graph.  They also illustrate the limits of a purely pictorial claim: reducing gadget count need not reduce depth on a particular coupling graph, and an exactly simplified real-angle diagram may still require approximation over a discrete fault-tolerant gate set.

\subsubsection{Scalars, normalization, and zero diagrams}

Many presentations of the ZX-calculus identify nonzero scalar multiples because global phase and normalization are irrelevant for a restricted task.  These notes do not make that identification silently.  A zero-input diagram intended as a physical state must be normalized before probabilities are computed.  A closed diagram denotes a scalar amplitude.  A postselected branch carries its success amplitude.  The empty diagram denotes $1$, while a closed loop made from the unnormalized cup and cap evaluates to
\[
 \braket\Omega\Omega=d.
\]
For qubits the value is $2$.  Suppressing this loop would change a probability calculation.

The zero scalar is especially important.  Orthogonal effects can annihilate a state, and a mismatched folded error diagram in \cref{sec:correcting} reduces to the zero map.  A calculus defined only projectively over nonzero scalars must still treat zero separately.

\subsubsection{Discarding and mixed processes}

A wire that ends in a bra is postselection on a particular effect.  Discarding is different: it traces out a subsystem and is represented by a dedicated discard symbol in graphical languages for completely positive maps.  If
\[
 \rho_{AB}=\proj{\Phi^+},
\]
then postselecting $B$ onto $\bra0$ prepares the unnormalized state $\ket0/\sqrt2$ on $A$, whereas discarding $B$ leaves $I/2$.  Erasing the distinction would confuse a probabilistic branch with ignorance of an outcome.

Pure-state ZX diagrams can represent unitary circuit identities and selected measurement branches.  General noise, reset, classical control, and open-system evolution require an enriched language---for example, doubling constructions or explicitly completely positive generators.  The Knill--Laflamme calculations below use pure operator representatives $E_a$, while the complete noise process is their Kraus sum.

\subsubsection{ZX diagrams and tensor networks}

Both formalisms draw tensors and contracted legs, but their usual questions differ.  A tensor network emphasizes a chosen factorization and the cost or approximation of contracting it.  A ZX diagram emphasizes an algebra of generators and sound equalities that change the factorization.  The same picture can be read both ways once tensors are assigned, but a generic tensor-network deformation is not automatically a ZX rewrite, and a ZX rewrite may deliberately change contraction complexity.

This distinction becomes useful in holographic codes.  A perfect-tensor network specifies an encoder and reconstruction geometry; a ZX representation may then expose stabilizer or Clifford structure inside particular tensors.  One should not assume that every tensor-network symmetry is generated by the chosen ZX rules, or that a graphically simple ZX normal form is the cheapest tensor contraction.

\subsection{From ordinary circuits to graph-like ZX diagrams}

A conventional circuit can be translated gate by gate.  Computational-basis preparations and effects are green spiders with suitable arity; $Z$ rotations are green one-to-one phase spiders; $X$ rotations are red one-to-one phase spiders; Hadamards are explicit basis-change boxes or decorated edges; and CNOT uses the joined green--red pattern of \cref{fig:spider-fusion-cnot}.  The tensor product and composition of the translation follow the circuit wiring.

After translation, simplification often aims for a \emph{graph-like} form:
\begin{enumerate}
\item all interior spiders have one colour, commonly green;
\item connections between spiders are Hadamard edges;
\item there are no parallel edges or self-loops after their scalar effects are reduced;
\item each boundary wire meets a designated boundary spider;
\item phases are collected on vertices.
\end{enumerate}
Local complementation and pivoting rules can then simplify Clifford structure while phase gadgets retain non-Clifford parity rotations.  This form connects circuit optimization to graph algorithms.

\subsubsection{Example: controlled-Z and graph states}

Using \cref{eq:cnot-cz-basis}, a controlled-$Z$ is obtained from the CNOT diagram by putting a Hadamard on the target before and after.  Colour change moves those Hadamards through the red target spider and converts it to green.  The result is two green spiders connected by a Hadamard edge.  Preparing both boundary inputs in $\ket+$ gives the two-vertex graph state
\[
 \operatorname{CZ}\ket{++}
 =\frac12(\ket{00}+\ket{01}+\ket{10}-\ket{11}).
\]
For a graph $G=(V,E)$, preparing $\ket+$ at every vertex and applying CZ on every edge produces
\begin{equation}
 \ket G=\prod_{(u,v)\in E}\operatorname{CZ}_{uv}\ket+^{\otimes |V|}.
 \label{eq:graph-state}
\end{equation}
In graph-like ZX form, the combinatorial graph is visible directly as Hadamard edges among green spiders.  This representation will return in the foliation construction of \cref{sec:correcting}.

\subsubsection{Example: GHZ correlations in two bases}

The normalized three-qubit GHZ state is
\[
 \ket{\mathrm{GHZ}}=\frac{\ket{000}+\ket{111}}{\sqrt2}.
\]
As a green spider, its computational-basis correlation is immediate: all three outcomes agree.  Applying $H$ to every leg changes the node to a red spider and gives
\[
 H^{\otimes3}\ket{\mathrm{GHZ}}
 =\frac12(\ket{000}+\ket{011}+\ket{101}+\ket{110}),
\]
the even-parity superposition.  Thus one spider expresses two complementary descriptions: equality in the $Z$ basis and even parity in the $X$ basis.  The repetition code uses the first description; phase-error diagnosis uses the second.

\subsubsection{Example: commuting a phase through parity computation}

Let a CNOT compute $a\oplus b$ into the target, apply $R_z(\alpha)$ there, and uncompute.  On a basis state $\ket{a,b}$, the acquired phase depends on the parity $a\oplus b$.  With the rotation convention of \cref{eq:euler-one-qubit}, the composite is exactly
\begin{equation}
 \exp\!\left(-\frac{i\alpha}{2}Z_1Z_2\right),
 \label{eq:two-qubit-parity-phase}
\end{equation}
which is the two-leg phase gadget of \cref{eq:phase-gadget}.  Diagrammatically, translate both CNOTs, fuse the target-colour spiders with the intervening phase, and use the interaction rules to expose one phase node connected symmetrically to both qubits.

This calculation is useful because it converts a compute--rotate--uncompute pattern into a commuting Pauli rotation.  Several such rotations can be reordered when their Pauli products commute, and equal supports combine by phase addition.  On hardware, the best extracted circuit depends on which pairs can interact directly.

\begin{figure}[ht]
\centering
\begin{tikzpicture}[scale=0.95]
  \draw[wire] (-0.5,0.7)--(0.5,0.7);
  \draw[wire] (-0.5,-0.7)--(0.5,-0.7);
  \node[zspider] (g1) at (0.8,0.7) {};
  \node[xspider] (r1) at (1.9,-0.7) {};
  \draw[wire] (g1)--(r1);
  \draw[wire] (0.5,0.7)--(g1); \draw[wire] (0.5,-0.7)--(r1);
  \node[zspider] (p) at (3.1,-0.7) {};
  \node[below,font=\sffamily\scriptsize] at (p.south) {$\alpha$};
  \draw[wire] (r1)--(p);
  \node[zspider] (g2) at (4.8,0.7) {};
  \node[xspider] (r2) at (3.9,-0.7) {};
  \draw[wire] (g2)--(r2); \draw[wire] (p)--(r2);
  \draw[wire] (g2)--(6.0,0.7); \draw[wire] (r2)--(6.0,-0.7);
  \node[font=\sffamily\large] at (6.7,0) {$\rightsquigarrow$};
  \draw[wire] (7.4,0.7)--(11.0,0.7);
  \draw[wire] (7.4,-0.7)--(11.0,-0.7);
  \node[zspider] (q) at (9.2,0) {};
  \node[right,font=\sffamily\scriptsize] at (q.east) {$\alpha$};
  \draw[wire] (q)--(8.4,0.7); \draw[wire] (q)--(8.4,-0.7);
  \node[below,font=\sffamily\small] at (2.7,-1.35) {compute, rotate, uncompute};
  \node[below,font=\sffamily\small] at (9.2,-1.35) {parity-phase gadget};
\end{tikzpicture}
\caption{Schematic exposure of a two-qubit parity rotation.  The precise graph-like convention may display Hadamard edges on the gadget legs.}
\label{fig:parity-gadget}
\end{figure}
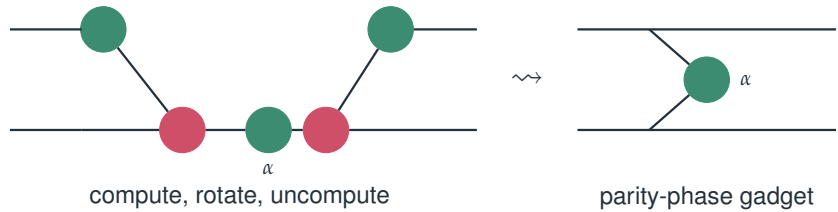

\subsection{Universality, completeness, and normal forms}

These terms answer different questions.

\paragraph{Universality.}  A generator set is universal if every target map in a declared class can be represented exactly or approximated arbitrarily well.  Arbitrary-phase ZX generators represent all finite pure-qubit linear maps.  Clifford+$T$ circuits approximate arbitrary unitaries but represent only a countable exact subset.

\paragraph{Completeness.}  A rewrite system is complete for a fragment if every equality true under the tensor semantics can be derived using its rules.  Stabilizer ZX completeness does not imply that a particular simplification heuristic will find a short proof, and arbitrary-phase expressivity does not by itself supply a complete axiom set.  Nevertheless, complete axiomatisations of the full pure-state qubit language are known \cite{hadzihasanovic-ng-wang}; their existence does not make an arbitrary heuristic rewriting strategy complete.

\paragraph{Normal form.}  A normal form selects a canonical or constrained representative.  If every diagram can be reduced to a unique normal form, equality checking becomes reduction and comparison.  Useful optimization forms need not be unique, and the form easiest for proof may be poor for hardware extraction.

For the stabilizer fragment, completeness and efficient classical simulation are compatible because Clifford operations map Pauli observables to Pauli observables.  Clifford+$T$ introduces non-stabilizer phases; complete axiomatisations exist, but simplification and optimal extraction become substantially harder.  PyZX therefore combines sound local rules with heuristic strategies rather than promising globally optimal circuits \cite{backens,jeandel-perdrix-vilmart,kissinger-wetering}.

\subsection{Three failure modes of an attractive redraw}

\paragraph{A hidden transpose.}  Sliding a non-symmetric gate around a cup without applying \cref{eq:cup-transpose} changes the represented map.  The error may remain invisible for $H$, $X$, or $Z$ because those matrices are symmetric, then fail for a generic complex gate.

\paragraph{A hidden scalar.}  Replacing normalized Bell tensors by the unnormalized compact cup can multiply a branch by $\sqrt2$ per cup.  A projective state ray is unaffected after renormalization, but a success probability or interference between branches is changed.

\paragraph{A hidden environment.}  Erasing a measurement or discard wire can turn a mixed channel into a coherent postselection.  The output density matrix and all later interference predictions change.  Open-system diagrams must retain the environmental or discard structure needed by complete positivity.

These failures motivate the hybrid workflow below.  The diagram discovers a local structural reason; operator semantics checks the orientation, scalar, and open-system type.

\subsection{Graphical proof as local mutation}

A diagrammatic proof is a sequence
\[
 D_0\rightsquigarrow D_1\rightsquigarrow\cdots\rightsquigarrow D_k
\]
in which each step replaces a small subdiagram by an equal one.  If every rule is sound, then $\llbracket D_0\rrbracket=\llbracket D_k\rrbracket$.  The advantage is not that matrices disappear; it is that the reason for a large equality can become spatially local.  Software such as PyZX automates many such transformations \cite{kissinger-wetering}, but scalar conventions, the chosen fragment, the cost function, hardware connectivity, and the verification target must all remain explicit.

\subsubsection{A reproducible graphical calculation}

A trustworthy diagrammatic calculation can be documented in five steps.

\begin{enumerate}
\item \textbf{Declare the semantics.}  Specify the tensor assigned to every generator, the direction of wires, and whether states are normalized.
\item \textbf{Declare the fragment.}  State which phase values and which generators are allowed.  A completeness theorem applies only to its stated fragment.
\item \textbf{Name every rewrite.}  Each mutation should cite fusion, identity, colour change, bialgebra, Hopf, a phase rule, or another explicit axiom.
\item \textbf{Track the boundary and scalar.}  The input/output type must remain fixed, and any loop or disconnected scalar component must be evaluated.
\item \textbf{Verify a terminal form.}  Translate the final small diagram into an operator or compare it numerically on a basis.  This last check is inexpensive and catches convention mismatches.
\end{enumerate}

For Clifford circuits, tableau methods offer an independent polynomial-time verification.  For a small arbitrary-phase circuit, direct matrix comparison is appropriate.  For a large parameterized circuit, one may compare symbolic phase polynomials or evaluate at enough parameter values to test the claimed identity, while remembering that numerical agreement is evidence rather than a formal proof.

\begin{example}[The identity spider]
The one-input, one-output green spider at phase zero is
\[
 \Zsp{0}{1}{1}=\ket0\bra0+\ket1\bra1=I.
\]
Fusion says that attaching it to any green spider changes nothing.  This graphical identity is not a convention imposed after the fact; it follows from the tensor semantics.  At phase $\alpha$, the same arity is instead
\[
 \Zsp{\alpha}{1}{1}=\operatorname{diag}(1,e^{i\alpha}),
\]
so erasing a labelled one-to-one spider would be unsound.
\end{example}

\begin{example}[A normalized cup changes the snake]
If the qubit cup is normalized as $\ket{\Phi^+}=\ket\Omega/\sqrt2$ and the cap is its adjoint, partial contraction gives $I/2$, not $I$.  The normalized tensor is the physical Bell state, while the unnormalized tensor supplies an exact compact-structure yanking equation.  Either convention is valid; mixing them is not.
\end{example}

\subsection{Problems for Drawing Quantum Circuits}

\begin{enumerate}[label=\textbf{2.\arabic*.},leftmargin=3.3em]
\item \textbf{Index semantics.}  Translate a diagram consisting of $A:U\to V\otimes W$ followed by $B:V\to X$ on its first output into indices.  Identify the type of the remaining boundary and verify that deforming the connecting wire does not change the expression.

\item \textbf{Snake normalization.}  Evaluate both partial contractions in \cref{eq:snake} using normalized cups and caps in dimension $d$.  Determine the scalar required to restore the identity.

\item \textbf{Sliding around a cup.}  Prove \cref{eq:cup-transpose} for an arbitrary $d\times d$ matrix.  Repeat with a cap and state which transpose or conjugate appears for each wire orientation.

\item \textbf{Teleportation branch.}  Expand \cref{eq:teleportation} for $\ket\psi=\alpha\ket0+\beta\ket1$.  Check all four Bell outcomes and the order of the Pauli corrections.

\item \textbf{Spider fusion.}  Compose $\Zsp{\alpha}{2}{2}$ with $\Zsp{\beta}{2}{1}$ along one leg and verify the fused formula directly.  What happens if the connected nodes have different colours?

\item \textbf{Classical points.}  Solve $\delta\ket\chi=\ket\chi\otimes\ket\chi$ for an unnormalized qubit vector $\ket\chi=a\ket0+b\ket1$.  Explain why allowing the zero vector is physically unhelpful.

\item \textbf{Scalar bookkeeping.}  Evaluate a closed qubit loop, a normalized Bell cup followed by a normalized Bell cap, and the overlap of $\ket{\Phi^+}$ with $\ket{\Phi^-}$.  Draw the corresponding closed diagrams and distinguish the values $2$, $1$, and $0$.

\item \textbf{Rewrite versus cost.}  Choose a small Clifford+$T$ circuit, translate it into ZX form, and simplify it by hand or with PyZX.  Report separately the number of spiders, the extracted two-qubit gate count, the $T$ count, and the depth on a linear nearest-neighbour architecture.
\end{enumerate}

\paragraph{Checks.}  Problem 2 produces $I/d$ for normalized bends.  In Problem 6 the nonzero copied solutions are the two computational-basis rays.  In Problem 7, the unnormalized loop has value $2$, the normalized Bell norm is $1$, and the orthogonal Bell overlap is $0$.

\begin{keyidea}[The best practical workflow]
Draw to expose locality and topology; rewrite to discover structure; translate the remaining diagram back to an operator; and use algebra or numerics to verify scalars, general error spans, approximation bounds, and hardware costs.
\end{keyidea}

The next section uses precisely this hybrid method.  The algebraic theorem is Knill--Laflamme.  The graphical operation is folding an encoding circuit against two error histories.  The local ZX reduction makes the logical-state independence of the result visible.

\subsection*{End-of-lecture synthesis: drawing with semantics}

A circuit picture earns the right to calculate only after its generators and boundary have been assigned precise tensors.  Wires then carry types; vertical or sequential connection is composition; juxtaposition is tensor product; and planar deformation preserves an expression only while connectivity, order, and orientation remain legitimate.  Cups and caps add a compact structure, but bending a wire is not innocent: a transpose or dual map appears, and normalized Bell states introduce dimension-dependent scalars.  The drawing therefore compresses index notation without abolishing it.

Spiders illustrate the productive tension between intuitive pictures and exact semantics.  A green spider copies the classical points of the computational basis and enforces equality of their labels.  It does not clone arbitrary quantum states.  A red spider expresses the corresponding structure in the complementary basis.  Fusion, colour change, bialgebra, and Hopf rules expose large circuit identities as small local mutations, while phase gadgets collect parity-dependent phases into a form useful for optimization.  For open systems, discarding must be drawn explicitly; otherwise a pure-state equation can be mistaken for equality of channels.

Four claims should remain separate.  \emph{Soundness} says each rewrite preserves the assigned linear map.  \emph{Universality} says the generators can express the maps of interest.  \emph{Completeness} says that, within a stated fragment, semantic equality can be derived by the rules.  \emph{Optimization} says that a chosen rewrite strategy improves a chosen cost after extraction to a particular architecture.  None of these claims implies the others.

The practical verification loop is consequently bilingual.  Declare the generator semantics and scalar convention; rewrite locally to reveal structure; translate the small residual diagram back into an operator or channel; and compare it algebraically or numerically on the fixed boundary.  This is not a retreat from diagrams.  It is what makes them auditable.  In the next lecture, the same loop turns error correction into a picture of two error histories folded around an encoder, while the Knill--Laflamme equations certify exactly what the picture must reduce to.

\clearpage
\section{Correcting Quantum Circuits}
\label{sec:correcting}

\subsection{An error is another process in the circuit}

An ideal wire denotes the identity.  A physical wire is a dynamical system and may accumulate unwanted rotations, dephasing, leakage, crosstalk, or measurement faults.  At the circuit level an error can be drawn as an unwanted gate inserted into the intended process.  More generally, noise is a channel
\[
 \cN(\rho)=\sum_a E_a\rho E_a^\dagger.
\]
The Kraus operators $E_a$ need not describe mutually exclusive classical events; they span the operator space of the channel.  This point is why the exact correction condition must test $E_a^\dagger E_b$ for every pair.

\subsubsection{Why Pauli errors are a useful basis}

Every operator on one qubit has an expansion
\begin{equation}
 A=a_0I+a_xX+a_yY+a_zZ,
 \label{eq:pauli-expansion}
\end{equation}
because $I,X,Y,Z$ are an orthogonal basis under the Hilbert--Schmidt inner product.  Consequently, a code that corrects the four Pauli components on a given qubit corrects their entire complex linear span, including a small coherent rotation such as
\[
 e^{-i\epsilon X/2}=\cos(\epsilon/2)I-i\sin(\epsilon/2)X.
\]
Quantum error correction does not require the physical error to have ``really been'' either $I$ or $X$.  The syndrome measurement can correlate orthogonal error sectors with different classical records, turning a coherent superposition of correctable components into a branch that can be recovered.

For $n$ qubits, tensor products of Pauli operators form a basis for all $2^n\times2^n$ operators.  The \emph{weight} of a Pauli is the number of nonidentity tensor factors.  Distance statements are usually formulated in this basis: a distance-$d$ code detects every Pauli of weight below $d$ and corrects arbitrary noise supported on at most $\lfloor(d-1)/2\rfloor$ qubits.

Leakage and loss require care because they leave the qubit Hilbert space.  One can enlarge the local space and include leakage operators in the error set, or use dedicated leakage-reduction units.  Likewise, correlated noise is not captured by a model that assigns independent single-qubit Pauli probabilities.  The theorem below accepts any declared operator set; the experimental challenge is declaring a set that faithfully covers the noise of interest.

\subsubsection{Detection, correction, and fault tolerance}

These three claims have different strengths.

\begin{itemize}
\item A code \emph{detects} an error set when the code space can be distinguished from the erroneous subspaces without distinguishing logical states.
\item It \emph{corrects} the set when a recovery channel restores every encoded state after any error in the linear span.
\item A protocol is \emph{fault tolerant} when faults occurring during encoding, syndrome extraction, recovery, and measurement do not spread into uncorrectable patterns faster than they are diagnosed.
\end{itemize}

An ideal recovery map is not yet a fault-tolerant circuit.  For example, one ancilla that sequentially interacts with many data qubits can propagate a single ancilla fault into a high-weight data error.  Repeated syndrome rounds, verified ancillas, locality, decoder timing, and a declared circuit-level noise model enter only at the fault-tolerant level.

A quantum code consists of a logical Hilbert space $\cH_L$, a physical Hilbert space $\cH_P$, and an isometric encoder
\[
 V:\cH_L\hookrightarrow\cH_P,
 \qquad V^\dagger V=I_L.
\]
The code space is $\cC=V\cH_L$ and its orthogonal projector is $P=VV^\dagger$.  A recovery is a channel $\cR$ intended to reverse the noise on states supported in $\cC$.

\subsection{Encoding is not cloning: the three-qubit repetition code}

The repetition encoder is
\begin{equation}
 V\ket0=\ket{000},\qquad V\ket1=\ket{111}.
 \label{eq:repetition-encoder}
\end{equation}
Therefore
\[
 V(\alpha\ket0+\beta\ket1)=\alpha\ket{000}+\beta\ket{111}.
\]
This is one logical qubit stored in correlations among three physical qubits, not three independent copies.  In the language of \cref{sec:drawing}, $V$ is the green spider $\Zsp{0}{1}{3}$.

Assume the promised error set
\[
 \mathscr E=\{I,X_1,X_2,X_3\},
\]
where $X_i$ flips physical qubit $i$.  The code detects these errors by measuring whether neighbouring bits agree.  The corresponding commuting observables are $Z_1Z_2$ and $Z_2Z_3$.

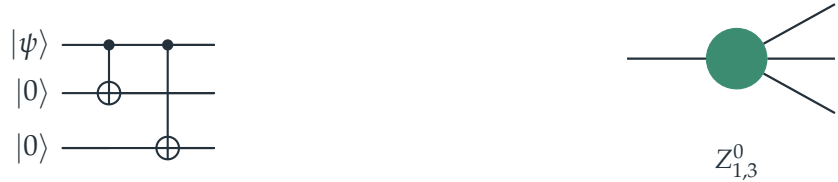
\begin{figure}[ht]
\centering
\begin{minipage}{0.54\textwidth}
\centering
\begin{quantikz}[row sep=0.4cm,column sep=0.45cm]
\lstick{$\ket\psi$} & \ctrl{1} & \ctrl{2} & \qw \\
\lstick{$\ket0$}    & \targ{}  & \qw      & \qw \\
\lstick{$\ket0$}    & \qw      & \targ{}  & \qw
\end{quantikz}
\end{minipage}
\hfill
\begin{minipage}{0.41\textwidth}
\centering
\begin{tikzpicture}
  \draw[wire] (-1.1,0)--(0,0); \node[zspider,minimum size=8mm] (z) at (0.35,0) {};
  \draw[wire] (z)--(1.65,0.72); \draw[wire] (z)--(1.65,0); \draw[wire] (z)--(1.65,-0.72);
  \node[below,font=\sffamily\small] at (0.35,-1.0) {$\Zsp{0}{1}{3}$};
\end{tikzpicture}
\end{minipage}
\caption{The two-CNOT encoder and its one-spider ZX representation define the same isometry.}
\label{fig:repetition-encoder}
\end{figure}

Encode agreement as syndrome bit $0$ and disagreement as $1$.  Then
\begin{center}
\begin{tabular}{@{}cccc@{}}
\toprule
error & $Z_1Z_2$ check & $Z_2Z_3$ check & recovery \\
\midrule
$I$   & $0$ & $0$ & $I$ \\
$X_1$ & $1$ & $0$ & $X_1$ \\
$X_2$ & $1$ & $1$ & $X_2$ \\
$X_3$ & $0$ & $1$ & $X_3$ \\
\bottomrule
\end{tabular}
\end{center}
The syndrome reveals the error sector, not the amplitudes $\alpha$ and $\beta$.  An ancilla may extract one parity by receiving CNOTs from the two data qubits and then being measured.  It learns their XOR without measuring either data value separately.

\subsubsection{Syndrome extraction as a coherent calculation}

To measure $Z_1Z_2$, prepare an ancilla in $\ket0$, apply CNOT from data qubit $1$ to the ancilla and then CNOT from data qubit $2$ to the ancilla, and measure the ancilla in the computational basis.  On a basis string $\ket{x_1x_2}$ the ancilla becomes $\ket{x_1\oplus x_2}$.  On a superposition of codewords, both $\ket{000}$ and $\ket{111}$ yield parity zero, so the logical amplitudes remain coherent.

\begin{figure}[ht]
\centering
\begin{quantikz}[row sep=0.42cm,column sep=0.48cm]
\lstick{data 1} & \ctrl{3} & \qw      & \qw      & \qw \\
\lstick{data 2} & \ctrl{2} & \ctrl{3} & \qw      & \qw \\
\lstick{data 3} & \qw      & \ctrl{2} & \qw      & \qw \\
\lstick{$\ket0_{a_1}$} & \targ{} & \qw & \meter{} & \cw \\
\lstick{$\ket0_{a_2}$} & \qw & \targ{} & \meter{} & \cw
\end{quantikz}
\caption{Ideal extraction of the two repetition-code parities.  The ancillas reveal agreement information and not the value of the encoded basis label.}
\label{fig:repetition-syndrome-circuit}
\end{figure}
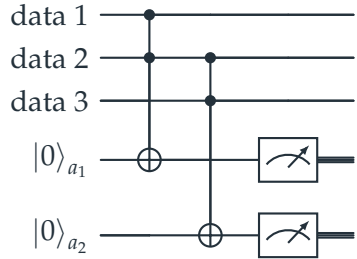

Let $\cC_0=\cC$ and $\cC_i=X_i\cC$ for $i=1,2,3$.  These four two-dimensional subspaces are pairwise orthogonal and fill the eight-dimensional physical Hilbert space:
\[
 (\C^2)^{\otimes3}=\cC_0\oplus\cC_1\oplus\cC_2\oplus\cC_3.
\]
The two syndrome bits label this direct-sum decomposition.  An ideal recovery first correlates each sector with its label and then applies $X_i$ conditionally.  Because the same logical vector $\alpha\ket{000}+\beta\ket{111}$ appears inside every sector after the appropriate inverse flip, the syndrome record is statistically independent of $\alpha,\beta$.

If two data qubits flip, the observed syndrome coincides with the syndrome of a different single-qubit error.  For instance $X_1X_2$ and $X_3$ both produce the same pair of parity disagreements.  Applying the minimum-weight correction $X_3$ leaves $X_1X_2X_3$, which acts as logical $X$ on the code.  This is the simplest example of decoder ambiguity: the syndrome identifies a fibre of possible histories, not necessarily the true history.

The promise has limits.  A phase flip gives
\[
 Z_1(\alpha\ket{000}+\beta\ket{111})
 =\alpha\ket{000}-\beta\ket{111},
\]
while both neighbouring pairs still agree.  The syndrome is $00$, so this code does not detect the phase error.  Conjugating by Hadamard exchanges $X$ and $Z$ errors; full quantum codes combine both kinds of protection.

\subsection{From repetition to full single-qubit protection}

The bit-flip repetition code protects equality in the $Z$ basis.  Conjugating every physical qubit by $H$ gives a phase-flip repetition code with codewords
\[
 \ket{0_L}=\ket{+++},\qquad
 \ket{1_L}=\ket{---}.
\]
It diagnoses one $Z$ error using $X$-basis parity checks but does not protect against an $X$ error.  A full distance-$3$ quantum code must distinguish the Pauli components $X,Y,Z$ on every one of its protected locations.

Shor's nine-qubit code concatenates the two repetition ideas \cite{shor-code}.  One layer protects phase information by repeating in the $X$ basis, while three-qubit GHZ blocks protect each layer against bit flips.  Its codewords can be written
\[
 \ket{0_L}=\frac{(\ket{000}+\ket{111})^{\otimes3}}{2\sqrt2},
 \qquad
 \ket{1_L}=\frac{(\ket{000}-\ket{111})^{\otimes3}}{2\sqrt2}.
\]
The construction is pedagogically transparent but not economical in qubit count.

Steane's seven-qubit code is a CSS code derived from the classical $[7,4,3]$ Hamming code \cite{steane-code}.  It encodes one qubit and has distance three, using separate $X$- and $Z$-type parity checks with the same support pattern.  The CSS structure allows bit and phase components to be decoded by related classical syndromes.

The five-qubit perfect code is the smallest code correcting an arbitrary error on one qubit.  It is not CSS: its stabilizers mix $X$ and $Z$ factors.  These examples show that full quantum protection is not tied to one construction.  The invariant statement is the distance or, more generally, the Knill--Laflamme condition for the declared error span.

\begin{center}
\begin{tabularx}{0.96\textwidth}{@{}>{\sffamily\bfseries}lccccX@{}}
\toprule
code & $n$ & $k$ & $d$ & CSS? & main lesson \\
\midrule
repetition &3&1&1&yes&bit-flip distance is $3$ under the restricted promise\\
five-qubit &5&1&3&no&smallest full one-qubit correction\\
Steane &7&1&3&yes&classical parity-check structure for $X$ and $Z$\\
Shor &9&1&3&yes&concatenate bit- and phase-repetition ideas\\
surface family &variable&variable&variable&yes&local checks and scalable distance\\
\bottomrule
\end{tabularx}
\end{center}

Formally, the repetition code is CSS with an empty $X$-check matrix; this classification does not imply protection in both Pauli sectors.  The code does not have quantum distance $3$ because the weight-one operator $Z_1$ acts nontrivially within the code without being detected.  It has distance $3$ only relative to the promised classical bit-flip error model.  Writing the promise into the parameter table prevents an instructive toy code from being mistaken for a full quantum code.

Distance still does not specify a fault-tolerant architecture.  The five-qubit code has excellent qubit efficiency but weight-four mixed-Pauli checks.  Surface codes use more qubits for a given distance but offer geometrically local checks and well-developed repeated syndrome circuits.  Hardware connectivity, native gates, measurement fidelity, leakage, and decoder speed determine which tradeoff is useful.

\subsection{The Knill--Laflamme condition}

The decisive question is not whether error sectors look different in a chosen example.  It is whether the environment or syndrome can learn anything about the logical state.

\begin{theorem}[Knill--Laflamme \cite{knill-laflamme}]
Let $P$ project onto a code space $\cC$ and let $\mathscr E=\{E_a\}$ be a set of errors.  There exists an exact recovery correcting the linear span of $\mathscr E$ if and only if there is a matrix of scalars $C=(c_{ab})$ such that
\begin{equation}
 P E_a^\dagger E_b P=c_{ab}P
 \qquad\text{for every }a,b.
 \label{eq:knill-laflamme}
\end{equation}
\end{theorem}

The scalar $c_{ab}$ may distinguish or correlate error descriptions, but it cannot depend on the encoded logical state.  The matrix $C$ is positive semidefinite.  Off-diagonal entries may be nonzero, and degenerate errors may act identically on the code.  By changing the Kraus basis one may diagonalize $C$, but the general condition is scalar action, not specifically the numbers $0$ and $1$.

\subsubsection{Why the condition is necessary}

Let the noise isometry be
\[
 W\ket\psi=\sum_a E_a\ket\psi\otimes\ket a_E.
\]
For encoded logical states $\ket\phi,\ket\psi$, the environment overlap between branches $a,b$ is controlled by
\[
 \bra\phi V^\dagger E_a^\dagger E_bV\ket\psi.
\]
If recovery is to restore an arbitrary superposition, the environment may carry information about which error sector occurred but not about the logical amplitudes.  The operator on the logical space must therefore be a scalar:
\[
 V^\dagger E_a^\dagger E_bV=c_{ab}I_L.
\]
Multiplying by $V$ and $V^\dagger$ gives \cref{eq:knill-laflamme}.  Necessity can also be obtained from reversibility of the channel restricted to the code, but the environment argument makes the information-flow content visible.

\subsubsection{Why the condition is sufficient}

Because $C$ is positive semidefinite, choose a unitary change of error basis that diagonalizes it.  The new error operators $F_r$ obey
\[
 P F_r^\dagger F_sP=\lambda_r\delta_{rs}P.
\]
For every $\lambda_r>0$, the map $F_rV/\sqrt{\lambda_r}$ is an isometry, and different $r$ have orthogonal images.  A syndrome measurement can project onto these images without learning the logical state.  On sector $r$, apply the inverse isometry to return to the code.  Extending the construction arbitrarily on the unused orthogonal complement produces a completely positive trace-preserving recovery.  Thus the pairwise scalar condition constructs the sector decomposition needed for correction.

This proof also explains why checking only $PE_aP$ is insufficient.  Coherent noise contains cross terms $E_a\rho E_b^\dagger$, so recovery must preserve interference for every pair.  The products $E_a^\dagger E_b$ compare precisely those branches.

\subsubsection{Degenerate errors}

Different physical operators may have the same action on the code.  If $E_aP=E_bP$, then
\[
 P E_a^\dagger E_bP=P E_a^\dagger E_aP,
\]
and the corresponding off-diagonal $c_{ab}$ need not vanish.  The syndrome need not distinguish $E_a$ from $E_b$ because no recovery decision depends on that distinction.  Surface codes are highly degenerate: multiplying an error string by a stabilizer changes its physical support while leaving its logical action and endpoint syndrome unchanged.

\begin{warningbox}[A frequent oversimplification]
The formula $PE_a^\dagger E_bP=\delta_{ab}P$ is sufficient but not the general theorem.  It describes a convenient orthonormal error basis for a nondegenerate example.  The correct invariant statement is scalar action $c_{ab}P$ after restriction to the code.
\end{warningbox}

\subsubsection{Approximate correction}

Physical noise rarely satisfies an exact finite promise.  Approximate QEC asks whether there is a recovery for which the corrected logical channel is close to the identity according to entanglement fidelity, diamond distance, or another operational metric.  A useful perturbative form of Knill--Laflamme is
\begin{equation}
 P E_a^\dagger E_bP=c_{ab}P+P\Delta_{ab}P,
 \label{eq:approx-kl}
\end{equation}
where the residual operators $\Delta_{ab}$ act weakly or nearly as scalars on the code.  The size and collective structure of these residuals bound how much logical information leaks to the environment.

Approximate correction is not permission to ignore a failed exact test.  One must choose an error metric and quantify the residual.  A small operator norm per pair may still accumulate across many correlated channels; conversely, a code can perform well for a high-probability noise distribution even if a worst-case low-probability operator is poorly corrected.  The declared metric determines the claim.

The distinction between exact and approximate statements parallels the later experiment discussion.  A finite device spectrum can approximate a target weighted graph within disorder and linewidth.  A logical protocol can approximate an identity channel within a measured failure rate.  In both cases the tolerance is part of the theorem or benchmark, not an aesthetic judgement from a plot.

For the repetition error set $\{I,X_1,X_2,X_3\}$, equal errors give $E_a^\dagger E_a=I$.  If $a\neq b$, the product flips one or two physical qubits and maps $\operatorname{span}\{\ket{000},\ket{111}\}$ to an orthogonal subspace.  Hence
\[
 P E_a^\dagger E_bP=\delta_{ab}P.
\]
The clean $0/I$ outcome is therefore a special nondegenerate example of \cref{eq:knill-laflamme}.

\subsection{Folding evaluates the same operator}

Because $P=VV^\dagger$, the Knill--Laflamme operator on the logical space is
\begin{equation}
 F(a,b)=V^\dagger E_a^\dagger E_bV.
 \label{eq:folded-map}
\end{equation}
Graphically, encode, apply the forward history $E_b$, reverse the history $E_a$, and decode.  Exact correction requires the folded logical wire to reduce to $c_{ab}I_L$.

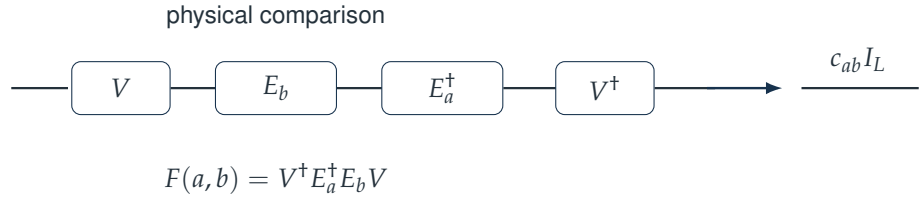
\begin{figure}[ht]
\centering
\begin{tikzpicture}[node distance=10mm]
  \draw[wire] (-1.0,0)--(-0.25,0);
  \node[gate,minimum width=13mm] (v) at (0.45,0) {$V$};
  \draw[wire] (v.east)--++(0.7,0); \node[gate,minimum width=16mm] (eb) at (2.5,0) {$E_b$};
  \draw[wire] (eb.east)--++(0.7,0); \node[gate,minimum width=16mm] (ea) at (4.7,0) {$E_a^\dagger$};
  \draw[wire] (ea.east)--++(0.7,0); \node[gate,minimum width=13mm] (vd) at (6.85,0) {$V^\dagger$};
  \draw[wire] (vd.east)--++(0.8,0);
  \node[font=\sffamily\small,above=3mm of eb] {physical comparison};
  \node[font=\sffamily\small,below=5mm of eb] {$F(a,b)=V^\dagger E_a^\dagger E_bV$};
  \draw[flow] (8.2,0)--(9.2,0);
  \draw[wire] (9.45,0)--(11.0,0);
  \node[above,font=\sffamily\small] at (10.2,0.08) {$c_{ab}I_L$};
\end{tikzpicture}
\caption{Folding is not a competing correctability condition; it is a diagrammatic evaluation of the Knill--Laflamme operator on the logical wire.}
\label{fig:folding-kl}
\end{figure}

For the repetition code, equal bit-flip histories cancel and the encoding spider fuses with its adjoint to a straight wire.  Distinct promised histories leave a mismatched computational-basis pattern, which the adjoint spider annihilates.  The local ZX reductions therefore certify the same $\delta_{ab}I_L$ calculation.

The diagrammatic and algebraic proofs have different strengths.  The operator equation is basis-independent, handles general Kraus spans and degeneracy, keeps scalar normalization explicit, and extends naturally to approximate correction.  The diagram exposes composition, locality, and the distinction between syndrome data and logical data.  The strongest method uses both.

\subsection{From syndromes to footprints}

In the repetition code the two parity bits are locally visible information left by an error history.  A more general field-theoretic formulation calls such information a \emph{footprint} \cite{rayan-diagrammatic}.

Let $R$ be a bounded region of a circuit spacetime and let $\mathscr H_R$ denote a chosen class of error histories supported in $R$.  A footprint map
\[
 \Foot_R:\mathscr H_R\longrightarrow\mathscr F_{\partial R}
\]
retains specified locally visible data on the boundary of $R$: for example charge sectors, defect endpoints, parity outcomes, locations, or measurement rounds.  For an observed footprint $f$, the fibre
\begin{equation}
 \Foot_R^{-1}(f)=\{e\in\mathscr H_R:\Foot_R(e)=f\}
 \label{eq:footprint-fibre}
\end{equation}
is the candidate set presented to the decoder.

The fibre is a set of histories, not an additional physical object supported on the boundary.  Its elements occupy the interior region $R$; their common footprint is the boundary datum $f$.  Correctability becomes a fibrewise question: do all histories compatible with the observed local data have a logical-state-independent comparison, or at least admit a common recovery class?

\begin{figure}[ht]
\centering
\begin{tikzpicture}[node distance=8mm and 14mm]
  \node[boxnode,text width=34mm] (hist) {interior histories\\$e_1,e_2,e_3,\ldots$};
  \node[coralnode,text width=30mm,right=of hist] (foot) {observed footprint\\$f$ on $\partial R$};
  \node[greennode,text width=35mm,right=of foot] (dec) {decoder chooses\\a recovery class};
  \draw[flow] (hist)--node[above,font=\sffamily\scriptsize] {$\Foot_R$} (foot);
  \draw[flow] (foot)--(dec);
  \node[below=6mm of foot,font=\sffamily\small] {$\Foot_R^{-1}(f)$ is the ambiguity the decoder must resolve};
\end{tikzpicture}
\caption{A syndrome is one kind of footprint.  In larger codes, one footprint generally has several compatible histories.}
\label{fig:footprint-fibre}
\end{figure}

\subsubsection{A decoder chooses a class, not a microscopic past}

Given footprint $f$, a decoder returns a recovery $r(f)$.  Success does not require $r(f)$ to equal the actual error history $e$.  It requires the comparison $r(f)e$ to act trivially on the logical information.  In a stabilizer code this means
\[
 r(f)e\in\mathcal S
\]
up to a physically irrelevant phase, where $\mathcal S$ is the stabilizer group.  If instead $r(f)e$ belongs to a nontrivial logical coset, the recovery introduces or completes a logical fault.

Maximum-likelihood decoding therefore asks for the most probable \emph{logical class} conditioned on the observed syndrome, not necessarily the most probable individual error.  Minimum-weight decoding is an approximation justified for some independent low-rate noise models; degeneracy can make the total probability of a larger class exceed that of its single most likely representative.

The field-theoretic language separates three objects:
\begin{enumerate}
\item the interior support and time ordering of a history $e$;
\item the measured boundary datum $f=\Foot_R(e)$;
\item the residual logical class of $r(f)e$.
\end{enumerate}
This separation is useful for repeated measurement, where measurement faults themselves form spacetime edges and a syndrome \emph{change} between rounds is often the relevant footprint.

\subsection{Stabilizers and CSS parity checks}

An $n$-qubit stabilizer code is the joint $+1$ eigenspace of an abelian subgroup $\mathcal S$ of the Pauli group that does not contain $-I$ \cite{gottesman}.  If $n-k$ independent generators are imposed, the code space has dimension $2^k$.  A Pauli error anticommutes with some stabilizer generators; the signs of those measured eigenvalues form its syndrome.

A CSS code separates $X$-type and $Z$-type stabilizers.  Let binary parity-check matrices $H_X$ and $H_Z$ have rows specifying their supports.  Commutation is exactly
\begin{equation}
 H_XH_Z^{\mathsf T}=0\pmod2.
 \label{eq:css-orthogonality}
\end{equation}
An $X$ error represented by a binary vector $e_X$ produces $Z$-check syndrome $H_Ze_X^{\mathsf T}$; a $Z$ error produces $X$-check syndrome $H_Xe_Z^{\mathsf T}$.  Logical operators lie in the corresponding kernels but outside the row spaces generated by stabilizers.

For the three-qubit repetition code,
\[
 H_Z=\begin{pmatrix}1&1&0\\0&1&1\end{pmatrix}.
\]
It checks $X$ errors but has no nontrivial $X$-type stabilizers protecting against $Z$ errors.  The all-ones vector lies in $\ker H_Z$ and represents the logical $X=X_1X_2X_3$.  This binary calculation is the finite-dimensional precursor of the chain-complex description of surface codes.

\begin{figure}[ht]
\centering
\begin{tikzpicture}[node distance=7mm and 9mm]
  \node[boxnode,text width=30mm] (e) {Pauli error vector\\$e_X$ or $e_Z$};
  \node[coralnode,text width=30mm,right=of e] (h) {parity checks\\$H_Ze_X^{\mathsf T}$\\or $H_Xe_Z^{\mathsf T}$};
  \node[greennode,text width=34mm,right=of h] (q) {kernel modulo row space\\logical quotient};
  \draw[flow] (e)--(h); \draw[flow] (h)--(q);
  \node[below=7mm of h,font=\sffamily\small,align=center] {local syndrome first;\\global equivalence second};
\end{tikzpicture}
\caption{CSS decoding has a linear-algebra layer and a quotient layer.  Zero syndrome means membership in a kernel, not necessarily membership in the stabilizer row space.}
\label{fig:css-kernel-quotient}
\end{figure}
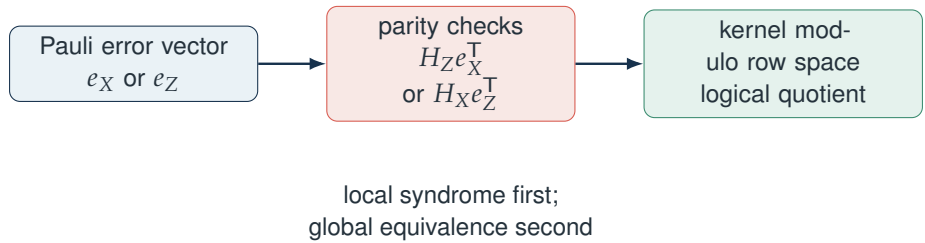

\subsection{Surface codes: boundaries, cycles, and logical classes}

Let a closed surface $\Sigma$ be equipped with a finite cellulation.  Over $\Ftwo$ its cellular chain complex is
\begin{equation}
 C_2(\Sigma)\xrightarrow{\partial_2}C_1(\Sigma)
 \xrightarrow{\partial_1}C_0(\Sigma),
 \qquad \partial_1\partial_2=0.
 \label{eq:cellular-chain}
\end{equation}
Place data qubits on edges.  In one standard convention, $Z$-type error strings are 1-chains $e\in C_1$.  Their endpoint syndrome is $\partial_1e$.  Thus
\[
 Z_1=\ker\partial_1
\]
is the space of closed, locally undetectable strings.  Face boundaries form
\[
 B_1=\im\partial_2
\]
and act as products of plaquette stabilizers.  The logical $Z$ classes are therefore
\begin{equation}
 H_1(\Sigma;\Ftwo)=Z_1/B_1.
 \label{eq:logical-homology}
\end{equation}
The dual statement identifies logical $X$ operators with
\[
 H^1(\Sigma;\Ftwo)=Z^1/B^1,
\]
or equivalently with nontrivial cycles on the dual cellulation.  The $X/Z$ naming may be interchanged by convention; the quotient structure does not change.

The exact criterion is ``not a boundary,'' not merely ``noncontractible.''  A separating simple closed curve on a higher-genus surface can be noncontractible yet homologically trivial because it bounds a subsurface.  If that subsurface is a union of faces, the corresponding string is a stabilizer product.  Homology supplies the invariant that removes the ambiguity.

The quotient in \cref{eq:logical-homology} is summarized by the short exact sequence
\begin{equation}
 0\longrightarrow B_1\longrightarrow Z_1\longrightarrow H_1\longrightarrow0.
 \label{eq:short-exact-homology}
\end{equation}
Exactness at $Z_1$ says precisely that the closed strings representing the zero logical class are the boundaries.

Now let two error histories $e,e'$ have the same endpoint footprint:
\[
 \partial_1e=\partial_1e'.
\]
Over $\Ftwo$, their difference $e+e'$ is a cycle.  If $e+e'\in B_1$, the two histories differ by stabilizers and belong to the same recovery class.  If $[e+e']\neq0$ in $H_1$, they differ by a logical operator.  Topology enters decoding at exactly this point.

\begin{figure}[ht]
\centering
\begin{tikzpicture}[scale=0.95]
  \draw[DeepBlue,line width=1pt,rounded corners=2mm] (0,0) rectangle (4.4,2.6);
  \foreach \x in {0.55,1.65,2.75,3.85}{\draw[MidGray!55] (\x,0)--(\x,2.6);}
  \foreach \y in {0.52,1.3,2.08}{\draw[MidGray!55] (0,\y)--(4.4,\y);}
  \draw[SpiderGreen,line width=3pt,rounded corners=2mm] (0.55,0.52)--(1.65,0.52)--(1.65,1.3)--(0.55,1.3)--cycle;
  \node[SpiderGreen,font=\sffamily\small] at (1.1,-0.35) {face boundary $\in B_1$};
  \draw[Coral,line width=3pt] (0,2.08)--(4.4,2.08);
  \draw[Coral,-{Latex[length=2mm]},line width=1pt] (4.3,2.25) to[bend left=25] (0.1,2.25);
  \node[Coral,font=\sffamily\small] at (2.2,3.05) {closed after edge identification; nonzero in $H_1$};
  \node[font=\sffamily\scriptsize,rotate=90] at (-0.22,1.3) {identified};
  \node[font=\sffamily\scriptsize,rotate=90] at (4.62,1.3) {identified};
  \begin{scope}[xshift=66mm]
    \node[boxnode,text width=18mm] (b) at (0,1.8) {$B_1$\\boundaries};
    \node[boxnode,text width=18mm] (z) at (3.0,1.8) {$Z_1$\\cycles};
    \node[coralnode,text width=18mm] (h) at (6.0,1.8) {$H_1$\\logical classes};
    \draw[flow] (-1.7,1.8)--(-1.0,1.8) node[left,font=\sffamily\small] {$0$};
    \draw[flow] (b)--(z); \draw[flow] (z)--(h); \draw[flow] (h)--++(1.5,0) node[right,font=\sffamily\small] {$0$};
    \node[below=6mm of z,font=\sffamily\small] {$\ker(Z_1\to H_1)=B_1$};
  \end{scope}
\end{tikzpicture}
\caption{Left: a plaquette boundary and a homologically nontrivial closed string.  Right: the exact sequence formalizing stabilizers and logical equivalence classes.}
\label{fig:homology-strings}
\end{figure}
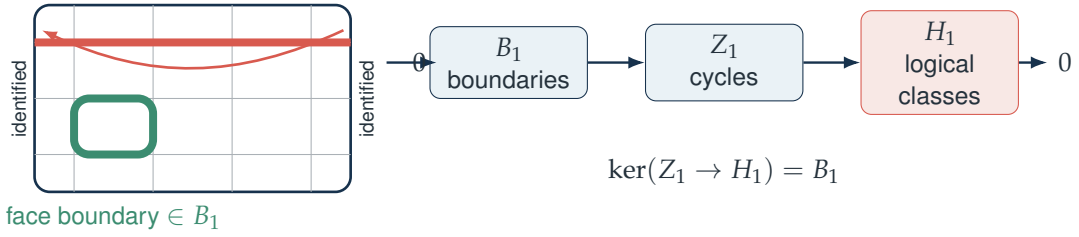

\subsubsection{Counting encoded qubits on a closed surface}

For a connected closed orientable surface of genus $g$,
\[
 \dim H_0=1,\qquad \dim H_1=2g,\qquad \dim H_2=1.
\]
With one qubit on each edge, there is a vertex check and a face check, but the product of all vertex checks and the product of all face checks are each constrained to the identity.  The number of independent stabilizers is therefore
\[
 (V-1)+(F-1)=V+F-2.
\]
Using $V-E+F=2-2g$, the encoded-qubit count is
\begin{equation}
 k=E-(V+F-2)=2g.
 \label{eq:surface-code-k}
\end{equation}
This agrees with the two independent homology directions for each handle.  On a torus, $g=1$ and there are two logical qubits: one may choose primal $Z$ loops around the two basic cycles and dual $X$ loops intersecting them once.

The intersection pairing explains logical anticommutation.  If a primal $Z$ string and a dual $X$ string cross an odd number of times, their Pauli operators anticommute; if the intersection number is even, they commute.  Thus homology supplies the logical supports while the primal--dual intersection form supplies the encoded Pauli algebra.

\subsubsection{A syndrome fibre on a surface}

Let defects appear at two vertices.  Any 1-chain joining the pair has the same boundary and hence the same endpoint syndrome.  Two candidate paths $e$ and $e'$ differ by the closed chain $e+e'$.  There are three possibilities:
\begin{enumerate}
\item $e+e'$ is a sum of face boundaries, so the recoveries are stabilizer equivalent;
\item $e+e'$ is a nontrivial homology class, so the choices differ logically;
\item on a surface with boundary, $e+e'$ is relative and its class depends on which endpoint sets are allowed to absorb strings.
\end{enumerate}

For independent edge noise, minimum-weight perfect matching pairs observed defects using short paths.  The algorithm chooses one representative in the syndrome fibre.  Its logical failure probability is the probability that the chosen representative and the actual error close to a nontrivial cycle.  This is the global event that no collection of local endpoint checks can decide by itself \cite{dennis-et-al}.

\begin{keyidea}[The code-space quotient in three languages]
In stabilizer language, harmless differences lie in $\mathcal S$.  In CSS linear algebra, they lie in a check row space.  In surface topology, they are cellular boundaries.  Logical errors are the corresponding normalizer, kernel, or cycle classes after quotienting out those harmless differences.
\end{keyidea}

\subsection{Boundaries, punctures, and gluing}

For a surface with an allowed boundary region $A\subseteq\Sigma$, relative homology classifies strings permitted to end on $A$.  The short exact sequence of chain complexes
\[
 0\longrightarrow C_\bullet(A)\longrightarrow C_\bullet(\Sigma)
 \longrightarrow C_\bullet(\Sigma,A)\longrightarrow0
\]
induces the long exact sequence
\begin{equation}
 \cdots\longrightarrow H_1(A)\longrightarrow H_1(\Sigma)
 \longrightarrow H_1(\Sigma,A)
 \xrightarrow{\delta} H_0(A)\longrightarrow H_0(\Sigma)\longrightarrow\cdots.
 \label{eq:relative-les}
\end{equation}
The connecting map $\delta$ records endpoint or boundary-component data.  Dual cohomology treats the complementary Pauli boundary condition.  For cutting, code surgery, and gluing, the Mayer--Vietoris long exact sequence tracks which logical classes survive, merge, or become boundaries \cite{hatcher}.

For example, let $A$ contain two disjoint boundary components on which a primal string may end.  A relative 1-cycle can be an arc joining those components.  Its ordinary boundary is nonzero, but that boundary lies in $C_0(A)$ and therefore vanishes in the relative complex.  The connecting map sends the relative class to the difference of its endpoint components in $H_0(A)$.  Such an arc can represent a logical operator even though it is not a closed curve in the ordinary sense.

Punctures can be treated as additional boundary components.  Moving or merging them changes the relative groups and hence the available logical operators.  In lattice surgery, a temporary set of joint checks changes which strings count as boundaries across an interface; measuring the new checks extracts a logical parity.  The long exact sequence is useful because it tracks the change in classes without relying on an ambiguous picture of whether a curve ``looks contractible.''

Gluing also has a circuit interpretation.  Each code patch has a boundary Hilbert space and a set of permitted endpoint data.  Joining patches identifies part of those boundaries and sums over compatible data, much like contracting two diagrammatic gates.  New closed histories may be created by the gluing, while some former boundary-to-boundary logical strings become stabilizers.  Topological exact sequences provide the bookkeeping for this change of type.

\subsection{Distance and the hyperbolic tradeoff}

The distance of a stabilizer code is the minimum weight of a nontrivial logical Pauli.  In a homological surface code it is the shorter of the minimum lengths of a nonzero primal homology class and a nonzero dual class.  A distance-$d$ code corrects arbitrary errors of weight at most $\lfloor(d-1)/2\rfloor$, although stochastic logical performance also depends on the decoder and noise model \cite{gottesman,dennis-et-al}.

Hyperbolic surface codes change the relation among the number $n$ of physical qubits, the number $k$ of logical qubits, and the distance.  For a regular $\{p,q\}$ cellulation of a closed orientable surface, with qubits on edges,
\[
 pF=2E,\qquad qV=2E,
\]
so Euler's formula gives
\begin{equation}
 2-2g=V-E+F=E\left(\frac2q+\frac2p-1\right).
 \label{eq:hyperbolic-euler}
\end{equation}
When $1/p+1/q<1/2$, the tiling is hyperbolic and the coefficient is negative.  The genus $g$, and hence $k=2g$ for the closed surface code, grows linearly with $E=n$.  This yields a nonzero asymptotic encoding rate with bounded-weight local checks.  The price in typical closed hyperbolic families is distance growing only logarithmically with $n$, rather than the square-root scaling familiar from planar Euclidean surface-code families.  The result is a different rate--distance tradeoff, not a universal dominance claim \cite{mahmoud-hqec}.

Rearranging \cref{eq:hyperbolic-euler} gives an explicit rate formula:
\begin{equation}
 k=2g=2+n\left(1-\frac2p-\frac2q\right),
 \qquad
 \frac kn=1-\frac2p-\frac2q+\frac2n.
 \label{eq:hyperbolic-rate}
\end{equation}
For a $\{5,4\}$ family the asymptotic rate is $1/10$; for $\{8,3\}$ it is $1/12$.  These are ideal closed regular-cellulation counts.  A finite experimental graph, an open patch, or a medial construction can have different vertex and edge counts and should not be inserted into \cref{eq:hyperbolic-rate} without checking that it is the same cell complex.

The distance is a combinatorial systole: the shortest nonzero class in the primal or dual first homology.  Negative curvature allows the number of cells to grow exponentially with graph radius, so a loop whose length is proportional to the radius may have length only $O(\log n)$.  The same exponential growth is what permits a linear number of handles and hence constant rate.  Rate and distance are two faces of the same geometry.

Practical code assessment therefore requires at least
\[
 (n,k,d_X,d_Z),
\]
the row weights of $H_X,H_Z$, the distribution of short logical representatives, a decoder, and a noise model.  Equal nominal distance does not guarantee equal logical error rate: the number of near-minimum logical paths and the degeneracy of syndrome fibres also matter.

The Hyperbolic Cycle Basis framework constructs the plaquette cycles, nontrivial cycles, and cocycles algorithmically for periodic hyperbolic lattices, enabling systematic computation of parity checks, logical operators, distances, and decoding benchmarks \cite{mahmoud-hqec}.

A useful computational workflow is:
\begin{enumerate}
\item build the quotient cellulation and its boundary matrices $\partial_2,\partial_1$;
\item row-reduce to compute $Z_1=\ker\partial_1$ and $B_1=\im\partial_2$;
\item select representatives for a basis of $H_1=Z_1/B_1$;
\item compute the dual cocycles and their intersection pairing;
\item search each nonzero class for minimum-weight representatives;
\item benchmark a decoder on sampled syndrome fibres.
\end{enumerate}
The cycle-basis step is not merely visualization.  It supplies check matrices, logical operator matrices, and consistency tests that can be passed to a circuit compiler.

\subsection{Foliation and fault-tolerant measurement-based computation}

A static two-dimensional code becomes a spacetime resource when successive primal and dual Tanner-graph layers are coupled into a three-dimensional cluster state.  This process is called foliation.  Logical 1-cycles and 1-cocycles sweep through the foliation direction to form correlation surfaces, while local check relations become three-dimensional neighbourhood constraints.

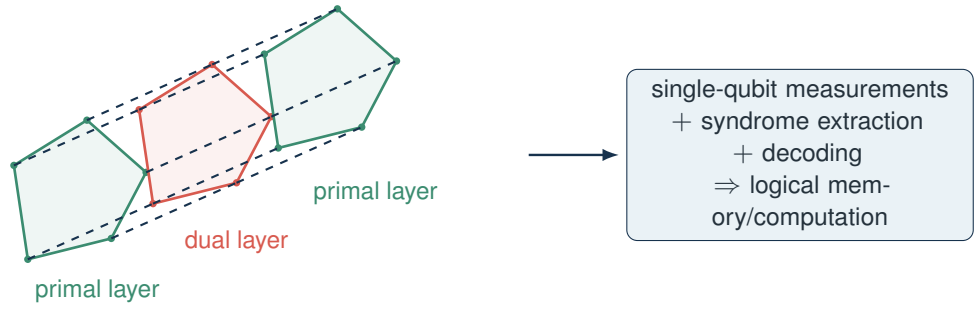
\begin{figure}[ht]
\centering
\begin{tikzpicture}[scale=0.92]
  \foreach \sx/\sy/\col in {0/0/SpiderGreen,1.8/0.8/Coral,3.6/1.6/SpiderGreen}{
    \begin{scope}[shift={(\sx,\sy)}]
      \draw[\col,line width=1pt,fill=\col!8] (0,0) -- (1.2,0.3) -- (1.7,1.25) -- (0.85,2.0) -- (-0.2,1.35) -- cycle;
      \foreach \x/\y in {0/0,1.2/0.3,1.7/1.25,0.85/2.0,-0.2/1.35}{\fill[\col] (\x,\y) circle (1.5pt);}
    \end{scope}
  }
  \foreach \x/\y in {0/0,1.2/0.3,1.7/1.25,0.85/2.0,-0.2/1.35}{
    \draw[DeepBlue,dashed,line width=0.8pt] (\x,\y) -- ++(1.8,0.8);
    \draw[DeepBlue,dashed,line width=0.8pt] ({\x+1.8},{\y+0.8}) -- ++(1.8,0.8);
  }
  \node[font=\sffamily\small,SpiderGreen] at (1.0,-0.45) {primal layer};
  \node[font=\sffamily\small,Coral] at (3.0,0.25) {dual layer};
  \node[font=\sffamily\small,SpiderGreen] at (5.0,0.95) {primal layer};
  \draw[flow] (7.2,1.5)--(8.5,1.5);
  \node[boxnode,text width=43mm] at (11.1,1.5) {single-qubit measurements\\$+$ syndrome extraction\\$+$ decoding\\$\Rightarrow$ logical memory/computation};
\end{tikzpicture}
\caption{Schematic of foliation: alternating code layers are joined along a discrete third direction to form a cluster-state resource.}
\label{fig:foliation}
\end{figure}

Recent hyperbolic cluster-state constructions retain the constant-rate feature of their two-dimensional hyperbolic code layers while exhibiting simulated fault-tolerance thresholds comparable to the Euclidean Raussendorf--Harrington--Goyal construction under a circuit-level depolarizing model \cite{mahmoud-cluster}.  The comparison is model-dependent, but the construction shows that negative curvature can be carried from a static code into a repeated measurement protocol.

\subsubsection{Cluster stabilizers and correlation surfaces}

For a graph $G=(V,E)$, the graph state is prepared by initializing every vertex in $\ket+$ and applying controlled-$Z$ on every edge.  Its stabilizer generators are
\begin{equation}
 K_v=X_v\prod_{u\in N(v)}Z_u.
 \label{eq:graph-state-stabilizer}
\end{equation}
Foliation starts from the Tanner graphs of the code checks and couples alternating primal and dual layers so that products of local cluster stabilizers form extended correlation surfaces.  Measuring most cluster qubits consumes the resource while propagating logical information through the foliation direction.

A data error or a wrong measurement outcome changes the parity of nearby correlation-surface checks.  In a spacetime decoding graph, the locations where parity changes are defect events.  A chain of physical or measurement faults has those events as its boundary.  The decoder again chooses a chain with the observed boundary, and failure occurs when the actual and chosen chains close to a nontrivial spacetime logical class.

This is the repeated-measurement version of the footprint fibre.  A single syndrome round cannot distinguish a data fault from a faulty check outcome.  Differences between consecutive rounds supply temporal boundary information.  At the initial and final time slices, state preparation and readout determine relative boundary conditions, so the relevant classification is naturally relative homology in a three-dimensional complex.

\subsubsection{What a threshold statement requires}

A numerical threshold is not a property of a code graph alone.  It belongs to a family of increasing instances together with a complete protocol: preparation, entangling-gate schedule, measurement order, decoder, and stochastic fault model.  Below threshold, the logical failure probability decreases with growing distance; above threshold it does not.  Changing biased noise, erasure information, measurement latency, or decoder approximation can change the crossing point.

For constant-rate hyperbolic families, one must also state how many logical qubits are being protected and which aggregate failure statistic is reported.  The probability that \emph{any} one of $k=\Theta(n)$ logical qubits fails can scale differently from the failure probability of a selected logical qubit.  The constant-rate advantage therefore makes benchmarking more informative but also more demanding.

\subsection{Problems for Correcting Quantum Circuits}

\begin{enumerate}[label=\textbf{3.\arabic*.},leftmargin=3.3em]
\item \textbf{Pauli span.}  Expand the amplitude-damping Kraus operators in the Pauli basis.  Explain why a code that corrects $I,X,Y,Z$ on one location corrects their linear combinations even though amplitude damping is not a random Pauli channel.

\item \textbf{Repetition syndrome.}  Propagate $\alpha\ket{000}+\beta\ket{111}$ through \cref{fig:repetition-syndrome-circuit} with each of $I,X_1,X_2,X_3$.  Show that the ancilla outcomes are independent of $\alpha,\beta$.

\item \textbf{Two flips.}  List the syndromes of $X_1X_2$, $X_1X_3$, and $X_2X_3$.  Under a minimum-weight single-error decoder, determine the residual logical operation in each case.

\item \textbf{Knill--Laflamme matrix.}  Let two error representatives obey $E_2P=E_1P$.  Compute the associated $2\times2$ block of $C=(c_{ab})$ and diagonalize it.  Interpret the zero eigenvalue.

\item \textbf{CSS quotient.}  For the repetition matrix $H_Z$, compute its kernel and row space over $\Ftwo$.  Identify the logical $X$ class and verify that the zero syndrome alone does not imply a stabilizer.

\item \textbf{Surface-code count.}  Derive \cref{eq:surface-code-k} using both stabilizer ranks and homology.  Repeat for a disconnected cellulation and identify which step must be modified.

\item \textbf{Relative string.}  Draw an annulus with both boundary components able to absorb primal strings.  Compute the relevant part of the long exact sequence and classify an arc joining the two components.

\item \textbf{Hyperbolic parameters.}  Derive \cref{eq:hyperbolic-rate} for a closed regular $\{p,q\}$ cellulation.  Evaluate the limiting rate for $\{5,4\}$, $\{8,3\}$, and the Euclidean boundary cases $\{4,4\}$ and $\{6,3\}$.

\item \textbf{Spacetime decoding.}  In three repeated rounds of a single parity check, suppose the recorded outcomes are $0,1,0$.  Give two distinct fault histories with that footprint and explain which additional assumptions a decoder uses to rank them.
\end{enumerate}

\paragraph{Checks.}  In Problem 3, every two-flip pattern has the syndrome of the complementary one-flip pattern and recovery leaves $X_1X_2X_3$.  In Problem 8, the Euclidean cases have zero asymptotic rate in this closed regular-family formula.  Problem 9 has no unique answer without a noise model; that is the point.

\begin{bridgebox}[From correcting to geometrizing]
A hyperbolic code is first a combinatorial and topological object: qubits, checks, and logical cycles on a bounded-degree graph.  To execute it, a physical system must realize the graph, supply addressable quantum degrees of freedom, implement interactions, and support repeated measurement.  The next section separates the part already achieved experimentally---graph-level spectral emulation---from the nonlinear and fault-tolerant programme still ahead.
\end{bridgebox}

\subsection*{End-of-lecture synthesis: correcting without reading}

Quantum error correction does not identify a damaged codeword and then replace it with a memorized copy.  It extracts a \emph{footprint} of the error that is independent of the logical amplitudes.  The Knill--Laflamme condition makes this independence exact: for every pair of correctable errors, the folded operator $P E_a^\dagger E_b P$ must act as a scalar on the code space.  The scalar matrix may reveal degeneracy---different physical errors can have the same action on encoded states---but it may not reveal the unknown logical state.  A recovery can then distinguish the relevant orthogonal error sectors coherently and reverse their effect.

This operator theorem, its folded string diagram, and stabilizer syndromes are three views of the same constraint.  The theorem handles the full linear span of errors and proves necessity and sufficiency.  The diagram makes local cancellations and logical-state independence visible.  Stabilizer algebra turns those conditions into commuting parity measurements and efficient classical decoding.  A sound argument moves deliberately among all three rather than asking one representation to do every job.

Topological codes add a quotient.  A syndrome is a boundary, while a physical fault is one of many chains with that boundary.  A decoder chooses a plausible representative; recovery succeeds when its difference from the actual chain is a stabilizer and fails when that closed difference carries a nontrivial logical homology class.  Boundaries that absorb strings require relative homology, and repeated noisy measurements promote the same reasoning to a spacetime complex.  Thus ``find the error'' is the wrong instruction: infer an equivalence class using a stated noise model and decoder.

Finally, code parameters never stand alone.  Distance, rate, check weight, locality, syndrome schedule, decoder complexity, and the definition of logical failure form one engineering statement.  Hyperbolic cellulations can support constant asymptotic rate with logarithmic distance, unlike planar Euclidean patches, but this combinatorial advantage does not automatically supply a threshold or a device.  The geometrizing programme begins precisely at that gap: how can a physical platform reproduce the required incidence graph, and which experimentally observed properties certify only connectivity, which certify coherent quantum dynamics, and which would be needed for fault-tolerant operation?

\clearpage
\section{Geometrizing Quantum Circuits}
\label{sec:geometrizing}

\subsection{Geometry as executable connectivity}

A hyperbolic surface code begins as a cellulation of a negatively curved surface.  The physical qubits and checks depend on incidence data: which edge meets which vertex and which face.  A hardware platform does not have to bend its substrate into the hyperbolic plane.  It must implement the same weighted adjacency graph with sufficiently controlled local elements and couplings.

This distinction separates three geometries:

\begin{enumerate}
\item the intrinsic geometry or combinatorics of the target lattice;
\item the coordinate picture used to draw it, such as the Poincar\'e disk;
\item the Euclidean layout of components and wiring on a fabricated chip.
\end{enumerate}

In the Poincar\'e disk, the metric
\begin{equation}
 ds^2=\frac{4\,(dx^2+dy^2)}{(1-x^2-y^2)^2}
 \label{eq:poincare-metric}
\end{equation}
makes cells that look smaller near the boundary intrinsically congruent.  The coordinate contraction is not itself a fabrication instruction.  A graph compiler may instead place resonators wherever routing permits and encode the desired geometry into coupling strengths.

\subsubsection{Regular tilings and curvature}

A regular $\{p,q\}$ tiling has $p$ sides around every face and $q$ faces meeting at every vertex.  Its curvature type is determined by
\[
 \frac1p+\frac1q
 \begin{cases}
 >\frac12,&\text{spherical},\\
 =\frac12,&\text{Euclidean},\\
 <\frac12,&\text{hyperbolic}.
 \end{cases}
\]
This criterion can be understood from the angle budget.  A regular Euclidean $p$-gon has interior angle $(1-2/p)\pi$; fitting $q$ of them around a point requires equality.  In the hyperbolic plane, polygons of a given side count can have smaller angles, allowing more combinatorial branching than flat geometry permits.

An infinite hyperbolic tiling is not a finite code or chip.  There are two common finite constructions.  One may cut out an open patch, which has a large physical boundary, or quotient by a discrete group of translations to identify sides and form a compact surface of genus $g\geq2$.  The quotient preserves local incidence while turning long translation cycles into nontrivial homology classes.  Code distance, boundary access, and experimental routing depend on which construction is chosen.

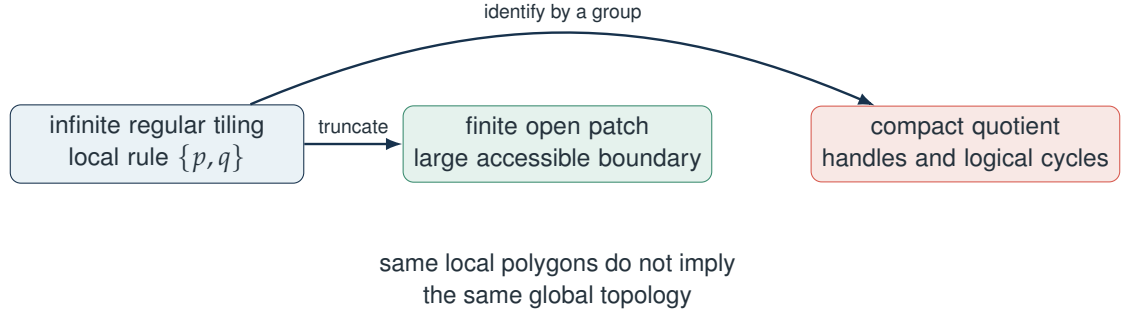
\begin{figure}[ht]
\centering
\begin{tikzpicture}
  \node[boxnode,text width=36mm] (inf) {infinite regular tiling\\local rule $\{p,q\}$};
  \node[greennode,text width=38mm,right=13mm of inf] (patch) {finite open patch\\large accessible boundary};
  \node[coralnode,text width=38mm,right=13mm of patch] (quot) {compact quotient\\handles and logical cycles};
  \draw[flow] (inf)--node[above,font=\sffamily\scriptsize]{truncate} (patch);
  \draw[flow] (inf) to[bend left=23] node[above,font=\sffamily\scriptsize]{identify by a group} (quot);
  \node[below=8mm of patch,font=\sffamily\small,align=center] {same local polygons do not imply\\the same global topology};
\end{tikzpicture}
\caption{Two finite descendants of a hyperbolic tiling serve different goals: boundary experiments favour open patches, while closed surface codes use compact quotients.}
\label{fig:patch-vs-quotient}
\end{figure}

\subsubsection{Graph geometry versus metric geometry}

The unweighted adjacency graph records which sites are neighbours.  A weighted graph adds coupling magnitudes and phases.  A metric model adds a notion of distance, perhaps derived from graph paths or from an embedding in a continuum.  These layers should not be collapsed.  Two Euclidean chip layouts can realize the same weighted graph; two graphs with the same degree sequence can have very different spectra; and a finite graph can approximate selected continuum observables without possessing a unique underlying smooth metric.

For quantum error correction, the cell structure matters because checks are attached to vertices and faces.  For a linear resonator experiment, the one-particle weighted adjacency is the immediate object.  For a continuum field-theory comparison, the graph Laplacian and its scaling with lattice spacing matter.  A compiler must therefore begin by naming which geometric data are intended to survive fabrication.

For linear bosonic modes, a weighted tight-binding model has the form
\begin{equation}
 H=\sum_i\omega_i a_i^\dagger a_i
   +\sum_{i\neq j} t_{ij}a_i^\dagger a_j,
 \label{eq:tight-binding}
\end{equation}
where $t_{ij}$ is nonzero when the compiled graph couples sites $i$ and $j$.  In the single-excitation sector, the Hamiltonian is the weighted adjacency matrix up to the on-site terms.  This is the basic route from geometry to a measurable circuit spectrum.

If all on-site frequencies are equal, write the single-particle matrix as
\[
 h=\omega_0I+T.
\]
An eigenvector $u^{(\nu)}$ of $T$ defines a normal-mode operator
\[
 b_\nu^\dagger=\sum_i u_i^{(\nu)}a_i^\dagger,
\]
with frequency $\omega_0+\lambda_\nu$.  Degeneracies, isolated bands, and mode clusters are therefore spectral properties of the compiled weighted graph.  A local drive at site $j$ couples to mode $\nu$ in proportion to $u_j^{(\nu)}$; a mode can exist while being dark at a chosen port.

For continuum comparisons one often uses a graph Laplacian rather than adjacency:
\[
 L=D-T,
\]
where $D$ is the weighted degree matrix.  Appropriate rescaling and boundary conditions can make low-lying eigenvectors approximate smooth Laplace--Beltrami modes.  Adjacency and Laplacian spectra contain related but not identical information, so the experimental observable must be compared with the operator actually implemented by the circuit.

In linear input--output theory, coupling selected sites to ports gives a frequency-dependent response of the schematic form
\begin{equation}
 S(\omega)=I-iK^\dagger
 \left(\omega I-h+\frac i2KK^\dagger+\frac i2\Gamma\right)^{-1}K,
 \label{eq:scattering-resolvent}
\end{equation}
where $K$ records port couplings and $\Gamma$ internal loss.  Poles of the resolvent locate damped normal modes.  Numerator overlaps determine which poles are visible in a particular transmission coefficient.  \Cref{eq:scattering-resolvent} explains why spectra from several ports are combined: changing the boundary vector changes visibility without changing the underlying eigenvalues.

\subsection{The first superconducting proof of principle}

Koll\'ar, Fitzpatrick, and Houck showed that meandering coplanar-waveguide resonators can realize connectivity unavailable to ordinary planar nearest-neighbour lattices \cite{kollar}.  Resonators acted as sites and engineered intersections and couplers supplied adjacency.  Their device implemented a finite heptagon-kagome line graph and exhibited a spectrally isolated, highly degenerate flat band in microwave transmission.

The result established a crucial principle: a flat superconducting substrate can host an effective hyperbolic graph.  It did not implement quantum error correction, and the line-graph architecture exposed scaling limits.  Long meandering resonators, radial crowding in a disk-like embedding, and a design tailored to one finite graph made it difficult to scale to multiple compact quotients.  Those limitations became specifications for a later architecture: high-quality local resonators, direct graph edges, metric data encoded in couplings rather than physical distance, and a reusable construction across tilings.

\subsubsection{Why a line graph produces a flat band}

Given a graph $G$, its line graph $L(G)$ has one vertex for each edge of $G$; two vertices of $L(G)$ are adjacent when the corresponding edges of $G$ meet.  Let $B$ be an unsigned incidence matrix of $G$.  Then, up to a diagonal term,
\[
 B^{\mathsf T}B=2I+A_{L(G)}.
\]
Vectors in $\ker B$ are therefore eigenvectors of the line-graph adjacency with eigenvalue $-2$.  When the graph has many independent cycles, this kernel can be large, producing a highly degenerate flat band.  Alternating amplitudes around suitable even cycles give a concrete picture: destructive interference prevents leakage into neighbouring line-graph sites.

\begin{figure}[ht]
\centering
\begin{tikzpicture}[scale=0.95]
  \begin{scope}
    \foreach \a in {0,60,...,300}{\coordinate (v\a) at (\a:15mm);}
    \foreach \a/\b in {0/60,60/120,120/180,180/240,240/300,300/0}{\draw[DeepBlue,line width=1pt] (v\a)--(v\b);}
    \foreach \a in {0,60,...,300}{\fill[DeepBlue] (v\a) circle (1.7pt);}
    \node[font=\sffamily\small] at (0,-23mm) {parent cycle $G$};
  \end{scope}
  \draw[flow] (23mm,0)--(37mm,0);
  \begin{scope}[xshift=60mm]
    \foreach \a in {30,90,...,330}{\coordinate (e\a) at (\a:15mm);}
    \foreach \a/\b in {30/90,90/150,150/210,210/270,270/330,330/30}{\draw[Coral,line width=1pt] (e\a)--(e\b);}
    \foreach \a in {30,90,...,330}{\node[resonator,fill=Coral!15] at (e\a) {};}
    \foreach \a/\sgn in {30/+,90/-,150/+,210/-,270/+,330/-}{\node[font=\sffamily\scriptsize] at ($(0,0)!1.25!(e\a)$) {$\sgn$};}
    \node[font=\sffamily\small] at (0,-23mm) {line-graph mode on $L(G)$};
  \end{scope}
  \draw[flow] (83mm,0)--(97mm,0);
  \node[greennode,text width=40mm] at (122mm,0) {alternating amplitudes\\cancel at shared neighbours\\$\Rightarrow$ localized flat-band state};
\end{tikzpicture}
\caption{Schematic of the line-graph mechanism.  The superconducting proof of principle used a much larger hyperbolic parent graph, but the cancellation principle is local.}
\label{fig:line-graph-flat-band}
\end{figure}
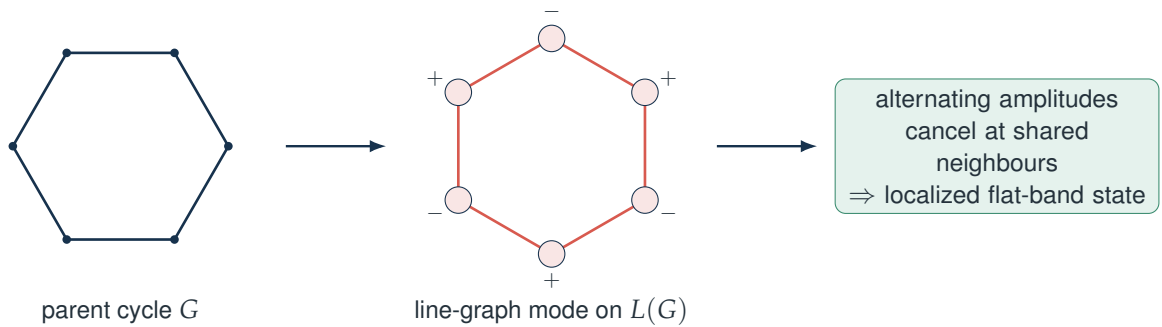

This spectral mechanism is significant in its own right.  It can amplify interaction effects once nonlinearity is added and provides a sharp signature of graph structure.  It should nevertheless not be called a surface-code gap: the degrees of freedom, stabilizers, and logical subspace of a QEC protocol have not yet been introduced.

\subsection{A graph compiler based on direct capacitive coupling}

The later superconducting framework of Xu, Mahmoud, Gorgichuk, Thomale, Rayan, and Mariantoni treats the chip explicitly as a graph compiler \cite{xu-et-al}.  A semi-lumped half-wave coplanar-waveguide resonator represents a vertex.  A designed capacitance represents an edge.  Relative distances in a projected hyperbolic geometry inform the relative coupling strengths, while the substrate remains flat.

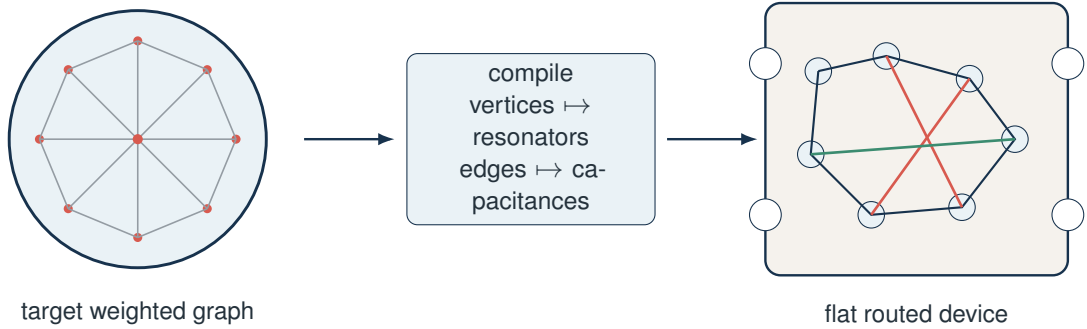
\begin{figure}[ht]
\centering
\begin{tikzpicture}[node distance=8mm and 10mm]
  \begin{scope}[xshift=0mm]
    \draw[DeepBlue,line width=1.1pt,fill=PaleBlue] (0,0) circle (17mm);
    \foreach \a in {0,45,...,315}{
      \coordinate (o\a) at (\a:13mm);
      \fill[Coral] (o\a) circle (1.7pt);
      \draw[MidGray!70,line width=0.6pt] (0,0)--(o\a);
    }
    \fill[Coral] (0,0) circle (1.9pt);
    \foreach \a/\b in {0/45,45/90,90/135,135/180,180/225,225/270,270/315,315/0}{
      \draw[MidGray!70,line width=0.6pt] (o\a)--(o\b);
    }
    \node[font=\sffamily\small] at (0,-23mm) {target weighted graph};
  \end{scope}
  \draw[flow] (22mm,0)--(34mm,0);
  \node[boxnode,text width=30mm] at (52mm,0) {compile\\vertices $\mapsto$ resonators\\edges $\mapsto$ capacitances};
  \draw[flow] (70mm,0)--(82mm,0);
  \begin{scope}[xshift=103mm]
    \draw[DeepBlue,line width=0.9pt,rounded corners=2mm,fill=WarmGray] (-20mm,-18mm) rectangle (20mm,18mm);
    \foreach \x/\y in {-13/9,-4/11,7/8,13/0,6/-9,-6/-10,-14/-2}{
      \node[resonator] (r\x\y) at (\x mm,\y mm) {};
    }
    \draw[DeepBlue,line width=0.8pt] (-13mm,9mm)--(-4mm,11mm)--(7mm,8mm)--(13mm,0mm)--(6mm,-9mm)--(-6mm,-10mm)--(-14mm,-2mm)--(-13mm,9mm);
    \draw[Coral,line width=1.0pt] (-4mm,11mm)--(6mm,-9mm);
    \draw[Coral,line width=1.0pt] (7mm,8mm)--(-6mm,-10mm);
    \draw[SpiderGreen,line width=1.0pt] (13mm,0mm)--(-14mm,-2mm);
    \node[port] at (-20mm,10mm) {}; \node[port] at (-20mm,-10mm) {};
    \node[port] at (20mm,10mm) {}; \node[port] at (20mm,-10mm) {};
    \node[font=\sffamily\small] at (0,-23mm) {flat routed device};
  \end{scope}
\end{tikzpicture}
\caption{Schematic of graph compilation.  Intrinsic adjacency and weights are preserved even though the Euclidean chip layout is freely routed.}
\label{fig:graph-compiler}
\end{figure}

\subsubsection{From a target matrix to circuit parameters}

At the linear level, capacitively coupled resonators are described by a capacitance matrix $C$ and an inductive matrix $L$.  Linearizing the circuit gives a generalized eigenvalue problem whose normal modes depend on both on-site elements and mutual capacitances.  In a weak-coupling rotating-wave description, a mutual capacitance $C_{ij}$ produces a hopping scale approximately proportional to
\[
 t_{ij}\propto C_{ij}\sqrt{\omega_i\omega_j/(C_iC_j)},
\]
with the precise prefactor fixed by the resonator normalization and layout.  The compiler therefore solves an inverse problem: choose realizable geometry so that the extracted circuit matrices reproduce the desired effective $h$ after fabrication tolerances are included.

Semi-lumped resonators help separate site size from electromagnetic wavelength.  A resonator combines distributed coplanar-waveguide sections with concentrated capacitive or inductive features, allowing many sites to fit on a chip without using the long meanders of the first architecture.  Direct capacitive edges also make the local coupling vocabulary reusable: a new target graph changes the netlist and routing more than the underlying physical principle.

A realistic compiler must address several coupled constraints:
\begin{enumerate}
\item \textbf{frequency placement}: keep on-site disorder below the spectral features of interest;
\item \textbf{coupling range}: realize weak and strong $t_{ij}$ values without parasitic modes;
\item \textbf{routing}: avoid unintended crossings and long-range capacitances;
\item \textbf{port visibility}: place drives so that important mode subspaces have nonzero overlap;
\item \textbf{fabrication variation}: include linewidths and component tolerances in the predicted response;
\item \textbf{quantum upgrade path}: reserve locations for nonlinear elements, readout, and control if the graph is to move beyond linear spectroscopy.
\end{enumerate}

The distinction between intrinsic and fabrication geometry is what makes this possible.  A strong hyperbolic coupling need not connect physically distant points on the chip through a long resonator path; the sites can be brought close in the Euclidean layout and assigned the desired capacitance directly.

The devices are cooled to approximately $10\,\mathrm{mK}$, where the resonators are quantized modes.  In the reported high-power transmission measurement, however, the response is consistently interpreted in the linear regime.  A coherent microwave tone is swept through an input port and a complex scattering parameter $S_{j1}(\omega)$ is measured at an output.  Resonance peaks identify driven normal modes; clusters and suppressed regions can be compared with the finite tight-binding spectrum, including disorder and linewidth effects.

The use of several input ports matters.  A mode weakly supported at one port may be invisible to that channel, so the spectra from $S_{21}$, $S_{31}$, and $S_{41}$ are aggregated to expose a larger part of the mode structure.  This is a graph-level validation of \cref{eq:tight-binding}; it is not yet a syndrome-extraction experiment.

The measured quantity is complex.  Its magnitude identifies transmission peaks and stop bands, while its phase contains propagation and interference information.  Cable delay, attenuation, background transmission, and impedance mismatch must be calibrated or modelled before comparing fine structure with the circuit resolvent.  A peak width combines external port loading and internal loss; it is not simply fabrication disorder.

Cooling the circuit below the photon energy scale ensures that a quantum description of each mode is available, but a high-power coherent drive populates many photons and can make the measured response well approximated by classical linear equations.  This is not a contradiction: linear classical waves and single-particle quantum amplitudes obey the same normal-mode equation.  To establish nonclassical state control one needs few-photon calibration and observables that cannot be reproduced by a coherent Gaussian field.

Mode counting also needs care.  A finite graph with $N$ ideal sites has $N$ single-particle eigenmodes, but a given $S_{j1}$ trace need not show $N$ resolved peaks.  Degeneracy combines modes at one frequency; loss merges nearby modes; a port may be orthogonal to an eigenspace; and disorder may split a theoretical degeneracy below the measurement resolution.  Agreement is therefore assessed at the level of the global clustered spectrum and modelled line shapes, not by forcing a one-to-one assignment to every ideal eigenvector.

\subsection{Two hyperbolic lattices and two global topologies}

The $\{8,3\}$ sample realizes a finite sublattice with nine octagonal faces, $48$ vertices represented by resonators, and $56$ edges represented by direct capacitive couplings.  Its unit-cell geometry compactifies on a genus-$2$ surface.  The calculated density of states has three large clusters separated by two prominent central gaps, with smaller gaps nearer the band edges.  The measured transmission reproduces the principal cluster structure, although not every predicted eigenmode is visible.

The $\{12,4\}$ sample targets a genus-$3$ geometry with five dodecagonal faces, $56$ resonators, and $60$ direct couplings.  Its calculated spectrum forms five clusters separated by four gaps.  Highly degenerate eigenstates align particularly clearly with measured peak clusters.  The two chips therefore use one local fabrication vocabulary to realize different global cell structures and different spectral organizations \cite{xu-et-al}.

\begin{figure}[ht]
\centering
\begin{tikzpicture}
  \node[boxnode,text width=50mm,minimum height=29mm] (a) at (0,0) {};
  \node[font=\sffamily\bfseries,DeepBlue] at ($(a.north)+(0,-4mm)$) {$\{8,3\}$ / genus $2$};
  \node[font=\sffamily\small,align=center] at ($(a.center)+(0,2mm)$) {9 octagonal faces\\48 resonators $\cdot$ 56 couplings};
  \begin{scope}[shift={($(a.south west)+(5mm,4mm)$)}]
    \foreach \x/\h in {0/4,1/7,2/10,3/6,4/3,7/4,8/8,9/11,10/7,11/4}{
      \draw[DeepBlue,line width=1.2pt] (\x*3mm,0)--++(0,\h*0.55mm);
    }
    \draw[Coral,line width=2.3pt] (15mm,0)--(18mm,0);
  \end{scope}
  \node[boxnode,text width=50mm,minimum height=29mm] (b) at (70mm,0) {};
  \node[font=\sffamily\bfseries,DeepBlue] at ($(b.north)+(0,-4mm)$) {$\{12,4\}$ / genus $3$};
  \node[font=\sffamily\small,align=center] at ($(b.center)+(0,2mm)$) {5 dodecagonal faces\\56 resonators $\cdot$ 60 couplings};
  \begin{scope}[shift={($(b.south west)+(5mm,4mm)$)}]
    \foreach \x/\h in {0/8,1/11,2/6,5/5,6/10,7/7,10/6,11/12,12/8,15/5,16/10,17/7}{
      \draw[DeepBlue,line width=1.2pt] (\x*2mm,0)--++(0,\h*0.5mm);
    }
    \foreach \x in {6mm,14mm,22mm,30mm}{\draw[Coral,line width=2.3pt] (\x,0)--++(3mm,0);}
  \end{scope}
  \node[below=7mm of a,font=\sffamily\small,align=center] {three principal mode clusters\\with prominent central gaps};
  \node[below=7mm of b,font=\sffamily\small,align=center] {five mode clusters\\separated by four gaps};
\end{tikzpicture}
\caption{Schematic comparison of the two hyperbolic devices.  The miniature spectra indicate clustered organization only; they are not reproductions of measured data.}
\label{fig:device-comparison}
\end{figure}
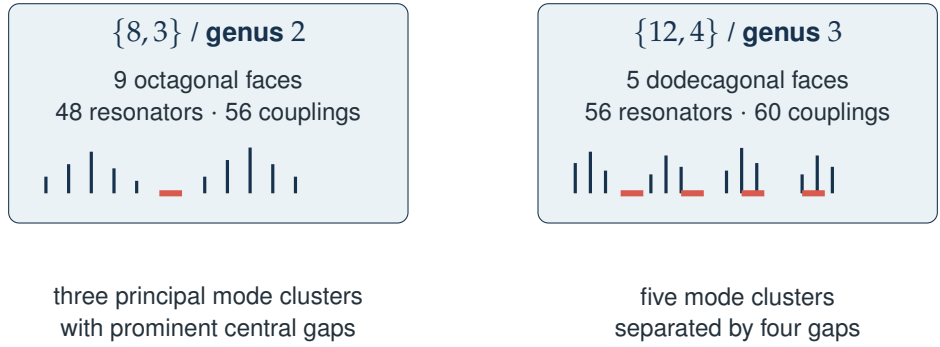

The numbers in \cref{fig:device-comparison} describe the reported resonator netlists and their associated hyperbolic unit-cell construction.  They should not be substituted into the closed surface-code incidence identities $pF=2E$ and $qV=2E$ from \cref{sec:correcting}.  The resonator graph, the polygonal unit-cell bookkeeping, and a QEC cellulation place degrees of freedom on different combinatorial objects.  This is an instructive instance of why the type of a geometry matters.

The two chips test scalability in more than raw component count.  They change $p$, $q$, genus, face organization, degeneracy pattern, and number of principal spectral gaps while retaining the same direct-coupling fabrication logic.  A one-off demonstration could accidentally depend on a special layout.  Reproducing distinct target spectra with a shared compiler is evidence that the architecture is portable across a family of hyperbolic graphs.

The compact quotient viewpoint also explains why genus is relevant even in a linear experiment.  Side identifications constrain long cycles and the allowed mode patterns.  Those global constraints affect the finite adjacency spectrum.  Spectral sensitivity to the quotient does not imply many-body topological order, but it does show that the device response retains information beyond local degree.

This result is topology-sensitive in a specific sense: distinct finite weighted graphs associated with different compact hyperbolic data produce distinct normal-mode spectra, and the fabrication framework resolves that distinction.  It is not a measurement of topological order.

\subsection{A spectral gap is not a code gap}

The word \emph{gap} is overloaded.  In the present devices it refers to separation between clusters of normal modes of a finite, essentially linear, weighted graph.  Such a gap is visible in microwave transmission and can be compared with the eigenvalue differences of \cref{eq:tight-binding}.

A code or many-body protection gap would instead separate a protected interacting subspace from excitations in a Hamiltonian whose structure enforces quantum order.  Demonstrating it would require suitable state preparation, calibrated quantum interactions, observables sensitive to many-body coherence, and ultimately logical or syndrome-level tests.

\begin{center}
\begin{tabularx}{0.94\textwidth}{@{}>{\bfseries\sffamily}lXX@{}}
\toprule
 & Normal-mode spectral gap & Quantum code/protection gap \\
\midrule
object & finite weighted adjacency problem & interacting encoded many-body system \\
observable & transmission peaks and suppressed bands & state, correlations, syndrome, logical lifetime \\
present status & demonstrated in the linear devices & not demonstrated on these devices \\
what it validates & graph compilation and spectral control & protected quantum information \\
\bottomrule
\end{tabularx}
\end{center}

The distinction is constructive, not dismissive.  A validated geometric substrate removes one uncertainty from the harder quantum experiment: whether the intended graph has actually been fabricated.

\subsection{Nonlinearity turns modes into quantum resources}

Linear resonators have equally spaced levels.  A Josephson element introduces anharmonicity, schematically
\begin{equation}
 H_{\mathrm{nl}}=\omega a^\dagger a+\frac{K}{2}a^{\dagger2}a^2,
 \label{eq:kerr}
\end{equation}
so the transition frequency depends on occupation.  The Kerr scale $K$ enables state-selective control, photon blockade, and entangling dynamics.  Selected transmon-like elements or nonlinear resonators can therefore turn a spectrally validated hyperbolic bath into a quantum simulator with controllable interactions.

The nonlinearity originates from the Josephson energy
\[
 -E_J\cos\varphi
 =-E_J+\frac{E_J}{2}\varphi^2-\frac{E_J}{24}\varphi^4+\cdots.
\]
The quadratic term contributes an effective inductance; the quartic term makes the oscillator anharmonic.  After quantization and a rotating-wave approximation, the quartic contribution produces the Kerr term in \cref{eq:kerr}.  Distributing a Josephson element across a normal mode makes $K$ depend on that mode's participation in the nonlinear component.  Circuit QED engineering is therefore an exercise in shaping both the linear eigenvectors and their nonlinear participation \cite{blais-cqed}.

Several regimes are possible.
\begin{itemize}
\item With $|K|$ much smaller than linewidths, the device remains effectively linear.
\item With resolvable anharmonicity, selected modes can serve as qubits or number-sensitive resonators.
\item With many weakly nonlinear sites, the graph approaches an interacting Bose--Hubbard model on hyperbolic connectivity.
\item With strongly nonlinear two-level sites and engineered exchange, it approaches a spin model whose frustration and propagation inherit the graph geometry.
\end{itemize}

Adding Josephson elements is not a local afterthought.  They introduce frequency crowding, flux sensitivity, extra loss channels, control lines, calibration overhead, and interactions whose effective range is set by the normal modes.  The validated linear compiler is valuable because it provides a baseline against which those nonlinear shifts can be measured.

An intermediate programme proposed by Bienias and collaborators couples a small number of superconducting qubits to a hyperbolic photonic lattice \cite{bienias}.  Qubit relaxation can probe the bath, curvature limits the size of photon--qubit bound states, and photon-mediated interactions follow hyperbolic geodesics and can generate frustrated spin models.  This lies between linear spectroscopy and full error correction: add a few nonlinear quantum probes before attempting a complete encoded architecture.

For a qubit of frequency $\omega_q$ coupled locally to site $j$ with strength $g$, the self-energy samples the local Green function
\[
 G_{jj}(\omega)=\bra j(\omega I-h+i0^+)^{-1}\ket j.
\]
Its imaginary part controls radiative decay into available modes, while its real part produces a Lamb shift.  In a gap, a qubit excitation can hybridize with an exponentially localized photonic cloud to form a bound state.  With two qubits, exchanging virtual photons through $G_{ij}(\omega_q)$ mediates an interaction whose magnitude and sign depend on the hyperbolic propagation paths.

\subsection{What evidence would establish the quantum regime?}

Different scientific claims require different observables.  A useful experimental ladder is:

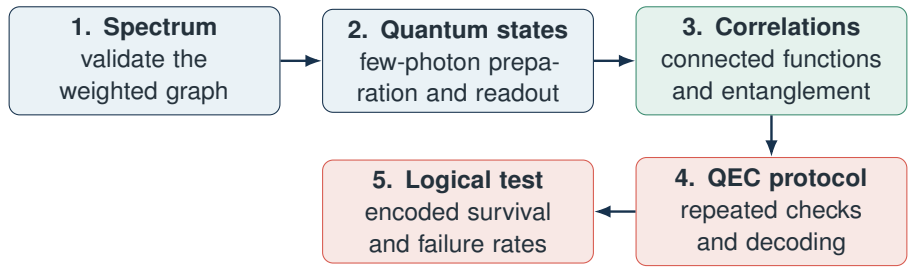
\begin{figure}[ht]
\centering
\begin{tikzpicture}[node distance=5.5mm]
  \node[boxnode,text width=33mm] (s1) {\textbf{1. Spectrum}\\validate the weighted graph};
  \node[boxnode,text width=33mm,right=of s1] (s2) {\textbf{2. Quantum states}\\few-photon preparation and readout};
  \node[greennode,text width=33mm,right=of s2] (s3) {\textbf{3. Correlations}\\connected functions and entanglement};
  \node[coralnode,text width=33mm,below=of s3] (s4) {\textbf{4. QEC protocol}\\repeated checks and decoding};
  \node[coralnode,text width=33mm,left=of s4] (s5) {\textbf{5. Logical test}\\encoded survival and failure rates};
  \draw[flow] (s1)--(s2); \draw[flow] (s2)--(s3); \draw[flow] (s3)--(s4); \draw[flow] (s4)--(s5);
\end{tikzpicture}
\caption{Increasingly strong claims require increasingly informative observables.  The present devices establish the first stage.}
\label{fig:experimental-ladder}
\end{figure}

State tomography can test preparation on a small subsystem.  Joint qubit or photon readout gives two-site correlations.  Connected higher-point functions reveal interactions that a Gaussian linear theory cannot supply.  Randomized measurements, classical shadows, or two-copy interference can estimate subsystem purities and R\'enyi entropies.  Repeated ancilla readout gives syndrome histories.  Finally, encoded survival or process tomography can compare logical and physical error rates.

For error correction, the decisive data are not attractive spectra but repeated syndromes and logical failure probabilities under a declared noise model and decoder.  For holographic or bulk--boundary questions, the decisive data are connected boundary correlators, entanglement structure, and reconstruction properties.

\subsubsection{Few-photon and nonclassical tests}

The second-order coherence
\[
 g^{(2)}(0)=
 \frac{\langle a^\dagger a^\dagger aa\rangle}
      {\langle a^\dagger a\rangle^2}
\]
distinguishes Poissonian coherent light from antibunched or bunched fields.  Photon blockade can give $g^{(2)}(0)<1$, evidence that the Josephson nonlinearity is acting at the level of individual excitations.  Spectroscopy of the $0\to1$ and $1\to2$ transitions measures anharmonicity directly.

Entanglement requires a joint statement.  For two qubits, tomography reconstructs $\rho_{AB}$ and permits concurrence, negativity, or a Bell-state fidelity witness.  For larger lattices, full tomography is exponential.  Random local measurements can estimate selected Pauli correlators, subsystem purities, or classical-shadow observables.  An entanglement witness must exceed a separable bound with uncertainty and readout errors included.

A linear network driven by coherent tones produces a Gaussian coherent state whose connected normally ordered correlations factorize.  Observing a complicated spatial intensity pattern is therefore not by itself evidence of entanglement.  Non-Gaussian state preparation, sub-Poissonian statistics, or a certified inseparability inequality supplies the missing distinction.

\subsubsection{From nonlinear simulation to QEC}

Even an interacting entangled device is not automatically an error-correcting memory.  A QEC experiment needs a code-space preparation, a schedule of commuting checks, ancillas or measurement modes, repeated rounds, a decoder, and a logical observable.  The natural sequence is incremental:
\begin{enumerate}
\item calibrate one nonlinear site and its readout;
\item entangle a small pair or plaquette;
\item measure one stabilizer without resolving its logical basis states;
\item repeat a compatible set of checks and decode injected faults;
\item compare an encoded lifetime or process fidelity with an unencoded baseline;
\item scale across the hyperbolic graph while tracking correlated errors and boundary effects.
\end{enumerate}
The graph compiler addresses the connectivity bottleneck at the bottom of this ladder.  The remaining steps require a control and measurement compiler layered on top of it.

\subsection{Why the boundary is unusually important}

In hyperbolic geometry, volume grows exponentially with radius and the boundary of a finite region remains large relative to its bulk.  Boundary drive and readout can therefore retain detailed information about propagation through the interior.  This makes hyperbolic lattices attractive for tabletop studies of bulk--boundary correspondence.

Boettcher and collaborators showed how discrete hyperbolic circuit lattices can approximate continuum field theory on the Poincar\'e disk in a long-wavelength regime \cite{boettcher}.  Dey and collaborators proposed and simulated a nonlinear electrical-circuit protocol in which two- and three-point boundary functions recover aspects of a holographic conformal field theory \cite{dey-et-al}.  In a linear theory, connected higher-point functions vanish or factorize.  Nonlinearity supplies the bulk interaction vertex needed for a genuine connected three-point response.

These are tabletop tests of selected features of bulk--boundary correspondence, not realizations of quantum gravity or full AdS/CFT.  The discipline is the same as in the spectral-gap discussion: state exactly which observable has been reproduced and which interpretation it supports.

\subsubsection{Exponential growth keeps the boundary extensive}

In curvature radius $\ell$, the circumference and area of a hyperbolic disk of radius $R$ are
\[
 C(R)=2\pi\ell\sinh(R/\ell),
 \qquad
 A(R)=2\pi\ell^2\bigl(\cosh(R/\ell)-1\bigr).
\]
Both grow as $e^{R/\ell}$, so the boundary-to-volume ratio approaches a nonzero constant rather than vanishing.  A finite hyperbolic lattice consequently has a macroscopic fraction of sites near its boundary.  This is a challenge for bulk thermodynamics and an opportunity for boundary control.

If sources $J_i$ are applied at boundary ports, derivatives of the response with respect to those sources generate correlation functions.  A linear network has a response linear in $J$ and therefore no connected three-point susceptibility.  A weak bulk nonlinearity produces terms schematically
\[
 \langle O_iO_jO_k\rangle_c
 \sim \sum_{v\in\mathrm{bulk}}
 G_{iv}G_{jv}G_{kv}\,\lambda_v,
\]
where $G_{iv}$ propagates from bulk interaction site $v$ to boundary point $i$ and $\lambda_v$ is the nonlinear vertex.  This factorized bulk sum is the discrete analogue of a tree-level Witten diagram.

For a boundary conformal theory, two-point functions of a primary operator of dimension $\Delta$ scale with boundary separation according to a power law in an appropriate conformal coordinate.  On a finite discrete lattice, one tests a regulated version over a limited range and compares several system sizes or graph refinements.  Agreement with one fitted exponent is suggestive; a stronger test checks multiple operators, three-point structure, covariance under boundary reparameterization, and robustness to disorder.

\subsection{Holographic codes reconnect geometry and error correction}

The HaPPY construction tiles a hyperbolic disk with perfect tensors \cite{pastawski}.  Bulk legs are logical inputs and open boundary legs are physical outputs.  The tensor network defines an isometric encoder, while a bulk operator may have several boundary reconstructions.  This makes quantum error correction part of the geometry of the bulk--boundary map.

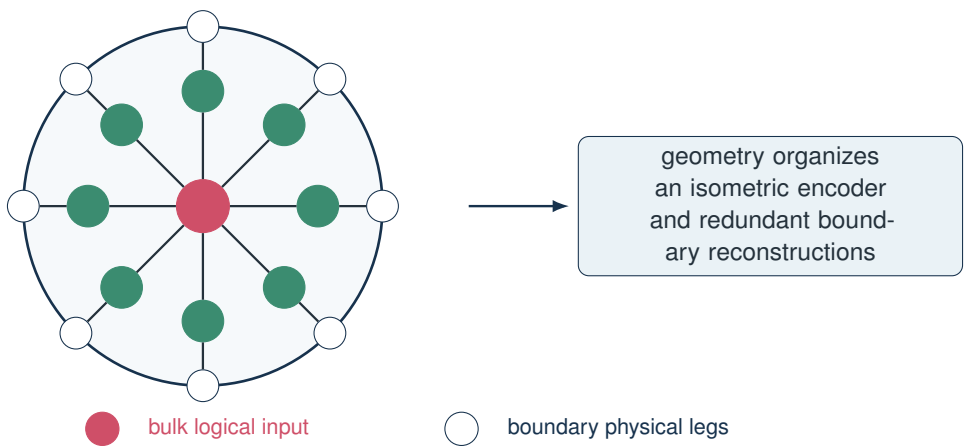
\begin{figure}[ht]
\centering
\begin{tikzpicture}[scale=0.95]
  \draw[DeepBlue,line width=1.0pt,fill=PaleBlue!45] (0,0) circle (25mm);
  \foreach \a in {0,45,...,315}{
    \node[port] (p\a) at (\a:25mm) {};
    \node[zspider,minimum size=5.5mm] (s\a) at (\a:16mm) {};
    \draw[wire] (s\a)--(p\a);
  }
  \node[xspider,minimum size=7mm] (c) at (0,0) {};
  \foreach \a in {0,45,...,315}{\draw[wire] (c)--(s\a);}
  \draw[flow] (37mm,0)--(51mm,0);
  \node[boxnode,text width=48mm] at (79mm,0) {geometry organizes an isometric encoder\\and redundant boundary reconstructions};
  \node[xspider,minimum size=4.4mm] at (-14mm,-31mm) {};
  \node[font=\sffamily\footnotesize,SpiderRed,anchor=west] at (-9mm,-31mm) {bulk logical input};
  \node[port,minimum size=4.4mm] at (36mm,-31mm) {};
  \node[font=\sffamily\footnotesize,DeepBlue,anchor=west] at (41mm,-31mm) {boundary physical legs};
\end{tikzpicture}
\caption{Schematic of a hyperbolic tensor-network encoder.  It is conceptual and does not depict a particular HaPPY tiling.}
\label{fig:holographic-encoder}
\end{figure}

A tensor with $2m$ legs is \emph{perfect} when every bipartition into $m$ input and $m$ output legs defines an isometry up to normalization.  This strong condition lets a tensor be reinterpreted in many directions, echoing the leg-bending flexibility of \cref{sec:drawing} but now with an isometry property across every balanced cut.  Contracting perfect tensors in a negatively curved tiling gives an encoding map from selected bulk legs to open boundary legs.

The greedy reconstruction rule starts from a boundary region and repeatedly absorbs any tensor for which at least half of the legs are already controlled.  Bulk operators inside the resulting entanglement wedge can be pushed to that boundary region.  If a complementary boundary region is erased, an operator remains recoverable whenever a representative can be pushed entirely into the surviving region.  This is an erasure-correction statement: geometry organizes redundant reconstructions of the same logical operator.

The tensor network also satisfies a discrete Ryu--Takayanagi-like entropy relation in an appropriate code subspace: a minimal cut counts entangled bonds, with bulk contributions inside the wedge.  It is a toy model, not a derivation of continuum gravity, but it makes the compatibility among isometry, redundancy, boundary regions, and hyperbolic growth unusually explicit.

Diagrammatic tools can analyze special choices of perfect tensors.  Stabilizer perfect tensors admit Pauli pushing rules that can be represented by Clifford ZX diagrams.  Yet the two frameworks answer different questions: perfection is a many-leg isometry condition, while a ZX spider has a specific basis-copying tensor and generally is not perfect at arbitrary arity.

A hyperbolic superconducting chip does not automatically implement a HaPPY code.  What it supplies is a controlled geometry on which tensor-network encoders, quantum checks, and boundary probes might be specified.  The scientific task is to design the interactions and measurements that turn graph adjacency into an actual encoded isometry or repeated correction protocol.

\subsection{Co-designing a hyperbolic quantum experiment}

Suppose the target is no longer only a spectrum but a protected logical process on a hyperbolic architecture.  The design must pass information in both directions between mathematics and hardware.

From the code side, one exports a cell complex, check supports, logical representatives, a syndrome schedule, and tolerable error correlations.  From the device side, one exports available frequencies, coupling strengths, nonlinearities, port and readout locations, coherence times, fabrication constraints, and a crosstalk model.  A compiler that accepts only an abstract adjacency matrix is insufficient because two edges of the same graph may have very different control costs.

\begin{center}
\begin{tabularx}{0.96\textwidth}{@{}>{\sffamily\bfseries}lXX@{}}
\toprule
layer & input specification & required validation \\
\midrule
topology & quotient cellulation, boundaries, logical cycles & ranks, homology, primal--dual intersections \\
code & qubit placement, checks, distance, decoder & simulated logical failure under a noise model \\
schedule & native gates, check ordering, resets & fault propagation and circuit-level threshold \\
linear device & on-site frequencies and weighted couplings & multiport spectrum and mode shapes \\
nonlinear control & Josephson participation, drives, readout & few-photon calibration and gate/process fidelity \\
logical experiment & encoded preparation and observables & repeated syndromes and logical benchmark \\
\bottomrule
\end{tabularx}
\end{center}

The validation arrows are not all forward.  If one check has excessive weight for the native connectivity, the code cellulation may need to change.  If a port is dark to a required normal-mode subspace, the boundary partition or layout may need to change.  If fabrication disorder closes a designed spectral separation, the target coupling pattern may need larger margins.  Co-design means iterating these constraints before the full chip is fabricated.

\subsubsection{A scale hierarchy}

Several energy and time scales must be ordered.  A schematic hierarchy for nonlinear resonators is
\[
 \text{control bandwidth},\ \kappa,\ \gamma
 \ll |K|,|t_{ij}|,\text{ spectral separations}
 \ll \omega_i,
\]
where $\kappa$ is photon loss and $\gamma$ a qubit decoherence rate.  The precise ordering depends on the experiment: resolving number states needs $|K|\gg\kappa$, while an adiabatic many-body preparation requires a drive slow relative to a minimum gap but fast relative to decoherence.

For repeated QEC, cycle time is another scale.  Checks must be measured faster than errors accumulate, but very fast control can introduce leakage and crosstalk.  Decoder latency matters when feed-forward is real-time.  A hyperbolic graph has bounded local degree in the abstract, yet Euclidean routing and shared control hardware can create nonlocal scheduling conflicts that lengthen the cycle.

Finite-size physics also sets a hierarchy.  The graph radius must be large enough to separate bulk behaviour from boundary effects for a continuum study, but boundary effects are the signal in a holographic experiment.  A compact quotient removes a physical edge but introduces short nontrivial cycles.  The same finite size can therefore be a bug or a feature depending on the observable.

\subsubsection{Three concrete first nonlinear experiments}

\paragraph{A local bound-state probe.}  Couple one tunable qubit to a selected resonator, sweep the qubit through a spectral gap, and measure relaxation and dressed-state spectroscopy.  The claim is curvature-dependent localization of a photon--qubit bound state.  The comparison requires the calibrated local Green function and a Euclidean or graph-randomized control.

\paragraph{A two-qubit mediated interaction.}  Place two qubits at sites with several hyperbolic paths between them.  Detune them into a dispersive regime and measure coherent exchange or a conditional phase.  Vary their graph separation while keeping local fabrication parameters comparable.  The claim is that $G_{ij}(\omega)$, rather than Euclidean chip distance, organizes the interaction.

\paragraph{A nonlinear boundary three-point function.}  Introduce one calibrated nonlinear site or a distributed weak nonlinearity, drive two boundary sources, and measure intermodulation at a third port.  Subtract the linear background and compare the connected signal with the discrete bulk Green-function sum.  The claim is a controlled interacting bulk contribution to boundary response, not yet a quantum-gravity simulation.

Each experiment builds on the linear spectral baseline and adds one qualitatively new observable.  This staged strategy localizes discrepancies: a failed nonlinear prediction can be separated from a failed graph compilation.

\subsection{A chip as a multi-input/multi-output gate}

The present superconducting samples have four microwave ports.  At the level of linear response, the device is already a multi-input/multi-output map described by a scattering matrix.  One may partition the ports as two inputs and two outputs, or one input and three outputs, and compare the resulting arity with the typed gates of \cref{sec:drawing}.  This does not make the chip a ZX spider: a spider obeys specific algebraic equations and basis-copying semantics.  It does suggest an experimental programme in which the \emph{whole process}, rather than only a stored state, is treated as the object to be protected.

That viewpoint aligns with field-theoretic quantum error correction.  A circuit history occupies a region, its ports and cuts define boundaries, errors are interior defect histories, and measured data are footprints.  Protecting a gate means ensuring that all correctable interior histories compatible with the observed boundary data act on the logical interface in the prescribed way.

At a fixed frequency, a four-port linear device has
\[
 \bm b_{\mathrm{out}}=S(\omega)\bm b_{\mathrm{in}}.
\]
Selecting two drive ports and two monitored outputs extracts a $2\times2$ sub-block; choosing one input and three outputs extracts a column.  These are operational partitions of the same scattering process, not changes to the internal graph.  Reciprocity, passivity, and loss impose constraints on $S$; an ideal lossless reciprocal network has a unitary symmetric scattering matrix after all channels are included.

The analogy with a quantum gate becomes exact only in a declared encoded subspace and quantum regime.  A coherent-state scattering matrix describes first moments of bosonic fields.  A quantum channel also specifies noise and higher moments.  Process tomography would prepare a spanning set of few-photon inputs and reconstruct the output channel, while a protected-process experiment would encode those inputs and show that logical process error decreases under correction.

This suggests a process-level Knill--Laflamme question.  Instead of preserving a stored identity channel, one may ask whether a desired isometry $W:L_{\mathrm{in}}\to L_{\mathrm{out}}$ is implemented despite a family of interior faults.  Folding two faulty histories against $W$ should again leave a logical-state-independent comparison on the interface.  The geometry of ports and cuts then becomes part of the code specification rather than merely experimental plumbing.

\subsection{Synthesis: from diagram to device}

The complete path can now be stated without collapsing its stages:

\begin{enumerate}
\item \textbf{Specify the process.}  Identify logical inputs, desired outputs, measurements, and allowed noise.
\item \textbf{Draw and simplify.}  Use typed diagrams and sound local rules to expose basis structure, locality, and compositional boundaries.
\item \textbf{Certify correction.}  Apply Knill--Laflamme algebraically or through folded diagrams; identify footprint fibres and logical ambiguities.
\item \textbf{Choose the geometry.}  Compute checks and logical cycles from a cellulation, including relative boundary data when required.
\item \textbf{Compile the graph.}  Map vertices and weighted edges to physical modes and couplings while respecting routing and control constraints.
\item \textbf{Validate by levels of evidence.}  Establish the spectrum first, then quantum state control, interactions and entanglement, repeated syndromes, and finally logical performance.
\end{enumerate}

\subsection{Problems for Geometrizing Quantum Circuits}

\begin{enumerate}[label=\textbf{4.\arabic*.},leftmargin=3.3em]
\item \textbf{Curvature criterion.}  Derive the condition $1/p+1/q\lessgtr1/2$ from the angle sum around a vertex.  Classify $\{3,3\}$, $\{4,4\}$, and $\{7,3\}$.

\item \textbf{Coordinate size.}  Integrate the metric \cref{eq:poincare-metric} along a radial segment to show that the unit-circle boundary is at infinite intrinsic distance.  Explain why cells shrink in Euclidean coordinates near the edge of the disk.

\item \textbf{Normal modes.}  Diagonalize a three-site path with equal nearest-neighbour hopping $t$.  Determine which modes are dark when a port couples only to the middle site and compare with \cref{eq:scattering-resolvent}.

\item \textbf{Line graph.}  Compute $B^{\mathsf T}B$ for a four-cycle and verify the $-2$ line-graph eigenvector with alternating amplitudes.  What changes for an odd cycle?

\item \textbf{Disorder.}  Add independent on-site shifts $\delta_i$ to a degenerate adjacency matrix.  Use first-order degenerate perturbation theory to predict the splitting.  Which port overlaps would be required to resolve all split modes?

\item \textbf{Spectral versus code gap.}  Give a minimal list of additional Hamiltonian terms, state preparations, and observables required to turn a measured normal-mode gap into evidence for an encoded many-body subspace.  Mark which items are platform dependent.

\item \textbf{Josephson expansion.}  Quantize the quadratic part of $-E_J\cos\varphi$ together with a charging energy and show how the quartic term produces a negative anharmonicity in the transmon regime.  Identify the approximation used in dropping non-number-conserving terms.

\item \textbf{Boundary growth.}  Compute $C(R)/A(R)$ for a hyperbolic disk and its $R\to\infty$ limit.  Compare with a Euclidean disk.  Relate the result to the fraction of addressable boundary sites in a finite lattice.

\item \textbf{Three-point response.}  Starting from a cubic bulk nonlinearity, derive the leading source dependence of a connected boundary three-point function in terms of three bulk-to-boundary Green functions.  State why it vanishes in a purely linear Gaussian model.

\item \textbf{Perfect-tensor erasure.}  On a small toy network of isometries, apply the greedy reconstruction rule to two different surviving boundary regions.  Identify a bulk operator with two distinct boundary representatives and explain why they agree on the code subspace.

\item \textbf{Evidence audit.}  Design a five-stage experimental programme for a nonlinear hyperbolic chip.  For each stage, state a falsifiable claim, a measured observable, a calibration, and a failure mode that the preceding stage would not detect.
\end{enumerate}

\paragraph{Checks.}  In Problem 3, the antisymmetric end-site mode is dark at the centre.  In Problem 4, the alternating vector closes consistently only on an even cycle.  In Problem 8, the hyperbolic boundary-to-area ratio approaches $1/\ell$, whereas the Euclidean ratio vanishes with radius.

\clearpage
\section*{\texorpdfstring{\(\boldsymbol{\infty}\)\quad Revisiting Quantum Circuits}{Infinity: Revisiting Quantum Circuits}}
\addcontentsline{toc}{section}{\texorpdfstring{\(\infty\)\quad Revisiting Quantum Circuits}{Infinity: Revisiting Quantum Circuits}}

The four lectures were presented in an order, but their relation is not a one-way hierarchy.  Thinking, drawing, correcting, and geometrizing are four ways of interrogating the same composable process.  Each exposes information that the others may suppress, and each can force a revision of the description that came before it.

\begin{keyidea}[The four verbs revisited]
Thinking identifies the quantum process and the information it computes.  Drawing makes composition and equality local.  Correcting separates local evidence of an error from its logical action.  Geometrizing compiles global code structure into adjacency, interactions, boundaries, and experimental observables.
\end{keyidea}

The point of rethinking quantum circuits is therefore not to abandon the familiar gate model.  It is to let the circuit acquire enough structure to connect an algorithmic promise, a graphical proof, a topological code, and a physical device without confusing one level for another.  The distinctions among these levels matter just as much as the bridges between them.

\subsection*{One circuit, four descriptions}

At the first level, a circuit is specified by its inputs, controls and measurements, and intended task.  This keeps the mathematical object tied to a question: a unitary matrix with no declared input convention, measurement, or promise is not yet an algorithm.

The diagrammatic level makes composition visible.  Wires record types, boxes record processes, and connected boundaries record which outputs become which inputs.  Cups, caps, and spiders permit local calculations even when their matrix expressions are global.  Their use depends on semantics: a deformation or rewrite must preserve the represented map, channel, probability, or other declared quantity.

Error correction introduces a second distinction.  The physical process and the logical process are not the same object.  Syndrome data reveal something about an error without, in the correctable regime, revealing the encoded state.  The Knill--Laflamme condition, folded diagrams, footprint fibres, and homological classes are different ways of making this separation exact.  They also show why adjacency alone does not define a code and why a measured defect is not automatically a logical failure.

Geometry asks how the types, checks, logical representatives, ports, and interactions are to be realized.  A cellulation may organize stabilizers and logical cycles, while a weighted circuit graph organizes normal modes and control pathways.  A hyperbolic graph may support a useful code, a useful simulator, both, or neither; the role is determined by the degrees of freedom and observables.

The resulting workflow is a loop.  Hardware constraints can force a change in a coupling graph or syndrome schedule.  That change can alter the code geometry, which can require a new diagrammatic factorization, which may reveal that the original process specification was incomplete.  Conversely, a graphical identity can expose a lower-weight check or a different boundary partition and thereby suggest a more realistic device layout.  Moving among the four descriptions is therefore part of the analysis, not a presentation step added after the mathematics is finished.

\subsection*{What the summer school changed}

The questions following the lectures repeatedly crossed the boundaries between these descriptions.  Questions about rotating cups and caps led back to the equivalence of differently typed multilinear maps.  Questions about folding error diagrams led back to the precise content of Knill--Laflamme.  Questions about cycles led from pictures of error strings to homology and logical operators.  Questions about superconducting lattices required a sharper separation between evidence for a graph spectrum, evidence for quantum state control, and evidence for an encoded logical process.

Those interactions made clear that the transitions between the four verbs carry much of the real content.  A lecture can cross such a transition quickly through speech, gesture, and a sequence of images.  Written notes cannot rely on timing or gesture.  They must state the types of the maps, the hypotheses behind a rewrite, the information contained in a syndrome, and the experimental claim supported by an observable.  This is why the present document is an expansion rather than a transcript of the four lectures.

The questions also changed the emphasis.  More space has been given to preparations, measurements, and channels alongside unitaries; to the linear algebra beneath leg bending and spider diagrams; to the comparison between algebraic and diagrammatic correction criteria; and to the hierarchy of evidence needed when an abstract geometry is compiled into hardware.  These additions are not digressions from quantum circuits.  They identify what must be supplied before a circuit description can be carried faithfully from one mathematical or physical setting to another.

\subsection*{Where to go from here}

On the interface between thinking and drawing, the next task is to treat adaptive circuits as naturally as closed unitary ones.  Mid-circuit measurements, classical outcomes, feed-forward, resets, noise channels, and parameterized controls should have diagrammatic descriptions whose semantics remain explicit.  This is important both for proofs and for verified compilation: an automated rewrite is valuable only when the quantity it preserves is known.

On the interface between drawing and correcting, one would like local graphical methods that scale from exact code identities to complete fault histories.  Measurement outcomes, gauge choices, leakage, correlated noise, and time ordering all have to coexist with the compositional advantages of a diagram.  Symbolic rewriting, stabilizer algebra, homology, and numerical channel calculations should be used together, with each method checking the regime in which the others apply.

On the interface between correcting and geometrizing, code design and device design should be allowed to constrain one another from the start.  Check weight, logical distance, routing, readout ports, control schedules, fabrication disorder, and decoder latency belong to one iterative specification.  For hyperbolic superconducting architectures, a responsible progression begins with calibrated spectral agreement, then moves to few-photon control and nonlinear response, then to repeated syndrome extraction, and only then to logical benchmarking.  Passing one level is evidence for attempting the next; it is not a substitute for it.

A reader can therefore revisit any circuit by asking four questions.  What process and measurements are claimed?  Which equalities can be made local?  Which physical distinctions are erased by encoding or become logical?  Which geometry and device data are needed to realize and test the proposal?  The goal is to compare the task, proof, code, and device without identifying them.  These questions are the invitation to keep revisiting quantum circuits.

\clearpage
\phantomsection


\addcontentsline{toc}{section}{References}
\begin{thebibliography}{99}
\begin{sloppypar}

\bibitem{atiyah}
M. F. Atiyah,
``Topological quantum field theories,''
\emph{Publications Math\'ematiques de l'IH\'ES} \textbf{68} (1988), 175--186.
\href{https://doi.org/10.1007/BF02698547}{doi:10.1007/BF02698547}.

\bibitem{backens}
M. Backens,
``The ZX-calculus is complete for stabilizer quantum mechanics,''
\emph{New Journal of Physics} \textbf{16} (2014), 093021.
\href{https://doi.org/10.1088/1367-2630/16/9/093021}{doi:10.1088/1367-2630/16/9/093021}.

\bibitem{bennett-teleportation}
C. H. Bennett, G. Brassard, C. Cr\'epeau, R. Jozsa, A. Peres, and W. K. Wootters,
``Teleporting an unknown quantum state via dual classical and Einstein--Podolsky--Rosen channels,''
\emph{Physical Review Letters} \textbf{70} (1993), 1895--1899.
\href{https://doi.org/10.1103/PhysRevLett.70.1895}{doi:10.1103/PhysRevLett.70.1895}.

\bibitem{bienias}
P. Bienias, I. Boettcher, R. Belyansky, A. J. Koll\'ar, and A. V. Gorshkov,
``Circuit Quantum Electrodynamics in Hyperbolic Space: From Photon Bound States to Frustrated Spin Models,''
\emph{Physical Review Letters} \textbf{128} (2022), 013601.
\href{https://doi.org/10.1103/PhysRevLett.128.013601}{doi:10.1103/PhysRevLett.128.013601}.

\bibitem{blais-cqed}
A. Blais, A. L. Grimsmo, S. M. Girvin, and A. Wallraff,
``Circuit quantum electrodynamics,''
\emph{Reviews of Modern Physics} \textbf{93} (2021), 025005.
\href{https://doi.org/10.1103/RevModPhys.93.025005}{doi:10.1103/RevModPhys.93.025005}.

\bibitem{boettcher}
I. Boettcher, P. Bienias, R. Belyansky, A. J. Koll\'ar, and A. V. Gorshkov,
``Quantum Simulation of Hyperbolic Space with Circuit Quantum Electrodynamics: From Graphs to Geometry,''
\emph{Physical Review A} \textbf{102} (2020), 032208.
\href{https://doi.org/10.1103/PhysRevA.102.032208}{doi:10.1103/PhysRevA.102.032208}.

\bibitem{choi}
M.-D. Choi,
``Completely positive linear maps on complex matrices,''
\emph{Linear Algebra and its Applications} \textbf{10} (1975), 285--290.
\href{https://doi.org/10.1016/0024-3795(75)90075-0}{doi:10.1016/0024-3795(75)90075-0}.

\bibitem{coecke-duncan}
B. Coecke and R. Duncan,
``Interacting Quantum Observables: Categorical Algebra and Diagrammatics,''
\emph{New Journal of Physics} \textbf{13} (2011), 043016.
\href{https://doi.org/10.1088/1367-2630/13/4/043016}{doi:10.1088/1367-2630/13/4/043016}.

\bibitem{dennis-et-al}
E. Dennis, A. Kitaev, A. Landahl, and J. Preskill,
``Topological quantum memory,''
\emph{Journal of Mathematical Physics} \textbf{43} (2002), 4452--4505.
\href{https://doi.org/10.1063/1.1499754}{doi:10.1063/1.1499754}.

\bibitem{deutsch}
D. Deutsch,
``Quantum theory, the Church--Turing principle and the universal quantum computer,''
\emph{Proceedings of the Royal Society of London A} \textbf{400} (1985), 97--117.
\href{https://doi.org/10.1098/rspa.1985.0070}{doi:10.1098/rspa.1985.0070}.

\bibitem{dey-et-al}
S. Dey, A. Chen, P. Basteiro, A. Fritzsche, M. Greiter, M. Kaminski,
P. M. Lenggenhager, R. Meyer, R. Sorbello, A. Stegmaier, R. Thomale,
J. Erdmenger, and I. Boettcher,
``Simulating Holographic Conformal Field Theories on Hyperbolic Lattices,''
\emph{Physical Review Letters} \textbf{133} (2024), 061603.
\href{https://doi.org/10.1103/PhysRevLett.133.061603}{doi:10.1103/PhysRevLett.133.061603}.

\bibitem{gottesman}
D. Gottesman,
\emph{Stabilizer Codes and Quantum Error Correction},
Ph.D. thesis, California Institute of Technology, 1997.
\href{https://arxiv.org/abs/quant-ph/9705052}{arXiv:quant-ph/9705052}.

\bibitem{grover}
L. K. Grover,
``A fast quantum mechanical algorithm for database search,''
in \emph{Proceedings of the Twenty-Eighth Annual ACM Symposium on Theory of Computing} (STOC '96), 1996, 212--219.
\href{https://doi.org/10.1145/237814.237866}{doi:10.1145/237814.237866}.

\bibitem{hadzihasanovic-ng-wang}
A. Hadzihasanovic, K. F. Ng, and Q. Wang,
``Two complete axiomatisations of pure-state qubit quantum computing,''
in \emph{Proceedings of the 33rd Annual ACM/IEEE Symposium on Logic in Computer Science} (LICS 2018), 2018, 502--511.
\href{https://doi.org/10.1145/3209108.3209128}{doi:10.1145/3209108.3209128}.

\bibitem{hatcher}
A. Hatcher,
\emph{Algebraic Topology},
Cambridge University Press, Cambridge, 2002.
\href{https://pi.math.cornell.edu/~hatcher/AT/AT.pdf}{Author's online edition}.

\bibitem{jeandel-perdrix-vilmart}
E. Jeandel, S. Perdrix, and R. Vilmart,
``A Complete Axiomatisation of the ZX-Calculus for Clifford+$T$ Quantum Mechanics,''
in \emph{Proceedings of the 33rd Annual ACM/IEEE Symposium on Logic in Computer Science} (LICS 2018), 2018, 559--568.
\href{https://doi.org/10.1145/3209108.3209131}{doi:10.1145/3209108.3209131}.

\bibitem{kissinger-wetering}
A. Kissinger and J. van de Wetering,
``PyZX: Large Scale Automated Diagrammatic Reasoning,''
\emph{Electronic Proceedings in Theoretical Computer Science} \textbf{318} (2020), 229--241.
\href{https://doi.org/10.4204/EPTCS.318.14}{doi:10.4204/EPTCS.318.14}.

\bibitem{knill-laflamme}
E. Knill and R. Laflamme,
``Theory of quantum error-correcting codes,''
\emph{Physical Review A} \textbf{55} (1997), 900--911.
\href{https://doi.org/10.1103/PhysRevA.55.900}{doi:10.1103/PhysRevA.55.900}.

\bibitem{kollar}
A. J. Koll\'ar, M. Fitzpatrick, and A. A. Houck,
``Hyperbolic Lattices in Circuit Quantum Electrodynamics,''
\emph{Nature} \textbf{571} (2019), 45--50.
\href{https://doi.org/10.1038/s41586-019-1348-3}{doi:10.1038/s41586-019-1348-3}.

\bibitem{mahmoud-hqec}
A. A. Mahmoud, K. M. Ali, and S. Rayan,
``Systematic Approach to Hyperbolic Quantum Error Correction Codes,''
\emph{Physical Review A} \textbf{113} (2026), 042426.
\href{https://doi.org/10.1103/95mp-w7kr}{doi:10.1103/95mp-w7kr};
\href{https://arxiv.org/abs/2504.07800}{arXiv:2504.07800}.

\bibitem{mahmoud-cluster}
A. A. Mahmoud, G. Tournaire, S. Bachmann, and S. Rayan,
``Hyperbolic Cluster States for Fault-Tolerant Measurement-Based Quantum Computing,''
preprint, 2026.
\href{https://arxiv.org/abs/2603.27004}{arXiv:2603.27004}.

\bibitem{nielsen-chuang}
M. A. Nielsen and I. L. Chuang,
\emph{Quantum Computation and Quantum Information}, 10th Anniversary Edition,
Cambridge University Press, Cambridge, 2010.
\href{https://doi.org/10.1017/CBO9780511976667}{doi:10.1017/CBO9780511976667}.

\bibitem{pastawski}
F. Pastawski, B. Yoshida, D. Harlow, and J. Preskill,
``Holographic quantum error-correcting codes: Toy models for the bulk/boundary correspondence,''
\emph{Journal of High Energy Physics} \textbf{2015}, no. 6 (2015), 149.
\href{https://doi.org/10.1007/JHEP06(2015)149}{doi:10.1007/JHEP06(2015)149}.

\bibitem{rayan-diagrammatic}
S. Rayan,
``A diagrammatic field theory of quantum error correction,''
preprint, 2026.
\href{https://arxiv.org/abs/2607.08911}{arXiv:2607.08911}.

\bibitem{selinger}
P. Selinger,
``A survey of graphical languages for monoidal categories,''
in B. Coecke (ed.), \emph{New Structures for Physics}, Lecture Notes in Physics \textbf{813}, Springer, Berlin, 2011, 289--355.
\href{https://doi.org/10.1007/978-3-642-12821-9_4}{doi:10.1007/978-3-642-12821-9\_4}.

\bibitem{shor-code}
P. W. Shor,
``Scheme for reducing decoherence in quantum computer memory,''
\emph{Physical Review A} \textbf{52} (1995), R2493--R2496.
\href{https://doi.org/10.1103/PhysRevA.52.R2493}{doi:10.1103/PhysRevA.52.R2493}.

\bibitem{steane-code}
A. M. Steane,
``Error Correcting Codes in Quantum Theory,''
\emph{Physical Review Letters} \textbf{77} (1996), 793--797.
\href{https://doi.org/10.1103/PhysRevLett.77.793}{doi:10.1103/PhysRevLett.77.793}.

\bibitem{stinespring}
W. F. Stinespring,
``Positive functions on $C^*$-algebras,''
\emph{Proceedings of the American Mathematical Society} \textbf{6} (1955), 211--216.
\href{https://doi.org/10.1090/S0002-9939-1955-0069403-4}{doi:10.1090/S0002-9939-1955-0069403-4}.

\bibitem{van-de-wetering}
J. van de Wetering,
``ZX-calculus for the working quantum computer scientist,''
preprint, 2020.
\href{https://arxiv.org/abs/2012.13966}{arXiv:2012.13966}.

\bibitem{watrous}
J. Watrous,
\emph{The Theory of Quantum Information},
Cambridge University Press, Cambridge, 2018.
\href{https://doi.org/10.1017/9781316848142}{doi:10.1017/9781316848142}.

\bibitem{wootters-zurek}
W. K. Wootters and W. H. Zurek,
``A single quantum cannot be cloned,''
\emph{Nature} \textbf{299} (1982), 802--803.
\href{https://doi.org/10.1038/299802a0}{doi:10.1038/299802a0}.

\bibitem{xu-et-al}
X. Xu, A. A. Mahmoud, N. Gorgichuk, R. Thomale, S. Rayan, and M. Mariantoni,
``A Scalable Superconducting Circuit Framework for Emulating Physics in Hyperbolic Space,''
preprint, 2025.
\href{https://arxiv.org/abs/2510.23827}{arXiv:2510.23827}.

\end{sloppypar}
\end{thebibliography}
\end{document}